\documentclass[preprint,authoryear,12pt]{elsarticle}

\usepackage{amsfonts,graphics,amssymb,amsmath,epsf,graphicx}
\usepackage{subfigure,amsthm,rotating}

\newtheorem{theorem}{Theorem}[section]
\newtheorem{proposition}{Proposition}

\journal{Applied Mathematical Modelling}

\begin{document}

\begin{frontmatter}



\title{The Compound Class of Linear Failure Rate-Power Series Distributions: Model, Properties and Applications}


\author{Eisa Mahmoudi\corref{cor1}}
\ead{emahmoudi@yazd.ac.ir}
\author{Ali Akbar Jafari}

\cortext[cor1]{Corresponding author}

\address{Department of Statistics, Yazd University,
P.O. Box 89175-741, Yazd, Iran}

\begin{abstract}
We introduce in this paper a new class of distributions which generalizes the linear failure rate (LFR) distribution and is obtained by compounding the LFR distribution and power series (PS) class of distributions. This new class of distributions is called the linear failure rate-power series (LFRPS) distributions and contains some new distributions such as linear failure rate geometric (LFRG) distribution, linear failure rate Poisson (LFRP) distribution, linear failure rate logarithmic (LFRL) distribution, linear failure rate binomial (LFRB) distribution and Raylight-power series (RPS) class of distributions. Some former works such as exponential-power series (EPS) class of distributions, exponential geometric (EG) distribution, exponential Poisson (EP) distribution and exponential logarithmic (EL) distribution are special cases of the new proposed model.

The ability of the LFRPS class of distributions is in covering five possible hazard rate function i.e., increasing, decreasing, upside-down bathtub (unimodal), bathtub and increasing-decreasing-increasing shaped. Several properties of the LFRPS distributions such as
moments, maximum likelihood estimation procedure via an EM-algorithm and inference for a large sample, are discussed in this paper. In order to show the flexibility and potentiality of the new class of distributions, the fitted results of the new class of distributions and some its submodels are compared using a real data set.
\end{abstract}


\end{frontmatter}



\section{Introduction}
In recent years, many distributions to model lifetime data have been introduced. The basic idea of introducing these models is that a lifetime of a system with $N$ (discrete random variable) components and the positive continuous random variable, say $X_i$  (the lifetime of $i$th component), can be denoted by the non-negative random variable $Y=\min(X_1,\dots,X_N)$ or $Y=\max(X_1,\dots,X_N)$, based on whether the components are series or parallel.\\
Some well-known lifetime distributions such as the exponential geometric (EG), exponential Poisson (EP), exponential logarithmic (EL), Weibull geometric (WG) and Weibull Poisson (WP) distributions introduced and studied by
\cite{ad-lo-98}, \cite{kus-07}, \cite{ta-re-08}, \cite{ba-de-co-11}, and  \cite{lu-sh-11},
 respectively.

Let $N$ be a discrete random variable having depends on the class of power series distributions with probability mass function
\begin{equation}\label{eq.N}
P\left(N=n\right)=\frac{a_n{\theta }^n}{C(\theta )},\ \ \ \ n=1,2,\dots,
\end{equation}
where $a_n\geq0$ depends only on $n$, $C\left(\theta \right)=\sum^\infty_{n=1}{a_n{\theta }^n}$ and $\theta \in (0, s)$ is chosen in a way such that $C\left(\theta \right)$ is finite and its first, second and third derivatives with respect to $\theta $ are defined and shown by $C'(.)$, $C''(.)$ and $C'''(.)$, respectively. For more details on the power series class of distributions, see \cite{noack-50}.
This family of distributions includes binomial, Poisson, geometric and logarithmic distributions
 \citep{jo-ke-ko-05}.\\
Some authors by combining the family of power series distributions with the well-known distributions, extended these distributions and proposed new distributions. For example; exponential-power series (EPS) distributions
\citep{ch-ga-09},
 Weibull-power series (WPS) distributions
 \citep{mo-ba-11},
  complementary exponential-power series (CEPS) distributions
 \citep{fl-bo-ca-11},
 generalized exponential-power series (GEPS) distributions
 \citep{ma-ja-12},
 extended Weibull-power series (EWPS) distributions
 \citep{si-bo-di-co-13}, 
  Birnbaum-Saunders power series (BSPS) distributions \citep{bo-si-co-12}, and
 exponentiated Weibull-Poisson distribution
  \citep{ma-se-13}.
In this paper, by combining a class of power series distributions and  a linear failure rate (LFR) distribution, we will propose a new class of lifetime distributions. The family includes as special cases, the EPS distributions
\citep{ch-ga-09}
which this family includes the lifetime distributions presented by
\cite{ad-lo-98}, \cite{ad-di-05}, \cite{kus-07}, and \cite{ta-re-08}.
This family also includes the extended linear failure rate (ELF) distribution which is introduced by
\cite{gh-ko-07}.
By putting $a=0$, a new compound class of distributions which is called the Rayligh-power series (RPS) distributions is produced.

We provide four motivations for the LFRPS class of distributions, which can be applied in some interesting
situations as follows:\\
(i) This new class of distributions due to the stochastic representation $Y=\min(X_1,\dots,X_N)$, can arises in series systems with identical components, where each component has the LFR distribution lifetime. This model appears in many industrial applications and biological organisms. (ii) The LFRPS class of distributions can be applied for modeling the time to relapse of cancer under the first-activation scheme. (iii) The time to the first failure can be appropriately modeled by the LFRPS class of distributions. (iv) The LFRPS class of distributions gives a reasonable
parametric fit to some modeling phenomenon with non-monotone failure rates such as the bathtub-shaped, unimodal and increasing-decreasing-increasing failure rates, which are common in reliability and biological studies.

The reminder of the paper is organized as follows: In Section \ref{se.class}, we define the LFRPS class of distributions and outline some special cases of the distribution. We investigate some properties of the distribution in this Section \ref{se.pro}. General expansions for the moments of the LFRPS distributions are given in this section. We focus on special cases of the LFRPS distributions in Section \ref{se.spe}. Maximum likelihood estimation and EM algorithm are discussed in Section \ref{se.mle}. A simulation study is performed in Section \ref{se.sim}. In Section \ref{se.app}, application of the LFRPS class of distributions is given using a set data set. Finally, Section \ref{se.con} concludes the paper.

\section{New class of linear failure rate}
\label{se.class}
The linear failure rate distribution with parameters $a\geq0$ and $b\geq0$ (such that $a+b>0$),
will be denoted by ${\rm LFR(}a,b)$, has the following cumulative distribution function (cdf)
\begin{equation}\label{eq.G}
G\left(x\right)=1-\exp  \left(-ax-\frac{b}{2}x^2\right),\ \ \ \ \ x\geq0,
\end{equation}
and the probability density function
\begin{equation}\label{eq.g}
{\rm g}\left(x\right)=\left(a+bx\right){\exp  \left(-ax-\frac{b}{2}x^2\right)},\ \ \ \ \ x\geq0.
\end{equation}

Note that if $b=0$ and $a\neq0$, then the exponential distribution with parameter $a$, ${\rm Exp}(a)$, and if $a=0$ and $b\neq0$ then we can obtain the Rayleigh distribution with parameter $b$, ${\rm Rayleigh}(b)$ is obtained. Also, The failure rate function of ${\rm LFR}(a,b)$ is constant if $b=0$ and is increasing if $b>0$ \citep{se-bh-95}.

Let $X_1,{\ X}_2,\ \dots ,X_N$ be a random sample from the LFR distribution with cdf in (\ref{eq.G}). Also, let $N$ be discrete random variable with the power series distributions with probability mass function in (\ref{eq.N}). If $X_{(1)}={\min  (X_1,\dots ,X_N)\ }$ then the conditional cdf of $X_{(1)}$  given $N=n$ is
\begin{equation}
G_{X_{(1)}|N=n}(x)=1-\exp(-anx-\frac{bn}{2}x^2),
\end{equation}
which has a LFR distribution with parameters $an$ and $bn$. Therefore, the cdf of the new class of LFR distributions is the marginal cdf of $X_{(1)}$, i.e.
\begin{equation}\label{eq.F}
F(x)=1-\frac{C\left(\theta {\exp  \left(-ax-\frac{b}{2}x^2\right)\ }\right)}{C(\theta )},\ \ x>0,
\end{equation}
and we denote it with ${\rm LFRPS}(a, b, \theta )$. The pdf of this class is
\begin{equation}\label{eq.f}
f(x)=\theta(a+bx)\exp (-ax-\frac{b}{2}x^2)\frac{C'\left(\theta\exp  \left(-ax-\frac{b}{2}x^2\right)\right)}{C(\theta )}.
\end{equation}
The survival function and the hazard rate function of the LFRPS class of distributions, are given, respectively by
\[S\left(x\right)=\frac{C\left(\theta {\exp  \left(-ax-\frac{b}{2}x^2\right)\ }\right)}{C(\theta )},\]
and
\begin{equation}\label{eq.h}
h\left(x\right)=\frac{\theta (a+bx){\exp  \left(-ax-\frac{b}{2}x^2\right)\ }C'\left(\theta {\exp  \left(-ax-\frac{b}{2}x^2\right)\ }\right)}{C\left(\theta {\exp  \left(-ax-\frac{b}{2}x^2\right)\ }\right)}.
\end{equation}
Consider $C\left(\theta \right)=\theta +{\theta }^{20}$. Then
$$
f\left(x\right)= (a+bx){\exp  \left(-ax-\frac{b}{2}x^2\right)\ }\frac{1+20{\theta }^{19}{\exp  \left(-19ax-\frac{19b}{2}x^2\right)\ }}{1+{\theta }^{19}},
$$
and
$$
h\left(x\right)=\frac{(a+bx)\left(1+20 \theta^{19}\exp  \left(-19ax-\frac{19b}{2}x^2\right)\right)}{1+\theta^{19}\exp  \left(-19ax-\frac{19b}{2}x^2\right)\ }.
$$

The plots of this density and its hazard rate function, for different values of parameters $a$, $b$ and $\theta$ are given in Fig. \ref{fig.bim}.
\begin{figure}
\centering
\includegraphics[scale=0.40]{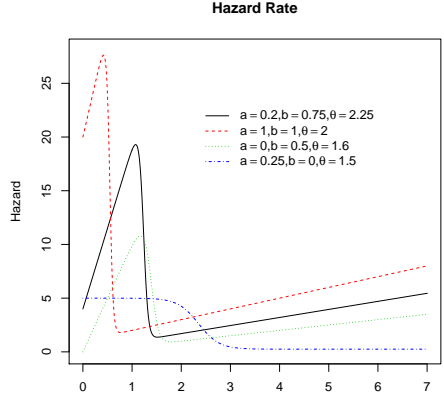}
\vspace{-0.8cm}
\caption{Plots of pdf and hazard rate function of LFRPS distributions with $C(\theta)=\theta+\theta^{20}$.  }\label{fig.bim}
\end{figure}

\begin{proposition} The limiting distribution of the ${\rm LFRPS}(a,\ b,\ \theta )$ cdf, when $\theta \to 0^+$ is
\[{\mathop{\lim}_{\theta \to 0^+} F\left(x\right)\ }=1-\exp  (-acx-\frac{bc}{2}x^2),\]
which is a LFR cdf with parameters $ac$ and $bc$, where $c={\min \left\{n\in {\mathbb N}:a_n>0\right\}}$.
\end{proposition}
\begin{proposition} The densities of LFRPS class can be expressed as infinite number of linear combination of density of order statistics. We know that $C'\left(\theta \right)=\sum^n_{n=1}{na_n{\theta }^{n-1}}$. Therefore,
\[f(x)=\sum^{\infty }_{n=1}P\left(N=n\right)g_{X_{(1)}}(x;n),\]
where $g_{X_{(1)}}(x;n)$ is the density function of $X_{(1)}=\min (X_1,\dots ,X_n)$, given by
\[g_{X_{(1)}}(x;n)=(an+bnx){\exp  \left(-anx-\frac{bn}{2}x^2\right)\ },\]
which is the density function of LFR distribution with parameters $an$ and $bn$.
\end{proposition}
\begin{proposition} For the pdf in (\ref{eq.f}), we have
\[\lim_{x\rightarrow 0^+}f(x)=\frac{\theta aC'(\theta )}{C(\theta )},\ \ \ \ \ \lim_{x\rightarrow \infty}f(x)=0\]
\end{proposition}
\begin{proposition} For the hazard rate function in (\ref{eq.h}), we have
\[
{\mathop{\lim }_{{\rm x}\to 0^{{\rm +}}} h\left(x\right)}=\frac{\theta a C'\left(\theta \right)}{C\left(\theta \right)},\ \ \ \ \ \ \ \ {\mathop{\lim }_{{\rm x}\to \infty } h\left(x\right) }=\infty,
\]
\end{proposition}
\begin{proposition} The quantile $\xi$ of the LFRPS class of distributions is given by
\[x_{\xi }=G^{-1}\left(1-\frac{1}{\theta }C^{-1}\left(\left(1-\xi \right)C\left(\theta \right)\right)\right),\]
where $C^{-1}\left(.\right)$ is the inverse function of $C\left(.\right)$ and $G^{-1}\left(.\right)$ is the inverse function of distribution function of LFR distribution, i.e.
\[G^{-1}\left(x\right)=\left\{
\begin{array}{lcl}
\frac{1}{b}(-a+\sqrt{a^2-2b{\log  (1-x)}}) &  & {\rm if}\ \ \ b>0, \\
-\frac{1}{a}{\log  \left(1-x\right)}  &  & {\rm  if } \ \ \ a>0, \ b=0. \end{array}
\right.\]
We can use this expression for generating a random data from LFRPS distributions with generating data from uniform distribution.
\end{proposition}
\begin{proposition}
If $X$ follows ${\rm LFRPS}(a, b,\theta )$ then $Y=X+\frac{ab}{2}X^2$ has the EPS distributions with parameters $a$ and $\theta$, with the following density function
\[f_Y\left(y\right)= \frac{a\theta C'\left(\theta {\exp  \left(ay\right)\ }\right)}{C(\theta )},\]
which is introduced by
\cite{ch-ga-09}.
\end{proposition}

\section{Properties}
\label{se.pro}
Now, we obtain the moment generating function of LFRPS distributions. Consider $X\sim {\rm LFRPS}(a,\ b,\ \theta )$. Then
\[M_X\left(t\right)=\sum^{\infty }_{n=1}{P\left(N=n\right)M_{X_{(1)}}\left(t\right)}.\]
But for $an-t>0$, we have
\begin{eqnarray*}
M_{X_{(1)}}\left(t\right)&=&\int^{\infty }_0{e^{tx}\left(an+bnx\right)e^{-anx-\frac{bn}{2}x^2}dx}\\
&=&\int^{\infty }_0{\left(an-t+bnx\right)e^{-(an-t)x-\frac{bn}{2}x^2}dx}+\int^{\infty }_0{te^{-(an-t)x-\frac{bn}{2}x^2}dx}\\
&=&1+t\int^{\infty }_0{e^{-(an-t)x-\frac{bn}{2}x^2}dx}.
\end{eqnarray*}
If $b=0$ then
\[M_{X_{(1)}}\left(t\right)=1+\frac{t}{an-t}=\frac{an}{an-t}.\]
If $b>0$ then
\begin{eqnarray*}
M_{X_{(1)}}\left(t\right)&=&1+t\int^{\infty }_0{e^{-(an-t)x-\frac{bn}{2}x^2}dx}\\
&=&1+t\int^{\infty }_0{e^{-\frac{bn}{2}\left[x^2+\frac{2(an-t)}{bn}x\right]}dx}\\
&=&1+t\int^{\infty }_0{e^{-\frac{bn}{2}\left[x^2+\frac{2(an-t)}{bn}x+\frac{{\left(an-t\right)}^2}{(bn)^2}\right]+\frac{{\left(an-t\right)}^2}{2bn}}dx}\\
&=&1+\frac{\sqrt{2\pi}}{\sqrt{bn}}te^{\frac{{\left(an-t\right)}^2}{2bn}}
\int^{\infty }_0{\frac{\sqrt{bn}}{\sqrt{2\pi}}e^{-\frac{bn}{2}{\left[x+\frac{2\left(an-t\right)}{bn}\right]}^2}dx}\\
&=&2-\frac{\sqrt{2\pi}}{\sqrt{bn}}te^{\frac{{\left(an-t\right)}^2}{2bn}}\Phi \left(\frac{2{\left(an-t\right)}}{\sqrt{bn}}\right),
\end{eqnarray*}
where $\Phi(.)$ is the distribution function of standard normal distribution.
Therefore
\begin{eqnarray*}
M_X\left(t\right)
&=&\left\{ \begin{array}{ll}
\sum^{\infty }_{n=1}{\frac{a_n{\theta }^n}{C\left(\theta \right)}\frac{an}{an-t}} & {\rm if}\ b=0,\ an-t>0 \\
2-\frac{t\sqrt{2\pi}}{\sqrt{b}}\sum^{\infty }_{n=1}{\frac{a_n{\theta }^n}{C\left(\theta \right)\sqrt{n}}e^{\frac{{\left(an-t\right)}^2}{2bn}}
\Phi \left(\frac{2{\left(an-t\right)}}{\sqrt{bn}}\right)} & {\rm if}\ b\ne 0 \\
\infty & {\rm if}\ an-t<0
 \end{array}
\right.
\end{eqnarray*}
We can use $M_X(t)$ to obtain the central moment functions, $E(X^k)$. But from the direct calculation, we have
\[E\left(X^k\right)=\int^{\infty }_0{x^kf\left(x\right)dx}=\sum^{\infty }_{n=1}{P\left(N=n\right)E\left(X^k_{\left(1\right)}\right)}.\]

If $\ b=0$
\[E\left(X^k_{\left(1\right)}\right)=\int^{\infty }_0{x^kane^{-anx}dx}=\frac{k!}{{\left(an\right)}^k}.\]

If $a=0$, then
\[E\left(X^k_{\left(1\right)}\right)=\int^{\infty }_0{x^{k+1}bne^{-\frac{bn}{2}x^2}dx}=\int^{\infty }_0{\frac{bn}{2}y^{\frac{k}{2}}e^{-\frac{bn}{2}y}dy}=\frac{\Gamma \left(\frac{k}{2}+1\right)}{{\left(\frac{bn}{2}\right)}^{\frac{k}{2}+1}}=2^{\frac{k}{2}}\frac{\Gamma \left(\frac{k}{2}+1\right)}{{(bn)}^{\frac{k}{2}+1}}.\]

If $a\ge 0$, $b>0$, from
\cite{na-mi-05},
we have
\[E\left(X^k_{\left(1\right)}\right)=\frac{k}{\sqrt{2bn}}e^{\frac{na^2}{2b}}\sum^{k-1}_{l=0}
{\left(\begin{array}{c}
  k-1 \\
  1
\end{array}\right){\left(-\frac{a}{b}\right)}^l{\left(\frac{2}{bn}\right)}^{\frac{k-1-l}{2}}\Gamma \left(\frac{k-l}{2},\frac{na^2}{2b}\right)}.\]
where $\Gamma \left(\alpha ,x\right)=\int^{\infty }_x{t^{\alpha -1}e^{-t}dt}$  denotes the upper incomplete gamma function.

\section{Special cases of the LFRPS distributions}
\label{se.spe}
In this section, we consider special cases of the LFRPS distributions.

\subsection{Linear failure rate distribution}
The LFR distribution is a special case of the LFRPS distributions with  $C\left(\theta \right)=\theta$, $a_1=1$, and $a_n=0$ for $n>1$. Then the density function in (\ref{eq.f}) becomes the density function of the LFR distribution. Note that the hazard function of LFR distribution is either constant or increasing.

\subsection{ Exponential power series distribution}

If $b=0$, then the density function in (\ref{eq.f}) changes to the exponential power series (EPS) density function which is introduced by
\cite{ch-ga-09}
and has the following density function
\[f\left(x\right)=\theta a{\exp  \left(-ax\right)\ }\frac{C'\left(\theta {\exp  \left(-ax\right)\ }\right)}{C(\theta )}.\]
This distribution contains several distributions as special cases: exponential geometric distribution
 \citep{ad-lo-98,ad-di-05},
exponential Poisson distribution
 \citep{kus-07},
and exponential logarithmic distribution
  \citep{ta-re-08}.

\subsection{ Rayleigh-power series distribution}

If $a=0$, then the LFRPS distributions gives a new class for Rayleigh distribution with the following density
\[f\left(x\right)=\theta bx \ {\exp  (-\frac{b}{2}x^2)\ }\frac{C'\left(\theta {\exp  \left(-\frac{b}{2}x^2\right)\ }\right)}{C(\theta )},\]
which we called Rayleigh power series (RPS) distribution. Note that if $X$ follows the RPS distributions then $X^2$ has a EPS distributions. Also, RPS distributions is a special case of the WPS distributions
\citep{mo-ba-11}
and contains Rayleigh geometric, Rayleigh Poisson, Rayleigh binomial and Rayleigh logarithmic distributions as special cases.
\subsection{LFR-geometric distribution}
The geometric distribution (truncated at zero) is a special case of power series distributions with $a_n=1$ and $C\left(\theta \right)=\frac{\theta }{1-\theta }$ ($0<\theta <1$). The density of the LFRG distribution is given by
\begin{equation}\label{eq.fLG}
f\left(x\right)=\frac{\left(1-\theta \right)(a+bx){\exp  \left(-ax-\frac{b}{2}x^2\right)\ }}{{\left(1-\theta {\exp  \left(-ax-\frac{b}{2}x^2\right)\ }\right)}^2}=\frac{\alpha (a+bx){\exp  \left(-ax-\frac{b}{2}x^2\right)\ }}{{\left(1-(1-\alpha ){\exp  \left(-ax-\frac{b}{2}x^2\right)\ }\right)}^2},
\end{equation}
where $\theta =1-\alpha $. The hazard rate function is given by
\[h\left(x\right)=\frac{(a+bx)}{1-\theta {\exp  \left(-ax-\frac{b}{2}x^2\right)\ }}.\]
If $b=0$ the density function in (\ref{eq.fLG}) becomes the exponential geometric (EG) density function \citep{ad-lo-98}.
 The hazard rate function of the EG distribution is decreasing. \cite{ad-di-05}
 extended the EG distribution by putting $\alpha =1-\theta >0$, and introduced the extended exponential geometric distribution. the EEG hazard rate function is monotonically increasing for $\alpha >1$; decreasing for $0<\alpha <1$ and constant for $\alpha =1$.
\cite{gh-ko-07}
introduced the LFRG distribution based on \cite{ma-ol-97}
and in fact, the function in (\ref{eq.fLG}) is also density function if  $\alpha =1-\theta >0$. This distribution is known to extended linear failure rate (ELFR) distribution, and its density function is decreasing if $\left(\alpha -2\right)a^2+\alpha b \le 0$, and is unimodal if $\left(\alpha -2\right)a^2+\alpha b >0$. \cite{gh-ko-07}
showed that the hazard rate function of the ELFR distribution is increasing, bathtub or increasing-decreasing-increasing.

\begin{figure}[]
\centering
\includegraphics[scale=0.40]{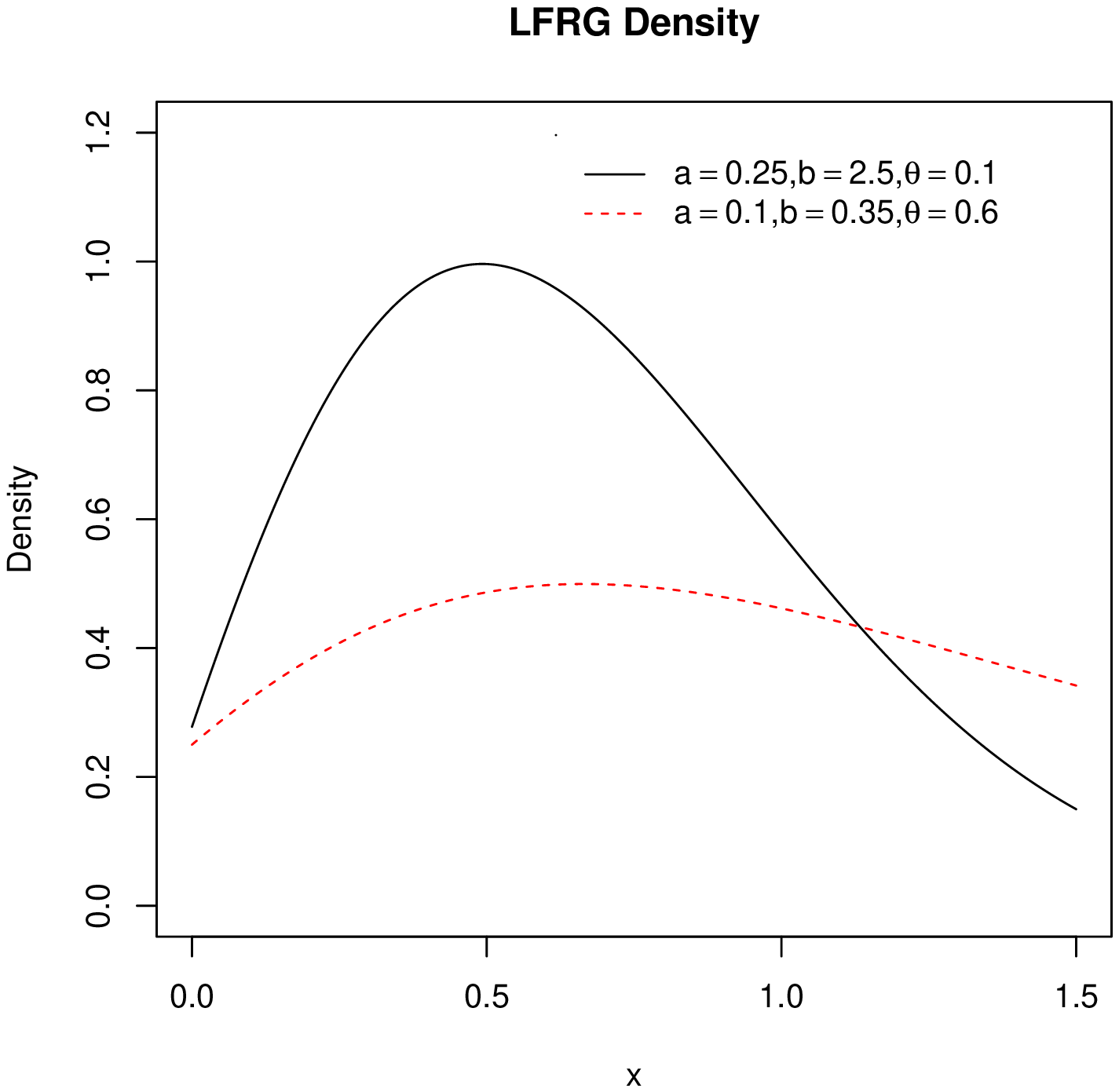}
\includegraphics[scale=0.40]{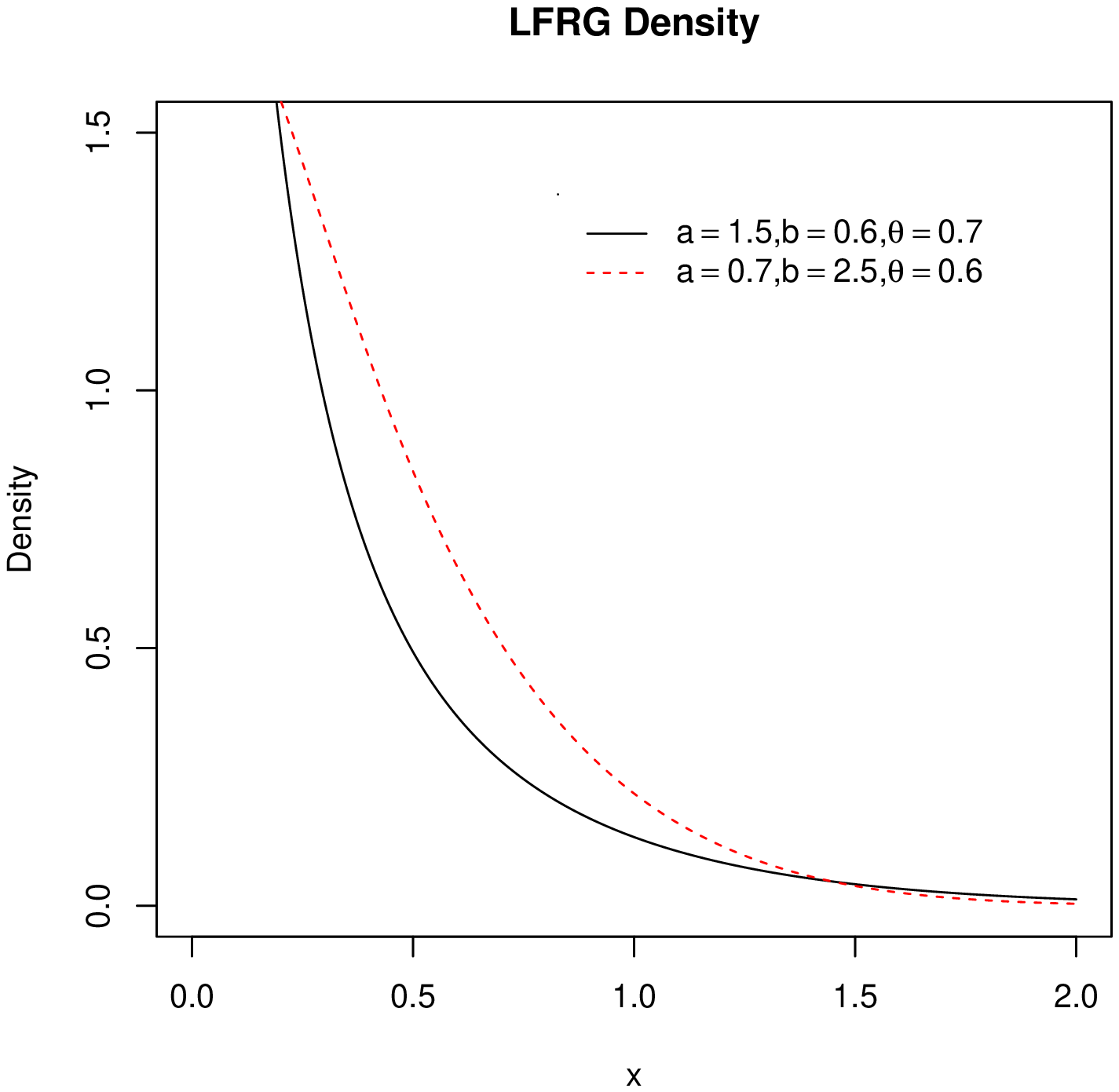}
\includegraphics[scale=0.40]{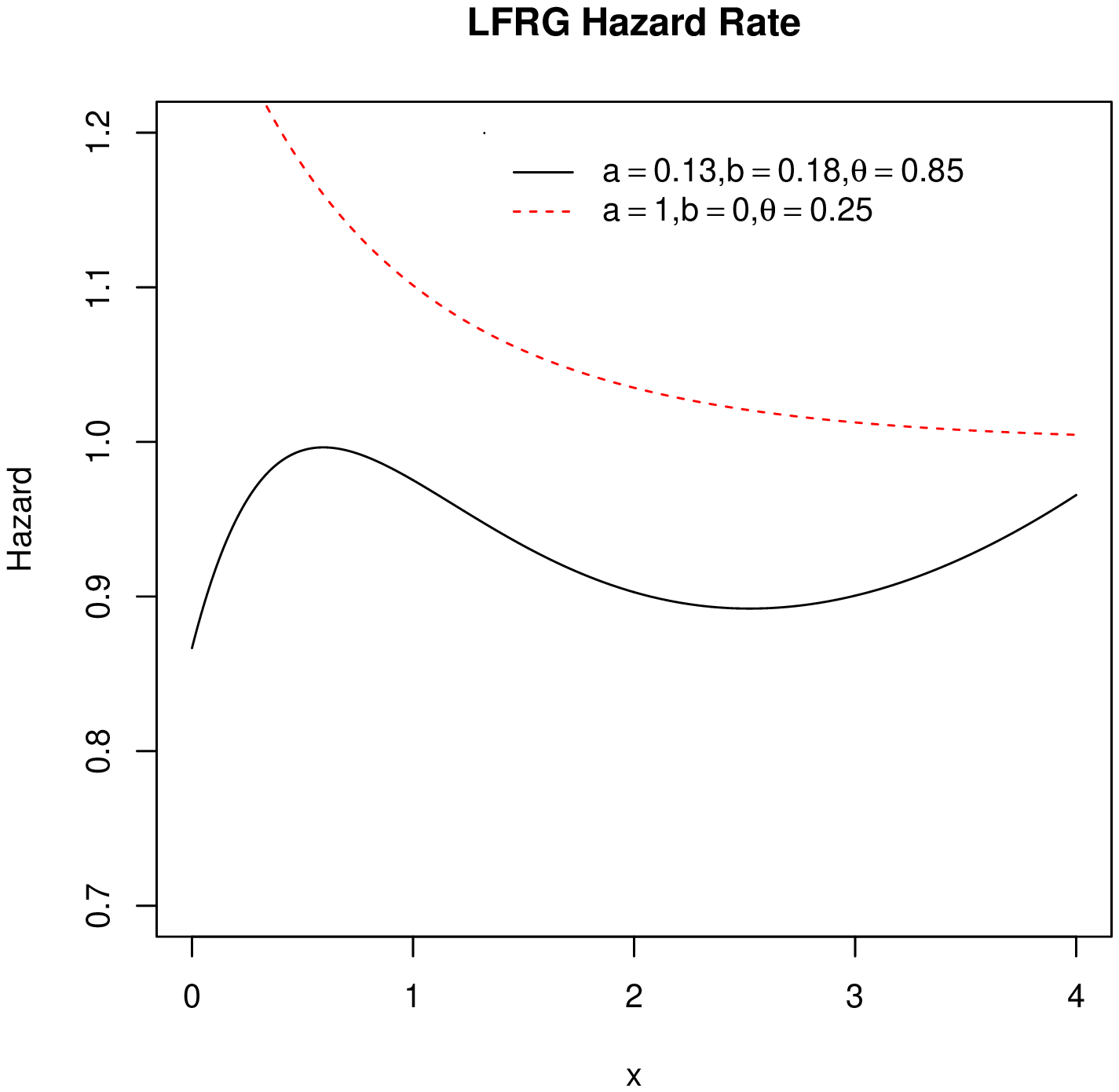}
\includegraphics[scale=0.40]{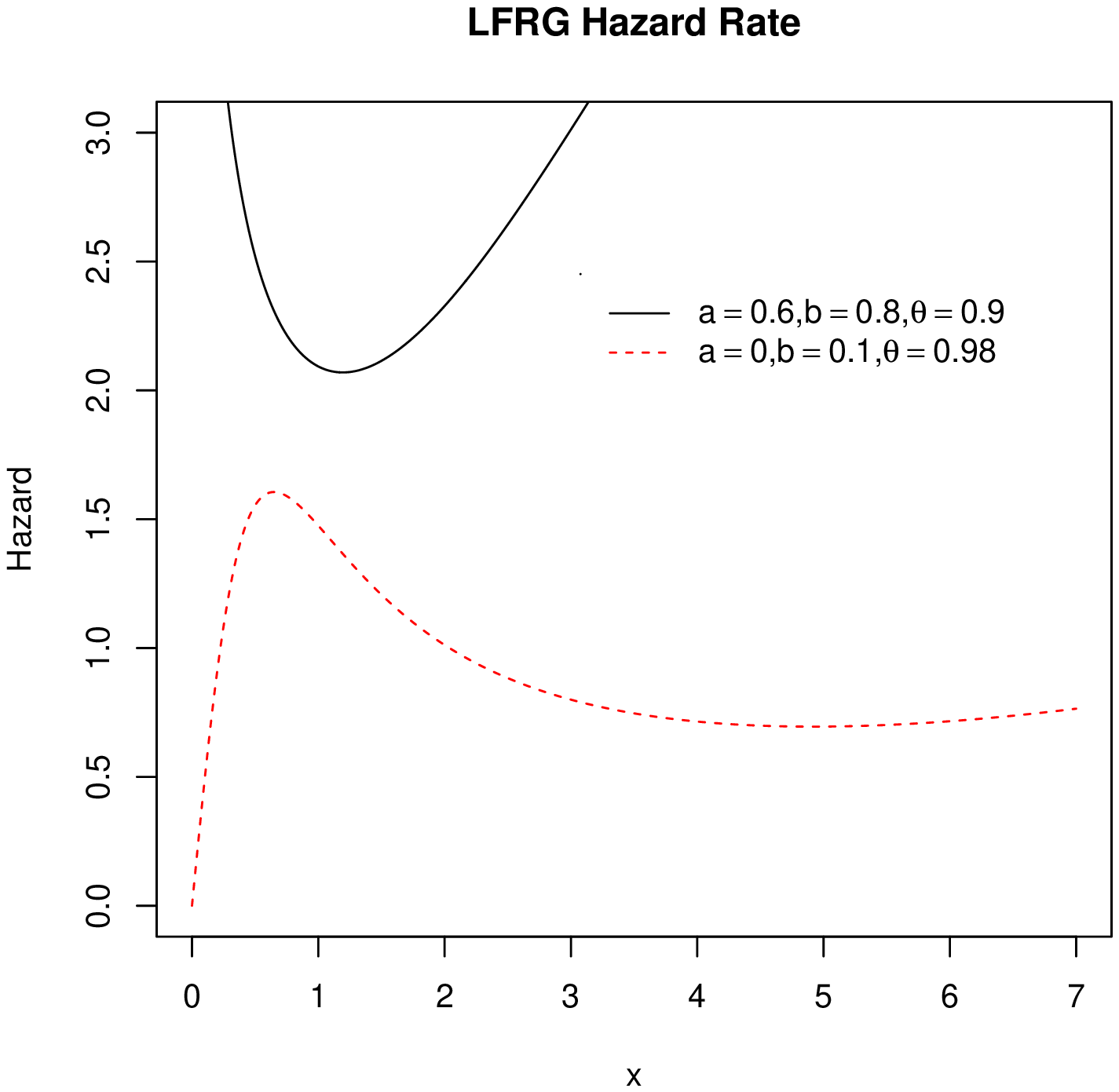}
\includegraphics[scale=0.40]{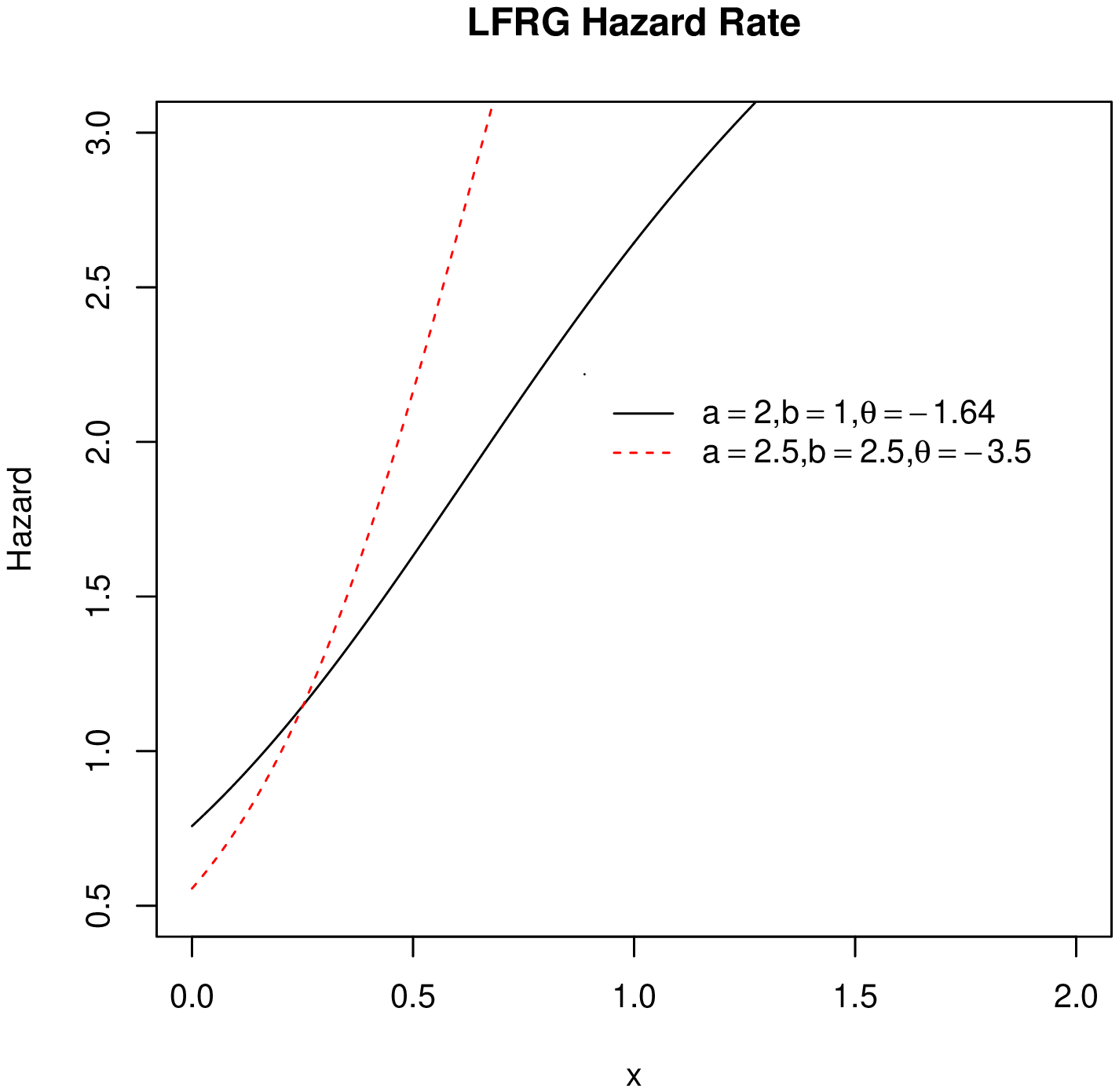}
\vspace{-0.8cm}
\caption[]{Plot of density and hazard rate function of the LFRG distribution for different values of its parameters.}
\end{figure}

\subsection{LFR-Poisson distribution}

The Poisson distribution (truncated at zero) is a special case of power series distributions with $a_n=\frac{1}{n!}$ and $C\left(\theta \right)=e^{\theta }-1$ ($\theta >0$). The density function of LFR-Poisson (LFRP) distribution is given as
\begin{equation}\label{eq.fLP}
f\left(x\right)=\theta (a+bx){\exp  \left(-ax-\frac{b}{2}x^2\right)\ }\frac{e^{\theta {\exp  \left(-ax-\frac{b}{2}x^2\right)\ }}}{e^{\theta }-1},
\end{equation}
and the hazard rate function of LFRP distribution is given by
\[h\left(x\right)=\theta (a+bx){\exp  \left(-ax-\frac{b}{2}x^2\right)\ }\frac{e^{\theta {\exp  \left(-ax-\frac{b}{2}x^2\right)\ }}}{e^{\theta \exp  \left(-ax-\frac{b}{2}x^2\right)}-1}.
\]
If $b=0$, the density function in (\ref{eq.fLP}) changes to the density of exponential-Poisson (EP) distribution \citep{kus-07}.
 The hazard rate function of the EP distribution is decreasing, increasing, bathtub and increasing-decreasing-increasing.

\begin{figure}[]
\centering
\includegraphics[scale=0.40]{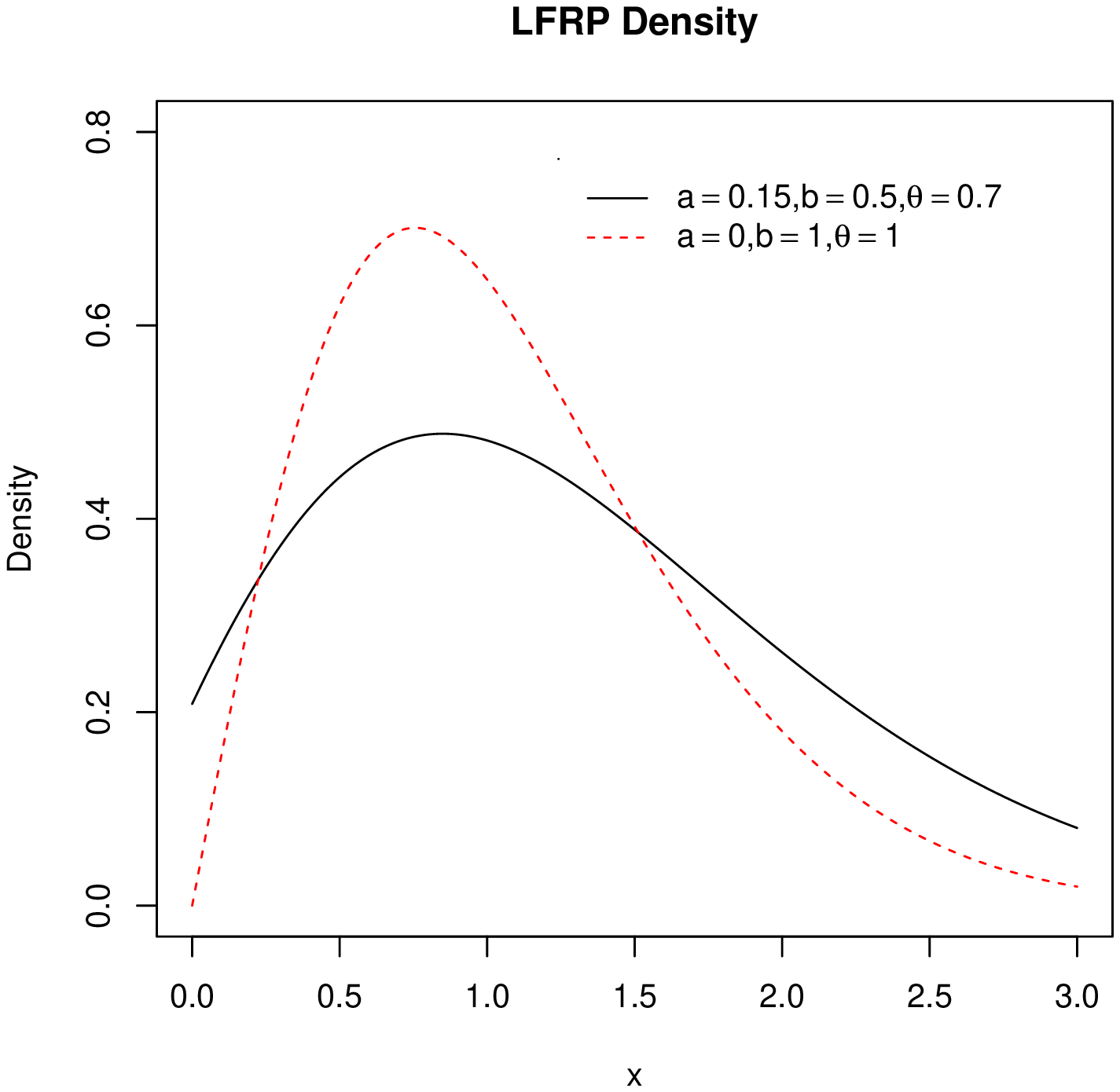}
\includegraphics[scale=0.40]{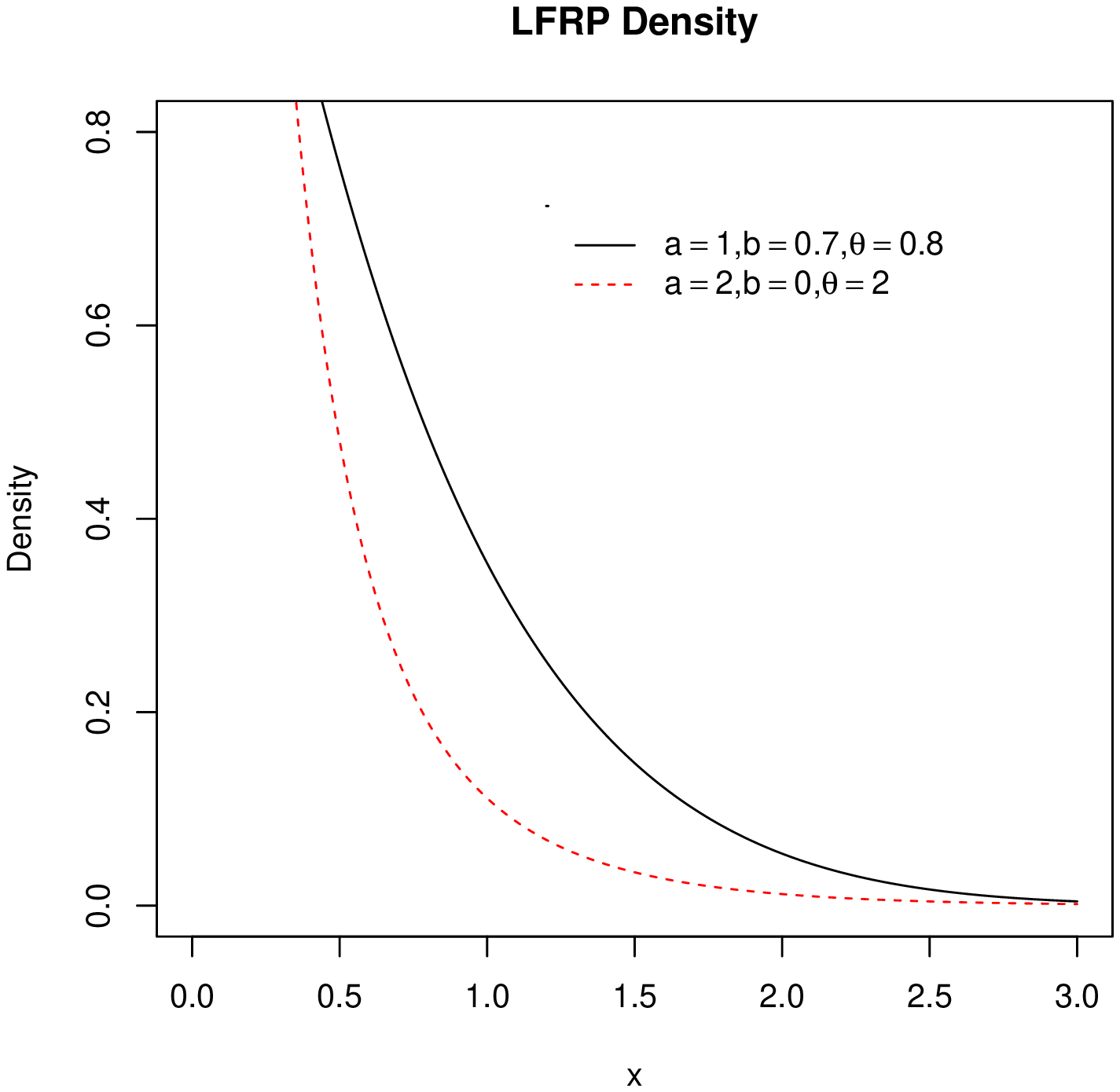}
\includegraphics[scale=0.40]{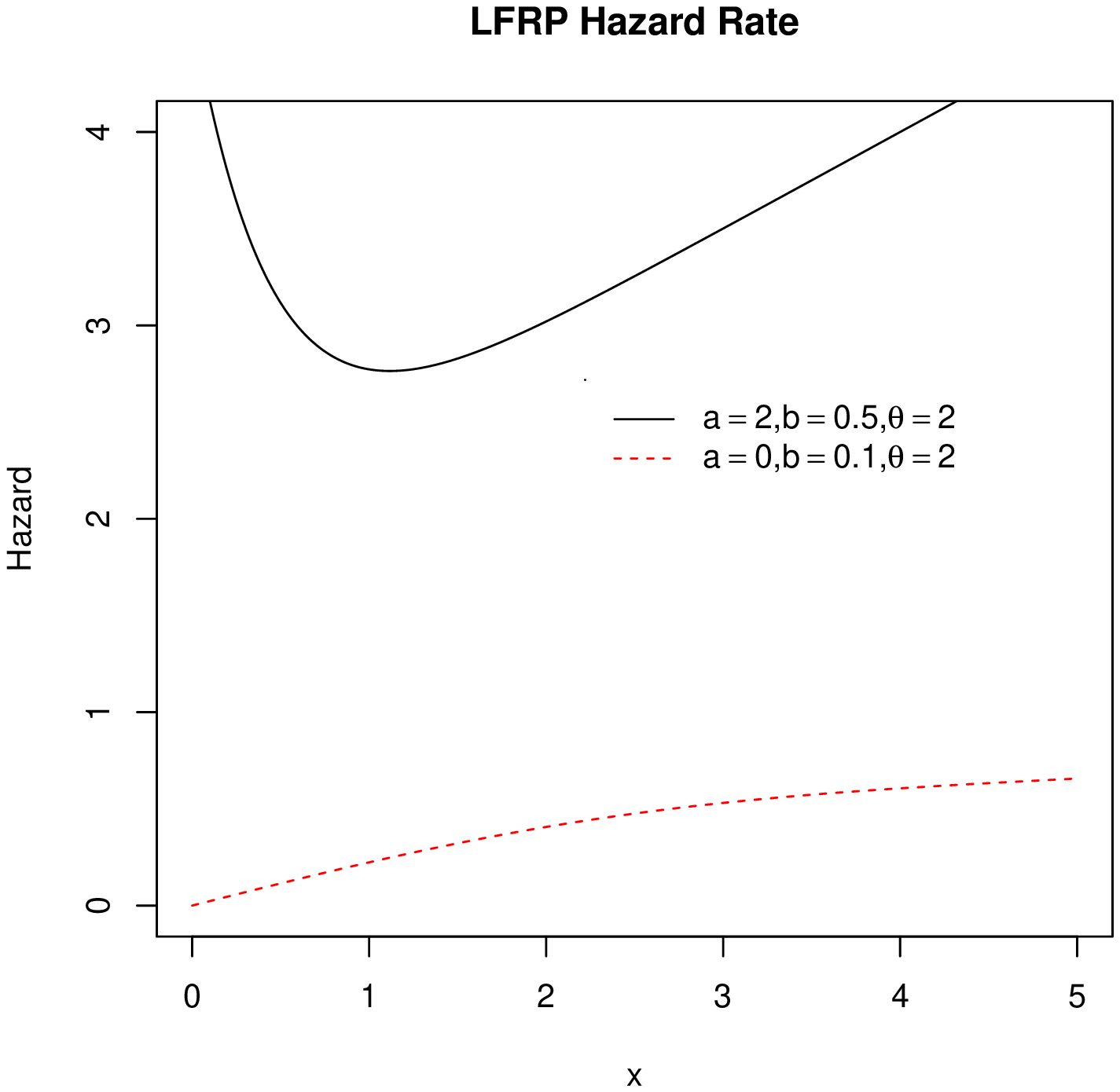}
\includegraphics[scale=0.40]{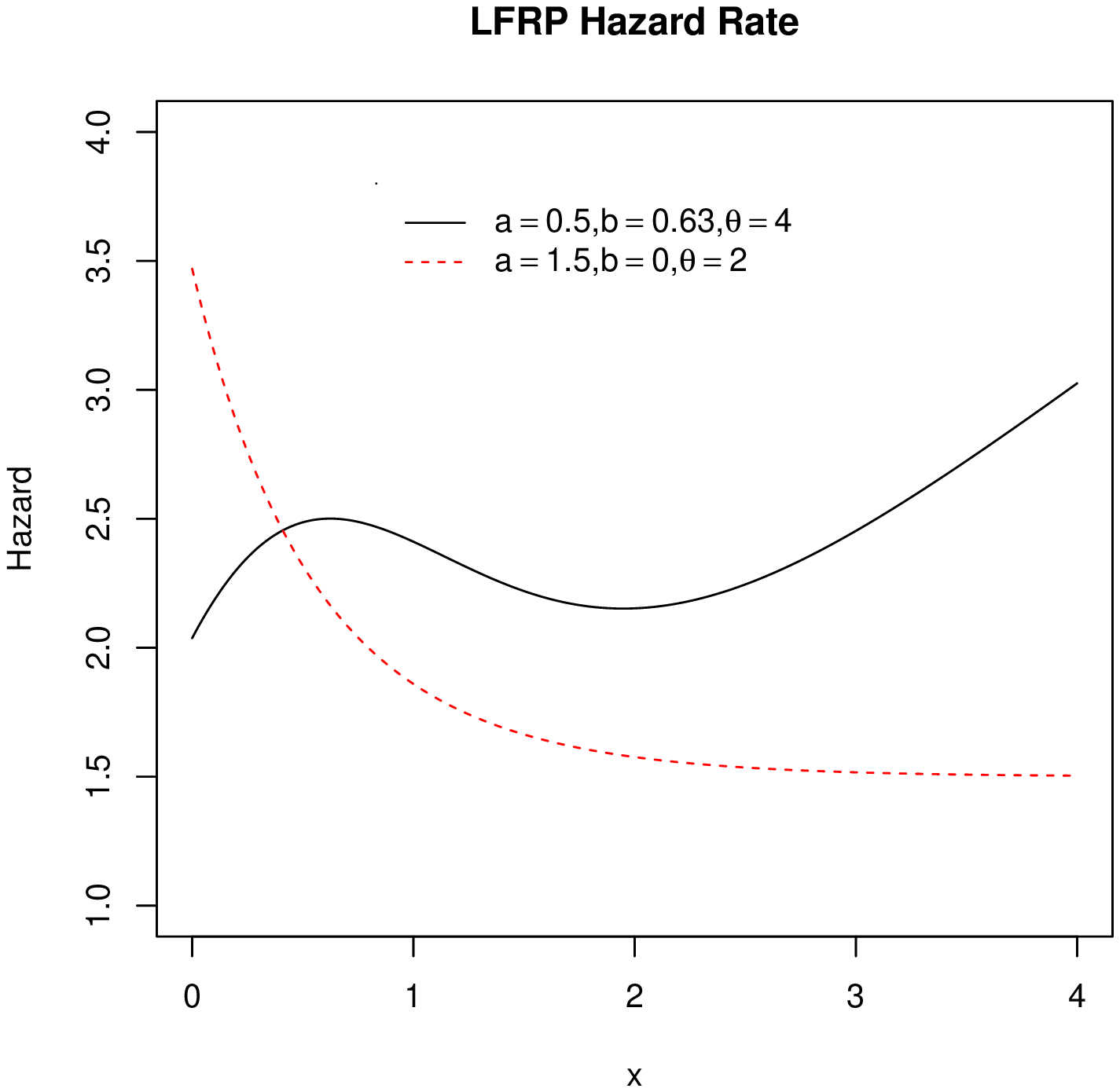}
\includegraphics[scale=0.40]{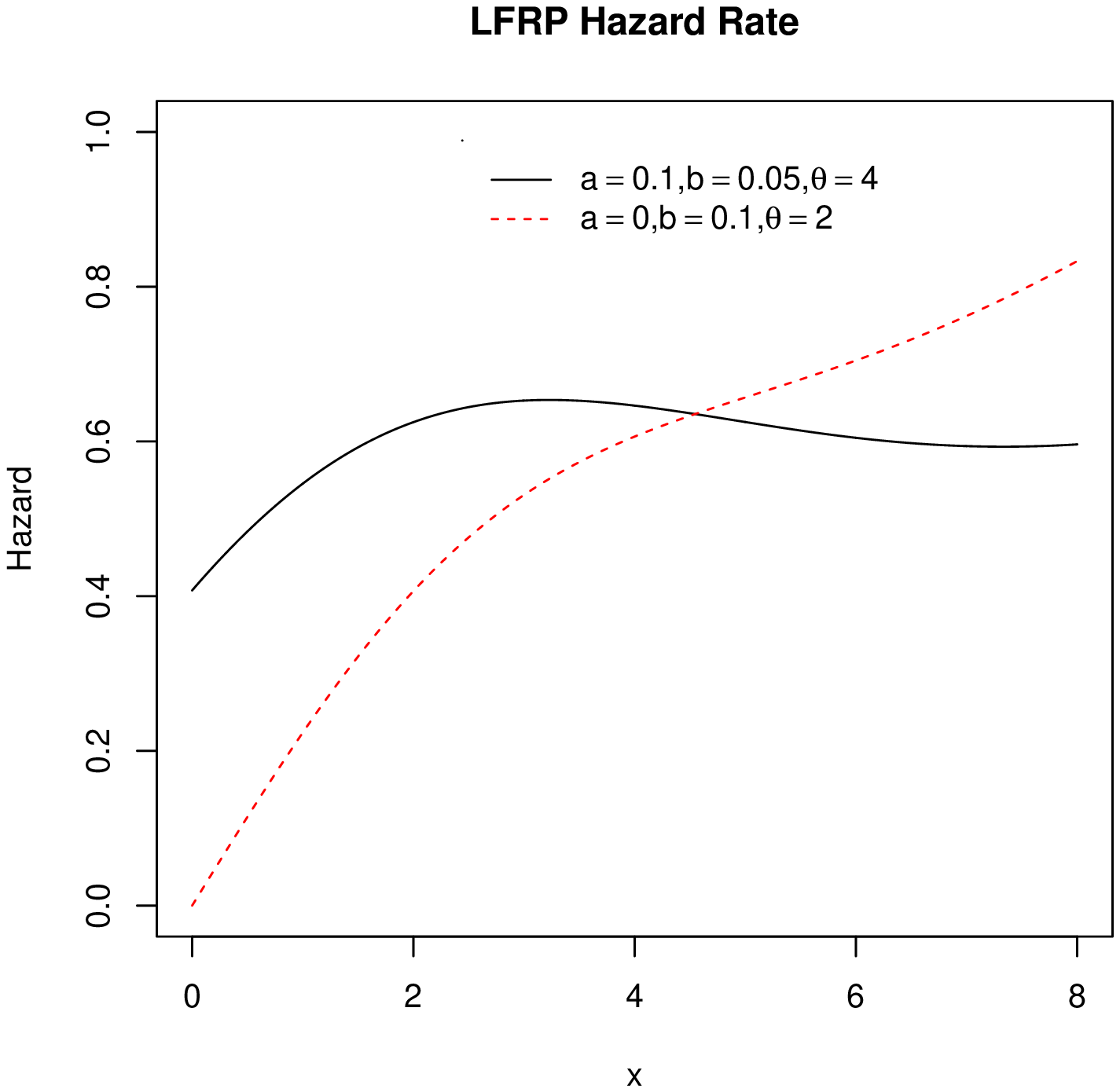}
\vspace{-0.8cm}
\caption[]{Plot of density and hazard rate function of the LFRP distribution for different values of its parameters.  }
\end{figure}

\subsection{LFR-binomial distribution}

The binomial distribution (truncated at zero) is also a special case of the class of power series distributions with
$a_n=\left(\begin{array}{c}
  m \\
  n \end{array}\right)$ and $C\left(\theta \right)={\left(\theta +1\right)}^m-1$ where $m$ ($n\le m$)  is the number of replicas. The density function of LFR-binomial (LFRB) distribution is given by
\begin{equation}\label{eq.fLB}
f\left(x\right)=\theta m(a+bx){\exp  \left(-ax-\frac{b}{2}x^2\right)}\frac{{\left(\theta {\exp  \left(-ax-\frac{b}{2}x^2\right)\ }+1\right)}^{m-1}}{{\left(\theta +1\right)}^m-1},
\end{equation}
and its hazard rate function is given as
\[h\left(x\right)=\frac {m \theta (a+bx)\exp \left(-ax-\frac{b}{2}x^2\right)(\theta \exp \left(-ax-\frac{b}{2}x^2\right)+1)^{m-1}}{(\theta \exp \left(-ax-\frac{b}{2}x^2\right)+1)^m-1}.\]
We can find that the LFRP distribution can be obtained as limiting of LFRB distribution if $m\theta \to \lambda >0$,  when $m\to \infty $.
If $b=0$, the density function in (\ref{eq.fLB}) becomes the density function of exponential binomial (EB) distribution \citep{ch-ga-09}.
If $m=1$, then the density function in (\ref{eq.fLB}) changes to the density of LFR distribution.

\begin{figure}[h]
\centering
\includegraphics[scale=0.40]{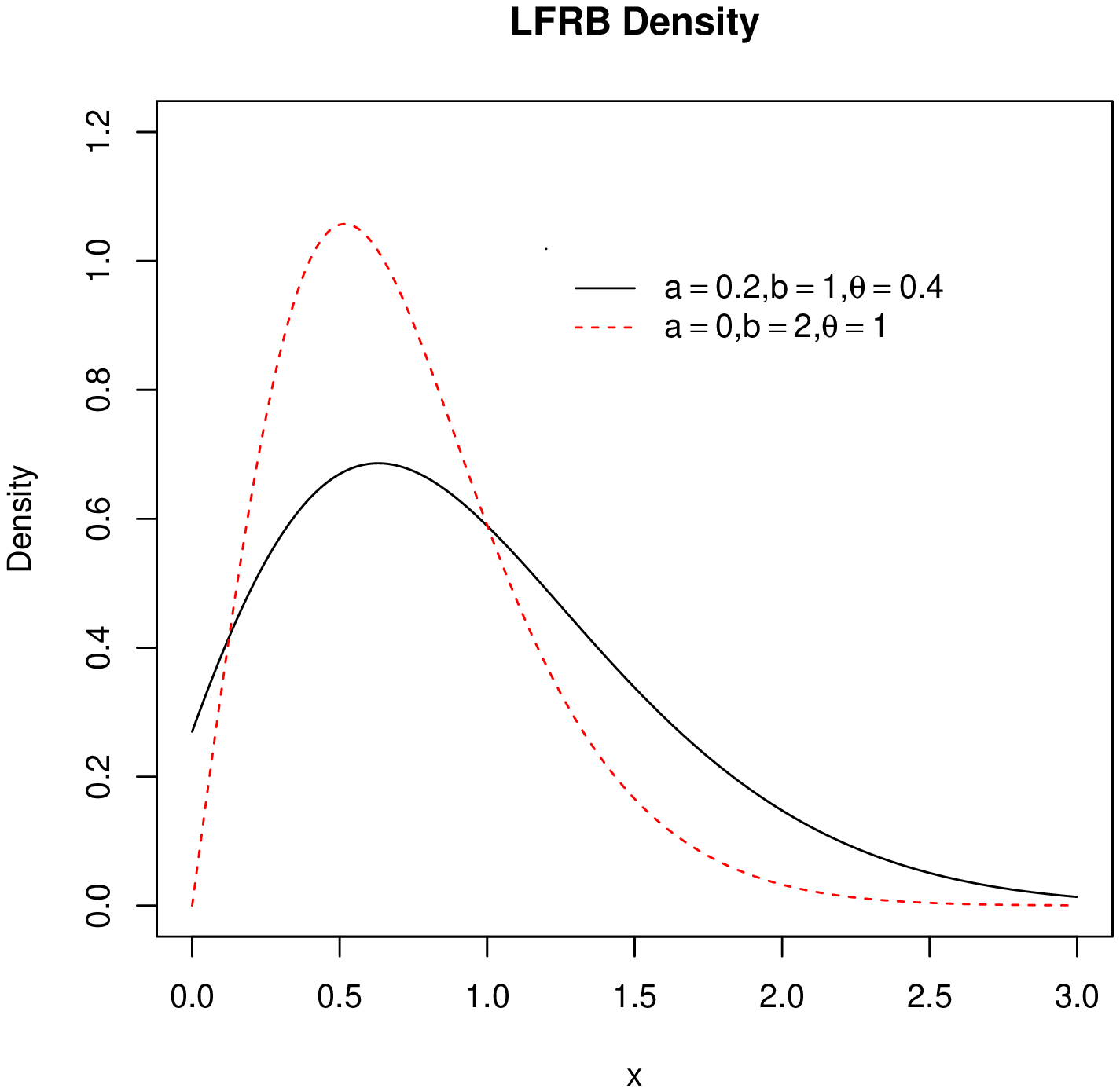}
\includegraphics[scale=0.40]{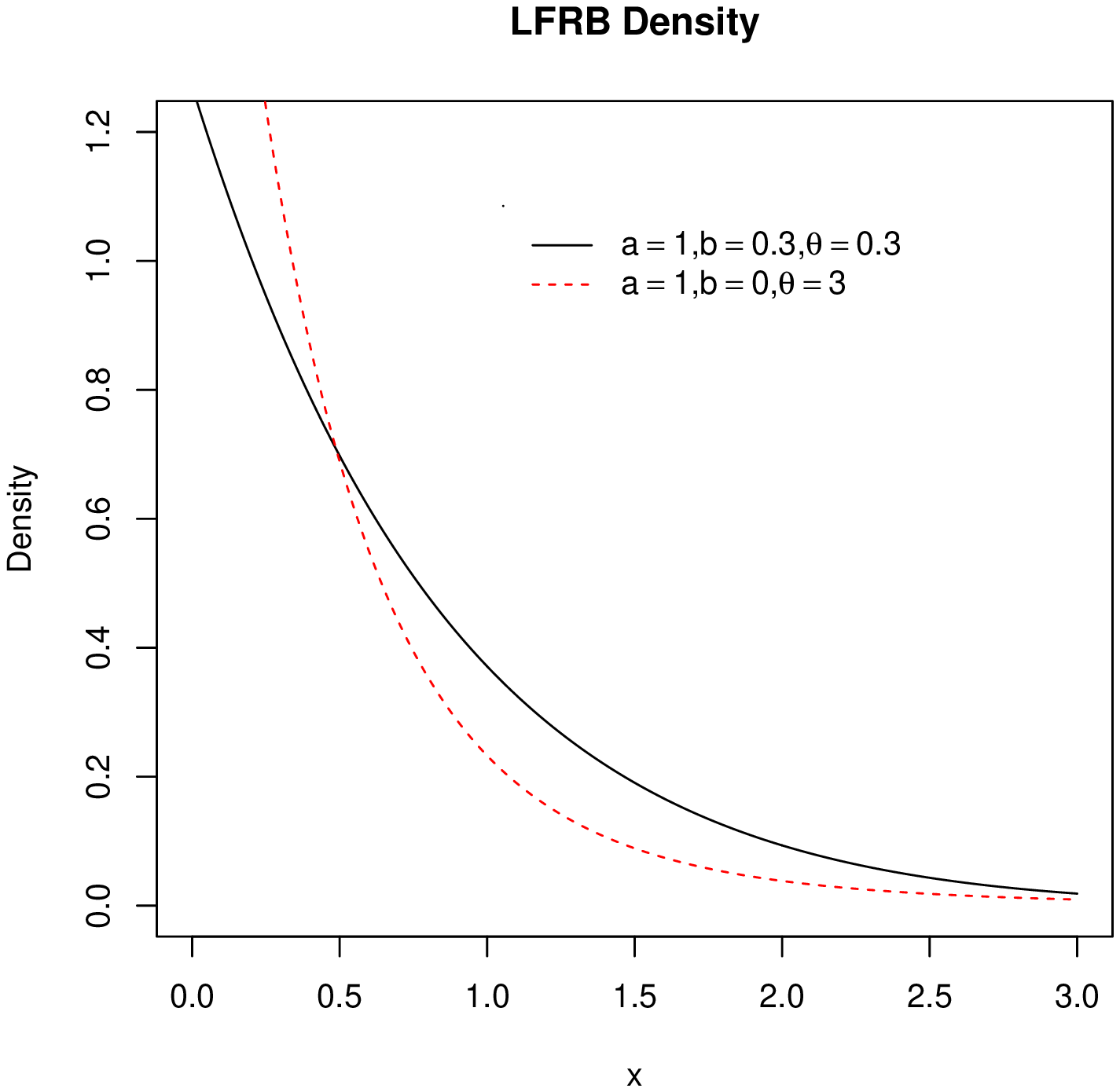}
\includegraphics[scale=0.40]{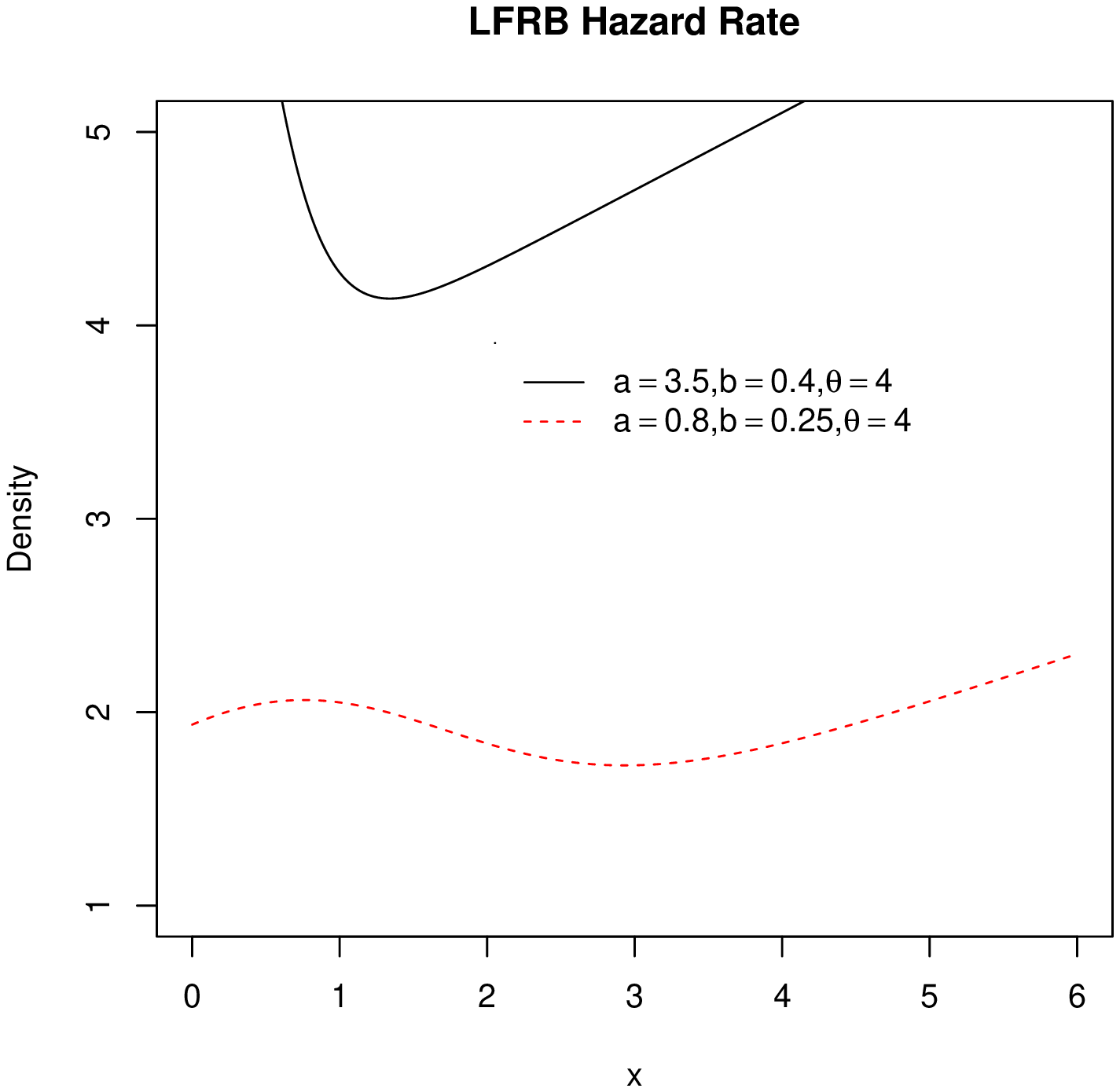}
\includegraphics[scale=0.40]{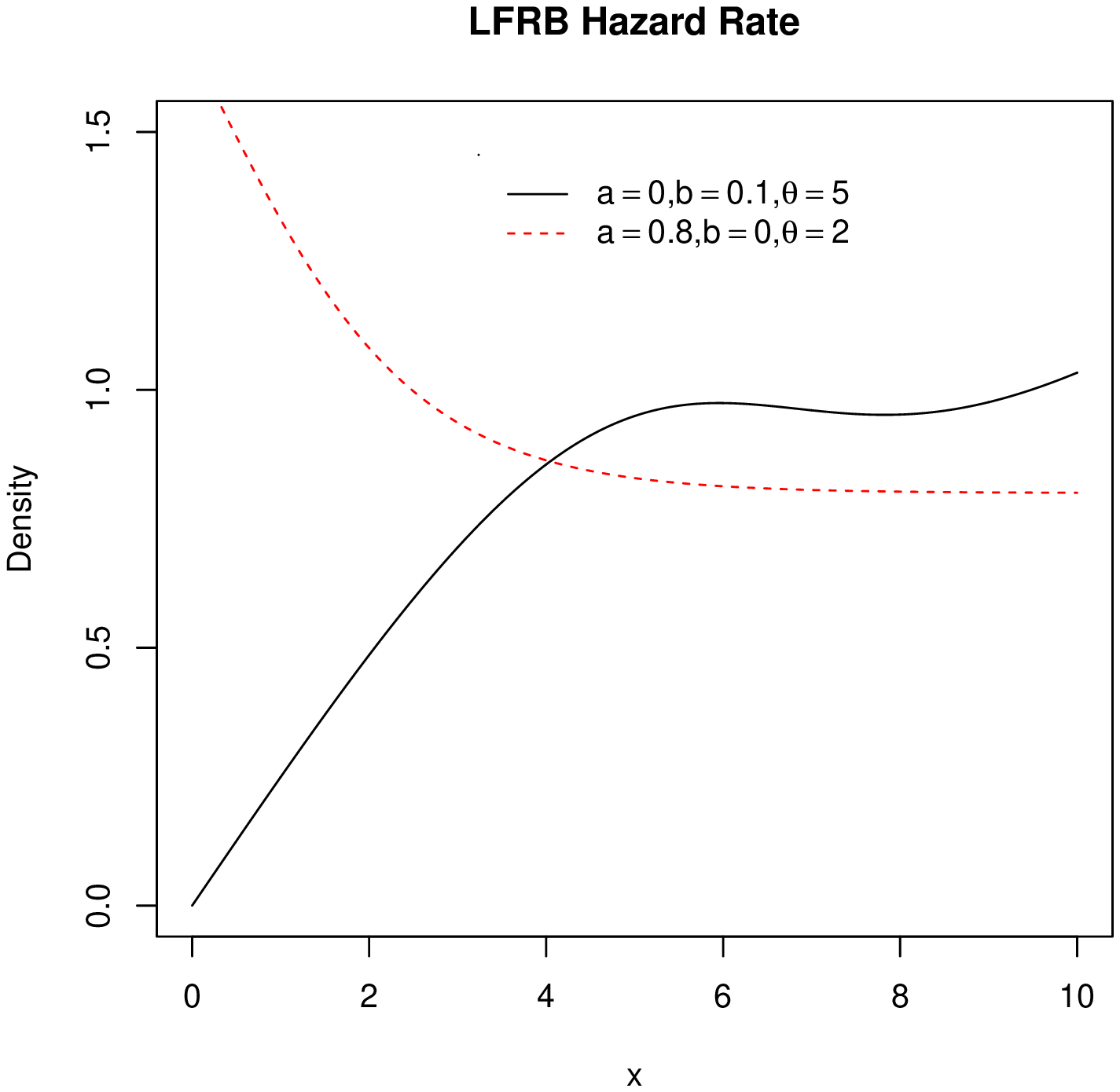}
\vspace{-0.8cm}
\caption[]{Plot of density and hazard rate function of the LFRB distribution for different values of its parameters.  }
\end{figure}

\subsection{ LFR-logarithmic distribution}
The logarithmic distribution (truncated at zero) is a special case of power series distributions with $a_n=\frac{1}{n}$ and $C\left(\theta \right)=-{\log  \left(1-\theta \right)\ }$ ($0<\theta <1$). The density function of LFR-logarithmic (LFRL) distribution is given by
\begin{equation}\label{eq.fLL}
f\left(x\right)=\frac{\theta (a+bx){\exp  \left(-ax-\frac{b}{2}x^2\right)\ }}{-{\log  \left(1-\theta \right)\ }\left(1-\theta {\exp  \left(-ax-\frac{b}{2}x^2\right)\ }\right)},
\end{equation}
and its hazard rate function is given by
\[h\left(x\right)=\frac {\theta (a+bx)\exp \left(-ax-\frac{b}{2}x^2\right)(\theta \exp \left(-ax-\frac{b}{2}x^2\right)-1)^{-1}}{\log\left(1-\theta \exp \left(-ax-\frac{b}{2}x^2\right)\right)}.\]
If $b=0$, the density function in (\ref{eq.fLL}) becomes the density function of the exponential-logarithmic (EL) distribution \citep{ta-re-08}.
\begin{figure}[]
\centering
\includegraphics[scale=0.40]{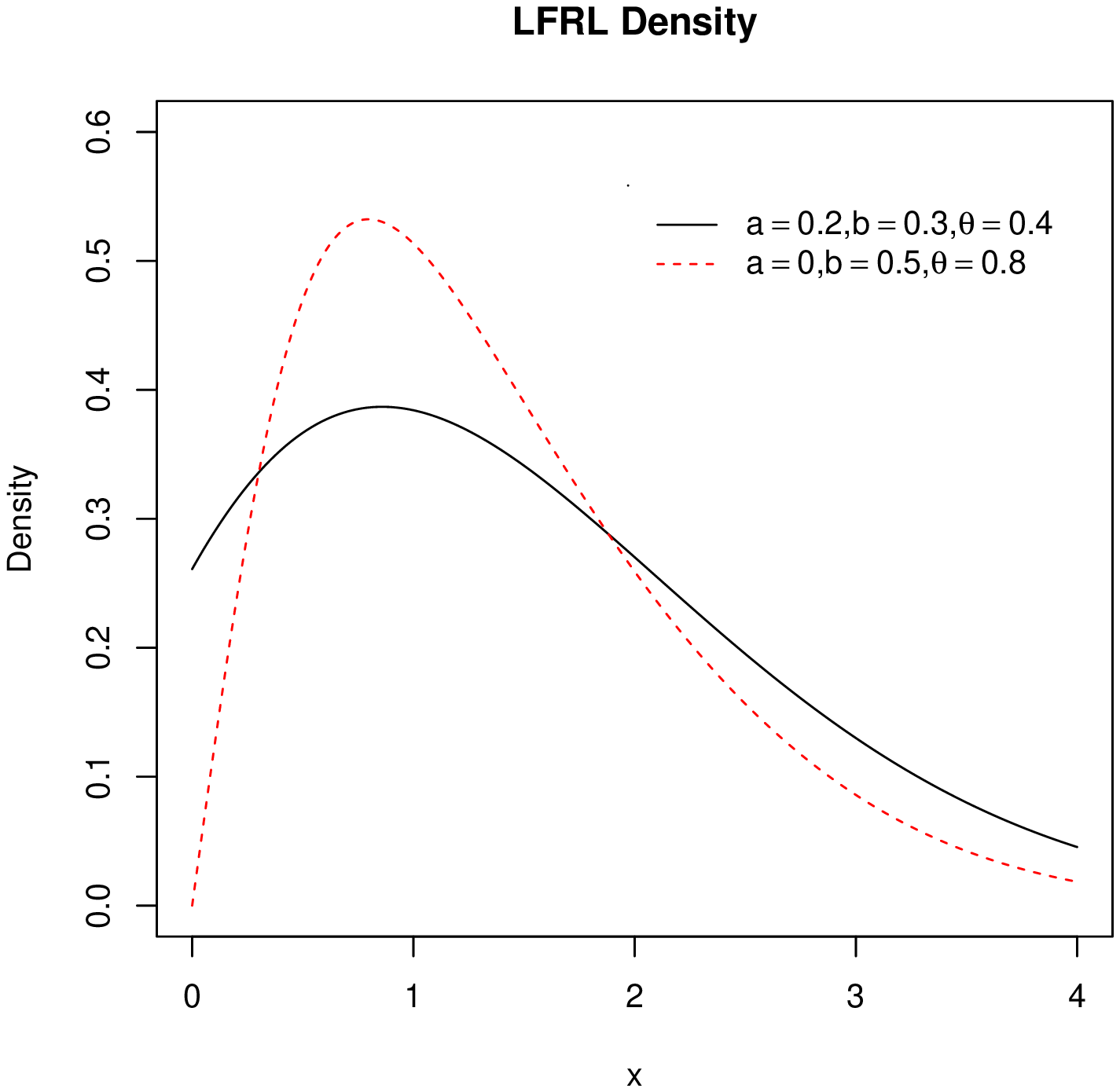}
\includegraphics[scale=0.40]{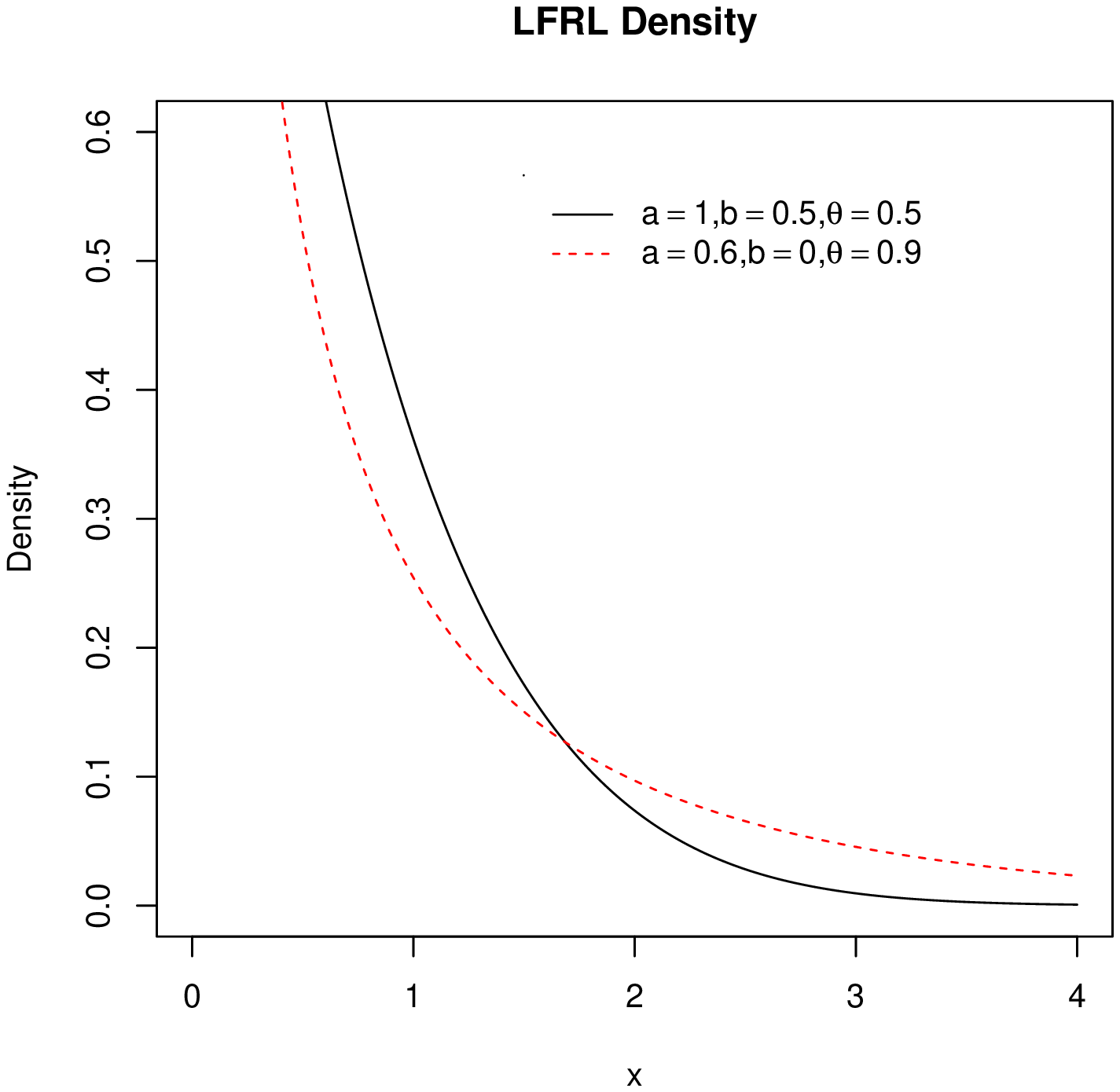}
\includegraphics[scale=0.40]{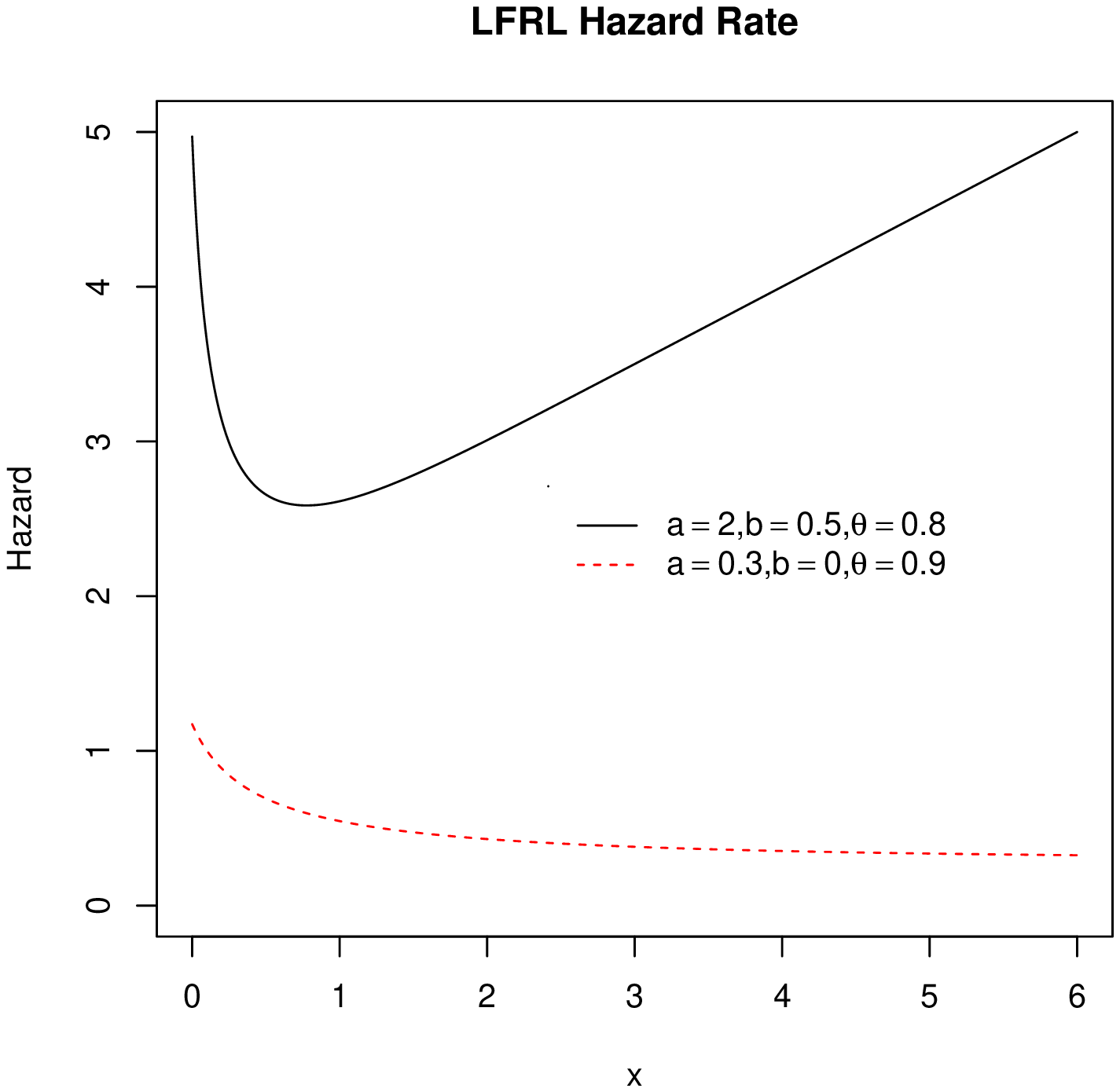}
\includegraphics[scale=0.40]{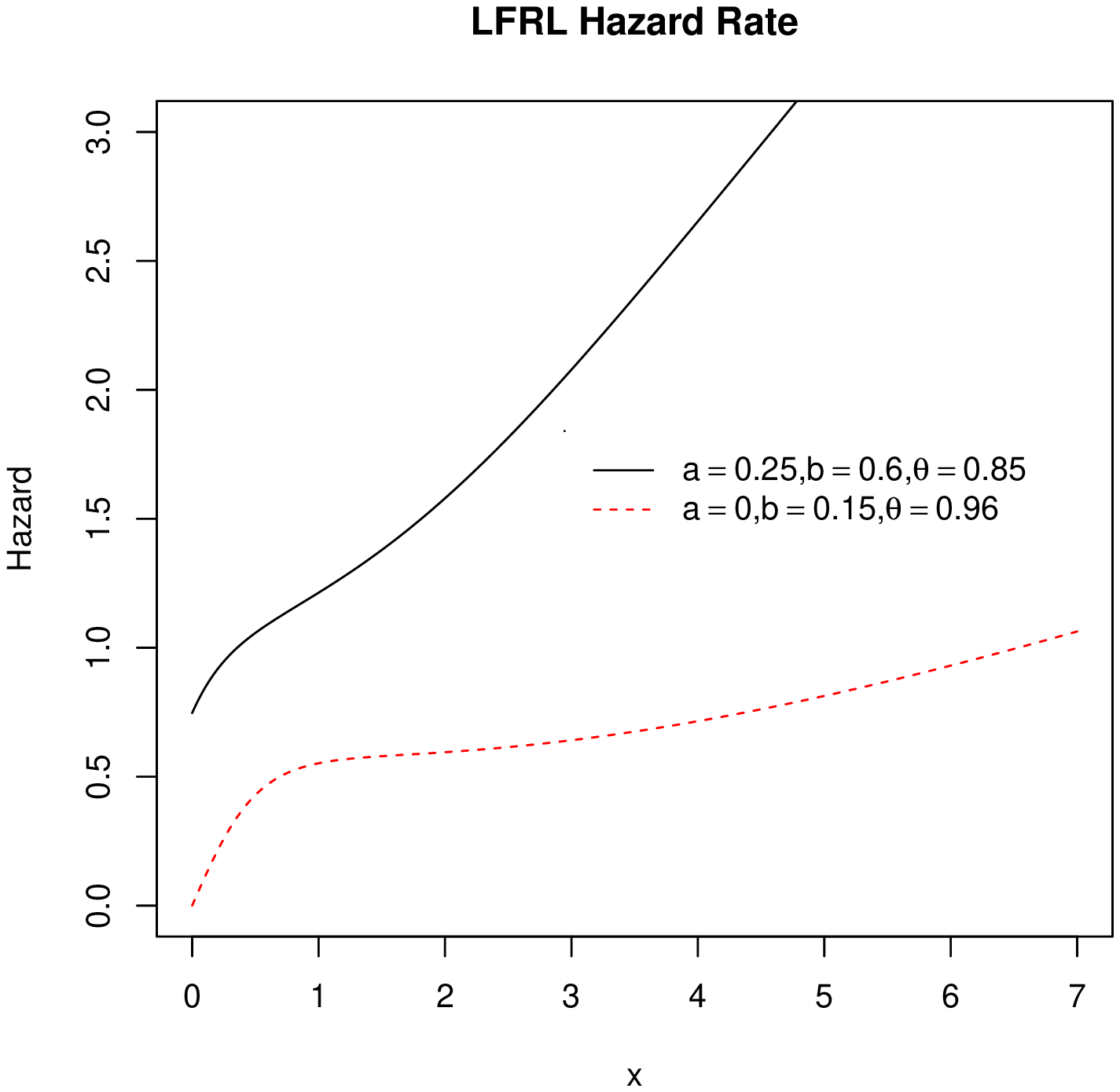}
\vspace{-0.8cm}
\caption[]{Plot of density and hazard rate function of the LFRL distribution for different values of its parameters.}
\end{figure}

\section{Estimation and inference}
\label{se.mle}
In this section, we discuss the maximum likelihood estimates (MLEs) of the parameters of the LFRPS distributions using a complete sample.

\subsection{MLE's}
Let $x_1,..., x_n$ be the observed values of a random sample of size \textit{n }from the ${\rm LFRPS}\left(a,b,\theta \right)$ distributions. The log-likelihood function for the vector of parameters ${\mathbf \Theta } ={\left(a,b,\theta \right)}^{{\rm T}}$ can be written as
\begin{eqnarray}\label{eq.ln}
{\ell }_n&=&{\ell }_n\left({\boldsymbol x};{\mathbf \Theta } \right)=n{\log  \left(\theta \right)\ }+\sum^n_{i=1}{{\log  (a+bx_i)\ }}-na\bar{x}-\frac{nb}{2}\bar{x^2}\nonumber\\
&&+\sum^n_{i=1}{{\log  \left(C'\left(\theta p_i\right)\right)\ }}-n{\log  \left(C(\theta )\right)\ },
\end{eqnarray}
where $p_i={\exp  (-ax_i-\frac{b}{2}x^2_i)\ }$, $\bar{x}=n^{-1}\sum^n_{i=1}{x_i}$ and $\bar{x^2}=n^{-1}\sum^n_{i=1}{x^2_i}$. The components of the score vector $U_n=\left(\frac{\partial{\ell }_n}{\partial a},\frac{\partial{\ell }_n}{\partial b}\ , \frac{\partial{\ell }_n}{\partial \theta}\right)$ are given by
\begin{eqnarray}
&&\frac{\partial{\ell }_n}{\partial a}=\sum^n_{i=1}{\frac{1}{a+bx_i}}-n\overline{x}-\sum^n_{i=1}{\frac{\theta x_ip_iC''\left(\theta p_i\right)}{C'\left(\theta p_i\right)}},\label{eq.mla}\\
&&\frac{\partial{\ell }_n}{\partial b}=\sum^n_{i=1}{\frac{x_i}{a+bx_i}}-\frac{n}{2}\bar{x^2}-\sum^n_{i=1}{\frac{\theta x^2_ip_iC''\left(\theta p_i\right)}{2C'\left(\theta p_i\right)}},\label{eq.mlb}\\
&& \frac{\partial{\ell }_n}{\partial \theta}=\frac{n}{\theta }+\sum^n_{i=1}{\frac{p_iC''\left(\theta p_i\right)}{C'\left(\theta p_i\right)}}-\frac{nC'(\theta )}{C(\theta )}.\label{eq.mlt}
\end{eqnarray}
The maximum likelihood estimator of ${\mathbf \Theta } ={\left(a, b,\theta \right)}^{{\rm T}}$ is obtained by numerically solving the nonlinear system of equations $U_n=0$. It is usually more convenient to use a nonlinear optimization algorithm (such as the quasi-Newton algorithm) to numerically maximize the log-likelihood function in (\ref{eq.ln}).

For interval estimation and hypothesis tests on the model parameters, we require the observed information matrix. The $3\times 3$ observed information matrix $I_{n}=I_{n}\left({\mathbf \Theta }\right)$ is obtained as
\[I_{n}({\mathbf \Theta } )=\left[ \begin{array}{ccc}
I_{aa} & I_{ab} & I_{a\theta } \\
I_{ba} & I_{bb} & I_{b\theta } \\
I_{\theta a} & I_{\theta b} & I_{\theta \theta } \end{array}
\right],\]
where the expressions for the elements of $I_{n}({\mathbf \Theta } )$ are given in Appendix A.

Applying the usual large sample approximation, MLE of ${\mathbf \Theta }$ i.e., $\widehat{\mathbf \Theta }$ can be treated as being approximately
$N_3({\mathbf \Theta } ,{J_n({\mathbf \Theta } )}^{-1}{\mathbf )}$, where $J_n\left({\mathbf \Theta }
\right)=E\left[I_n\left({\mathbf \Theta } \right)\right]$. Under conditions
that are fulfilled for parameters in the interior of the parameter
space but not on the boundary, the asymptotic distribution of
$\sqrt{n}(\widehat{{\mathbf \Theta } }{\rm -}{\mathbf \Theta } {\rm )}$ is $N_3({\mathbf
0},{J({\mathbf \Theta } )}^{-1})$, where $J\left({\mathbf \Theta }\right)={\mathop{\lim
}_{n\to \infty } {n^{-1}I}_n({\mathbf \Theta } )\ }$ is the unit information
matrix. This asymptotic behavior remains valid if $J({\mathbf \Theta }
)$ is replaced by the average sample information matrix
evaluated at $\widehat{{\mathbf \Theta } }$, say ${n^{-1}I}_n(\widehat{{\mathbf \Theta }
})$. The estimated asymptotic multivariate normal $N_3({\mathbf \Theta }
,{I_n(\widehat{{\mathbf \Theta }})}^{-1})$ distribution of $\widehat{{\mathbf \Theta }}$
can be used to construct approximate confidence intervals for the
parameters and for the hazard rate and survival functions. An
$100(1-\gamma )$ asymptotic confidence interval for each parameter
${{\mathbf \Theta }}_{{\rm r}}$ is given by
\[{ACI}_r=({\widehat{{\mathbf \Theta }
}}_r-Z_{\frac{\gamma }{2}}\sqrt{{\hat{I}}^{rr}},{\widehat{{\mathbf \Theta }
}}_r+Z_{\frac{\gamma }{2}}\sqrt{{\hat{I}}^{rr}}),\] where
${\hat{I}}^{rr}$ is the (\textit{r, r}) diagonal element of
$I_n(\widehat{{\mathbf \Theta } })^{-1}$ for $r=1,~2,~3$ and $Z_{\frac{\gamma
}{2}}$ is the quantile $1-\gamma /2$ of the standard normal
distribution.

For each element of the power-series distributions (geometric,
Poisson, logarithmic and binomial), we have the following theorems
for MLE's:
\begin{theorem} \label{th.b1}
Let ${\rm g}_ 1\left(a;b,\theta,{\boldsymbol x}\right)$ denotes the function on RHS of the expression in (\ref{eq.mla}), where $b$ and $\theta$ are the true value of the parameters. Then,

\noindent (i) for a given $b>0$, and $\theta>0$, the root of ${\rm g}_1\left(a;b,\theta,{\boldsymbol x}\right)=0$ lies in the interval:
\[\left(\ (\bar{x}+\frac{k_1}{n})^{-1}-bx_{(n)},\frac{1}{\bar{x}}-bx_{\left(1\right)}\right),\]

\noindent  (ii) for a given $b>0$, and $\theta<0$, the root of ${\rm g}_1\left(a;b,\theta,{\boldsymbol x}\right)=0$ lies in the interval:
\[\left(\frac{1}{\bar{x}}-bx_{(n)},(\bar{x}+\frac{k_1}{n})^{-1}-bx_{(1)} \right),\]
where $\bar{x}=n^{-1}\sum_{i=1}^n  x_i$,  $k_1=\sum^n_{i=1}{\frac{\theta x_iv^b_i C''\left(\theta v^b_i\right)}{C'\left(\theta v^b_i\right)}}$, $v_i=e^{-\frac{1}{2}x^2_i}$, $x_{(1)}=\min(x_1,\dots,x_n)$ and $x_{(n)}=\max(x_1,\dots,x_n)$.
\end{theorem}
\begin{proof}
See Appendix B.1.
\end{proof}
\begin{theorem}
 Let ${\rm g}_2\left(b;a,\theta,{\boldsymbol x}\right)$ denotes the function on RHS of the expression in (\ref{eq.mlb}), where $a$ and $\theta$ are the true value of the parameters. Then,

\noindent (i) for a given $a{\rm >}0$, and $\theta>0$, the root of ${\rm g}_2\left(b;a,\theta ,{\boldsymbol x}\right)=0$ lies in the interval:
\[\left((\frac{\bar{x^2}}{2}+\frac{k_2 }{2n})^{-1}-\frac{a}{x_{\left(1\right)}},\ \frac{2}{\bar{x^2}}-\frac{a}{x_{\left(n\right)}} \right),\]

\noindent (ii) for a given $a{\rm >}0$, and $\theta<0$, the root of ${\rm g}_2\left(b;a,\theta ,{\boldsymbol x}\right)=0$ lies in the interval:
\[\left(\ \frac{2}{\bar{x^2}}-\frac{a}{x_{(1)}}, (\frac{\bar{x^2}}{2}+\frac{k_2 }{2n})^{-1}-\frac{a}{x_{(n)}} \right),\]
where $\bar{x^2}=n^{-1}\sum_{i=1}^n  x_i^2$, $k_2=\sum^n_{i=1}\frac{\theta x^2_iu^a_iC''\left(\theta u^a_i\right)}{C'\left(\theta u^a_i\right)}$, $u_i=e^{-x_i}$.
\end{theorem}
\begin{proof}
The proof is similar to the proof of Theorem \ref{th.b1}.
\end{proof}
\begin{theorem}
 Let ${\rm g}_3\left(\theta;a,b,{\boldsymbol x}\right)$ denotes the function on RHS of the expression in (\ref{eq.mlt})  where $a$ and $b$ are the true values of the parameters.
 \begin{enumerate}
   \item[a)] The equation ${\rm g}_3\left(\theta;a,b,{\boldsymbol x}\right)=0$ has at least one root if for all LFRG, LFRP and LFRL distributions $\sum^n_{i=1} p_i>\frac{n}{2}$, where $p_i=\exp(-ax_i-\frac{b}{2}x^2_i)$.

   \item[b)] If ${\rm g}_3\left(\theta;a,b,{\boldsymbol x}\right) =\frac{\partial l_n}{\partial p}$, where $p=\frac{\theta}{\theta +1}$ and $p\in(0,1)$ then the equation ${\rm g}_3\left(\theta ;a,b,{\boldsymbol x}\right)=0$ has at least one root for LFRB distribution if $\sum^n_{i=1}p_i>\frac{n}{2}$ and $\sum^n_{i =1}\frac{1}{p_i}>\frac{nm}{1-m}$.
 \end{enumerate}
\end{theorem}
\begin{proof}
See the Appendix B.2.
\end{proof}

\subsection{EM-algorithm}

The solution of the three non-linear normal equations in (\ref{eq.mla})-(\ref{eq.mlt}) is needed using a numerical method. In some cases, solving these equations is difficult; therefore, we propose the use of the Expectation--Maximization (EM) algorithm
\citep{de-la-ru-77}.
 In each iteration of this algorithm, there are two steps, called the Expectation step or the E-step and the Maximization step or the M-step. EM algorithm is a very powerful tool in handling the incomplete data problem.

For doing this, we define an hypothetical complete-data distribution with a joint density function
\[{\rm g}\left(x, z,\theta \right)=\frac{a_z{\theta }^z}{C\left(\theta \right)}z\left(a+bx\right)e^{-azx-\frac{b}{2}zx^2},\ x>0,\ z\in N.\]
The E-step of an EM cycle requires the conditional expectation of $(Z|X;{\mathbf \Theta }^{(r)})$, where ${\mathbf \Theta }^{(r)}=\left(a^{(r)}, b^{\left(r\right)},{\theta }^{\left(r\right)}\right)$ is the current estimate of ${\mathbf \Theta }$. From
\[{\rm g}_{Z|X}\left(z|x\right)=\frac{za_z{\theta }^{z-1}e^{-\left(ax+\frac{b}{2}x^2\right)(z-1)}}{C'\left(\theta e^{-\left(ax+\frac{b}{2}x^2\right)}\right)},\]
we have
\[E\left(Z|X=x \right)=1+\frac{\theta e^{-\left(ax+\frac{b}{2}x^2\right)}C''\left(\theta e^{-\left(ax+\frac{b}{2}x^2\right)}\right)}{C'\left(\theta e^{-\left(ax+\frac{b}{2}x^2\right)}\right)}.\]

The EM cycle is completed with the M-step using the maximum likelihood estimation over ${\mathbf \Theta }$, with the missing $Z$'s replaced by their conditional expectations given above. The log-likelihood for the complete-data is
\begin{equation}\label{eq.lns}
{\ell }^*_n\left({\boldsymbol x}; {\boldsymbol z};{\mathbf \Theta }\right)\propto n\bar{z}\log (\theta)+\sum^n_{i=1}{{\log  (a+bx_i)\ }}-a\sum^n_{i=1}{z_ix_i}-\frac{b}{2}\sum^n_{i=1}{z_ix^2_i}-n\log  \left(C\left(\theta \right)\right),
\end{equation}
where ${\boldsymbol x}=(x_1,...x_n)$, ${\boldsymbol z}=(z_1,...z_n)$, and $\bar{z}=n^{-1}\sum^n_{i=1}z_i$. The components of the score function $U_c\left({\boldsymbol y};\Theta\right)=\left(\frac{{\partial \ell }^*_n}{\partial a}, \frac{{\partial \ell }^*_n}{\partial b},\frac{{\partial \ell }^*_n}{\partial\theta }\right)^{\rm T}$, where ${\boldsymbol y}=({\boldsymbol x},{\boldsymbol z})$, are obtained by differentiation of (\ref{eq.lns}) with respect to parameters $a$, $b$, and $\theta $, as
\begin{equation}
\frac{{\partial \ell }^*_n}{\partial a}=\sum^n_{i=1}{\frac{1}{a+bx_i}}-\sum^n_{i=1}{z_ix_i},\quad
\frac{{\partial \ell }^*_n}{\partial b}=\sum^n_{i=1}{\frac{x_i}{a+bx_i}}-\frac{1}{2}\sum^n_{i=1}{z_ix^2_i},\quad
\frac{{\partial \ell }^*_n}{\partial\theta }=\frac{n\bar{z}}{\theta }-n\frac{C'\left(\theta \right)}{C\left(\theta \right)}.
\end{equation}
Therefore, the iterative procedure of the EM-algorithm reduces as the following equations:
\begin{eqnarray}
&&\sum^n_{i=1}{\frac{1}{{\hat{a}}^{(t+1)}+{\hat{b}}^{(t)}x_i}}=\sum^n_{i=1}{{\hat{z}}^{(t)}_ix_i},\\
&&\sum^n_{i=1}{\frac{x_i}{{\hat{a}}^{(t)}+{\hat{b}}^{(t+1)}x_i}}=\frac{1}{2}\sum^n_{i=1}{{\hat{z}}^{(t)}_ix^2_i},\\
&&{\hat{\theta }}^{(t+1)}=\frac{C\left({\hat{\theta }}^{(t+1)}\right)}{nC'\left({\hat{\theta }}^{(t+1)}\right)}\sum^n_{i=1}{{\hat{z}}^{(t)}_i},
\end{eqnarray}
where
\begin{equation}
{\hat{z}}^{(t)}_i=1+\frac{{\hat{\theta }}^{(t)}e^{-\left({\hat{a}}^{(t)}x_i+\frac{{\hat{b}}^{(t)}}{2}x^2_i\right)}C''\left({\hat{\theta }}^{(t)}e^{-\left({\hat{a}}^{(t)}x_i+\frac{{\hat{b}}^{(t)}}{2}x^2_i\right)}\right)}{C'\left({\hat{\theta }}^{(t)}e^{-\left({\hat{a}}^{(t)}x_i+\frac{{\hat{b}}^{(t)}}{2}x^2_i\right)}\right)}.
\end{equation}
The estimations of the parameters  based on the EM algorithm are obtained by using these equations and we have the following theorems about the roots of equations.
\begin{theorem}
Let
\[h_1\left(a\right)=\sum^n_{i=1}{\frac{1}{a+
{\hat{b}}^{(t)}x_i}}-c_1,\]
where $c_1=\sum^n_{i=1}{{\hat{z}}^{\left(t\right)}_ix_i}$. Then the root of $h_1\left(a\right)=0$ is unique and lies in the interval:
\[\left(\frac{n}{c_1}-{\hat{b}}^{(t)}x_{(n)},\ \frac{n}{c_1}-{\hat{b}}^{(t)}x_{(1)}\right).\]
\end{theorem}
\begin{proof}
See Appendix C.1.
\end{proof}
\begin{theorem}
Let
\[h_2\left(b\right)=\sum^n_{i=1}{\frac{x_i}{{\hat{a}}^{(t)}+bx_i}}-
\frac{c_2}{2},\]
where $c_2=\sum^n_{i=1}{{\hat{z}}^{\left(t\right)}_ix^2_i}$. Then the root of $h_2\left(b\right)=0$ is unique and lies in the following interval:
\[\left(\frac{2n}{x_{(n)}c_2}
-\frac{{\hat{a}}^{\left(t\right)}}{x_{(n)}},\ \frac{2n}{x_{(1)}c_2}
-\frac{{\hat{a}}^{\left(t\right)}}{x_{(1)}}\right).\]
\end{theorem}
\begin{proof}
See Appendix C.2.
\end{proof}
\begin{theorem}
 Let
\[h_3\left(\theta \right)=\theta -\frac{c_0 C(\theta )}{nC'(\theta )},\]
where $c_0=\sum^n_{i=1}{{\hat{z}}^{\left(t\right)}_i}$. Then the root of $h_3\left(\theta \right)=0$,

\noindent (i) is unique and is equal to $1-\frac{n}{c_0}$ for LFRG distribution, and $0<\theta<1$.

\noindent (ii) is unique for LFRP, LFRL and LFRB  distributions.

\end{theorem}
\begin{proof}
See Appendix C.3.
\end{proof}

In this part we use the results of \cite{lou-82}
to obtain the standard errors of the estimators from the EM-algorithm.

 The elements of the $3\times 3$ observed information matrix $I_c\left(\Theta ,{\boldsymbol y}\right)=-\left[\frac{\partial U_c\left({\boldsymbol y}; \Theta\right)}{\partial \Theta }\right]$ are given by
\begin{eqnarray*}
&&\frac{{{\partial}^2\ell }^*_n}{\partial a^2}=-\sum^n_{i=1}{\frac{1}{{\left(a+bx_i\right)}^2}},\ \ \ \
\frac{{{\partial }^2\ell }^*_n}{\partial b\partial a}=\frac{{{\partial}^2\ell }^*_n}{\partial a\partial b}=-\sum^n_{i=1}{\frac{x_i}{{\left(a+bx_i\right)}^2}},\\
&&\frac{{{\partial }^2\ell }^*_n}{\partial\theta \partial a}=\frac{{{\partial }^2\ell }^*_n}{\partial a\partial\theta }=0, \ \ \ \
\frac{{{\partial }^2\ell }^*_n}{\partial b^2}=-\sum^n_{i=1}{\frac{x^2_i}{{\left(a+bx_i\right)}^2}},\ \ \ \ \frac{{{\partial }^2\ell }^*_n}{\partial\theta \partial b}=\frac{{{\partial }^2\ell }^*_n}{\partial b\partial \theta }=0,\\
&& \frac{{{\partial }^2\ell }^*_n}{\partial {\theta }^2}=-\frac{n\bar{z}}{{\theta }^2}-n\frac{C''\left(\theta \right)}{C\left(\theta \right)}+\frac{(C'\left(\theta \right))^2}{(C\left(\theta \right))^2}.
\end{eqnarray*}
Taking the conditional expectation of $I_c\left(\Theta ,\boldsymbol y\right)$ given ${\boldsymbol x}$, we obtain the $3\times 3$ matrix
\begin{equation}\label{eq.lc}
{\ell }_c\left( \Theta;{\boldsymbol x}\right)=E\left(I_c\left(\Theta ,{\boldsymbol y}\right)|{\boldsymbol x}\right)=\left[c_{ij}\right],
\end{equation}
where
\begin{eqnarray*}
c_{11}=\sum^n_{i=1}{\frac{1}{{\left(a+bx_i\right)}^2}},&&  c_{12}=c_{21}=\sum^n_{i=1}{\frac{x_i}{{\left(a+bx_i\right)}^2}},\ \ \ c_{13}=c_{31}=c_{23}=c_{32}=0,\\
c_{22}=\sum^n_{i=1}{\frac{x^2_i}{{\left(a+bx_i\right)}^2}},&&  c_{33}=\frac{1}{{\theta }^2}\sum^n_{i=1}E(Z_i|x_i)+n\frac{C^{''\left(\theta \right)}}{C\left(\theta \right)}-\frac{(C'\left(\theta \right))^2}{(C\left(\theta \right))^2},
\end{eqnarray*}
and
\[E\left(Z_i| x_i\right)=1+\frac{èe^{-\left(ax_i+\frac{b}{2}x^2_i\right)}C''\left(èe^{-\left(ax_i+\frac{b}{2}x^2_i\right)}\right)}{C'\left(\theta e^{-\left(ax_i+\frac{b}{2}x^2_i\right)}\right)}.\]
Moving now to the computation of ${\ell }_m\left(\Theta;{\mathbf x}\right)$ as
\begin{equation}\label{eq.lm}
{\ell }_m\left(\Theta;{\mathbf x}\right)=Var\left[U_c\left({\mathbf y};\Theta\right)|{\mathbf x}\right]=\left[v_{ij}\right],
\end{equation}
where
\begin{eqnarray*}
&&v_{11}=\sum^n_{i=1}{x_iVar\left(Z_i|x_i\right)},\ \ \ \qquad
v_{12}=v_{21}=\frac{1}{2}\sum^n_{i=1}{x^3_iVar\left(Z_i|x_i\right)},\\
&&
v_{13}=v_{31}=\frac{1}{\theta }\sum^n_{i=1}{x_iVar\left(Z_i|x_i\right)}, \ \ \ \
v_{22}=\frac{1}{4}\sum^n_{i=1}{x^2_iVar\left(Z_i|x_i\right)},\\
&&
v_{23}=v_{32}=-\frac{1}{2\theta }\sum^n_{i=1}{x^2_iVar\left(Z_i|x_i\right)},\ \ \ \
v_{33}=\frac{1}{{\theta }^2}\sum^n_{i=1}{Var\left(Z_i|x_i\right)},
\end{eqnarray*}
and
\begin{eqnarray*}
Var\left(Z_i|x_i\right)&=&E\left(Z^2_i|x_i\right)-{\left[E\left(Z_i|x_i\right)\right]}^2\\
&=&\frac{1}{{\left[C'\left({\theta }_*\right)\right]}^2}\sum^{\infty }_{z=1}{a_zz^3{\theta }^{z-1}_*}-\frac{{\left[C'\left({\theta }_*\right)+{\theta }_*C''\left({\theta }_*\right)\right]}^2}{{\left[C'\left({\theta }_*\right)\right]}^2}\\
&=&\frac{1}{{\left[C'\left({\theta }_*\right)\right]}^2}\left[{\theta }^2_*C'''\left({\theta }_*\right)+C'\left({\theta }_*\right)+3{\theta }_*C''\left({\theta }_*\right)\right]\\
&&
-\frac{{\left[C'\left({\theta }_*\right)+{\theta }_*C''\left({\theta }_*\right)\right]}^2}{{\left[C'\left({\theta }_*\right)\right]}^2},
\end{eqnarray*}
in which ${\theta }_*=\theta e^{-\left(ax_i+\frac{b}{2}x^2_i\right)}$. Using (\ref{eq.lc}) and (\ref{eq.lm}), we obtain the observed information matrix as
\begin{equation}\label{EM information}
J(\hat{\Theta};{\boldsymbol x})={\ell }_c(\hat{ \Theta};{\boldsymbol x})-{\ell }_m(\hat{\Theta};{\boldsymbol x}).
\end{equation}
The standard errors of the MLE's of the EM-algorithm are the square root of the diagonal elements of the $J(\hat{\Theta};{\boldsymbol x})$.

\section{Simulation study}
\label{se.sim}
This section presents the results of a simulation study based on the
assumptions given in Theorems 5.1-5.6 and Equation (\ref{EM information}) which provides
a method for calculating the standard errors of the MLEs of the
EM-algorithm. The proposed EM-algorithm is considered. No
restriction on the maximum number of iterations and convergence is
assumed when the absolute differences between successive estimates
are less than ${10}^{-5}$.

Firstly, simulations have been performed in order to investigate the
proposed estimator of $a,$ $b$ and $\theta $ of the
proposed EM-scheme. We generated 1000 samples of size $n=30,\ 70,\
100$ and $200$ from the LFRG and LFRP distributions for each one of
the twelve set of values of $\left(a ,b ,\theta \right).$
In the second stage, we assess the accuracy of the approximation of
the standard error of the MLEs of the EM-algorithm determined though
the Fisher information matrix. The simulated values of
$se(\hat{a})$, $se(\hat{b})$, $se(\hat{\theta })$, $Cov(\hat{a},\hat{b})$, $Cov(\hat{a},\hat{\theta})$ and $Cov(\hat{b},\hat{\theta})$ as
well as the approximate values of $se(\hat{a})$, $se(\hat{b})$, $se(\hat{\theta })$, $Cov(\hat{a},\hat{b})$, $Cov(\hat{a},\hat{\theta})$ and $Cov(\hat{b},\hat{\theta})$, obtained by averaging
the corresponding values of the observed information matrices, are
computed. The results for the LFRG and LFRP distributions are shown in
Tables 1-4, which indicate the following results: (i)
convergence has been achieved in all cases and this emphasizes the
numerical stability of the EM-algorithm. (ii) The differences
between the average estimates and the true values are almost small.
(iii) These results suggest that the EM estimates have performed
consistently. (iv) The standard errors and covariance of the MLEs decrease when the
sample size increases. (v) Additionally, the standard errors and covariance of the MLEs
of the EM-algorithm obtained from the observed information matrix
are quite close to the simulated ones for large values of $n$.
\begin{sidewaystable}
\caption{The average MLE's, mean of the simulated
standard errors, mean of the simulated covariances, mean of the standard errors and covariances of EM estimators
obtained using observed information matrix of the LFRG distribution for $n=30,~70$.}
{\footnotesize
\begin{tabular}{|c|ccc|ccc|ccc|ccc|ccc|ccc|ccc|ccc|} \hline
&\multicolumn{3}{|c|}{Parameter}&
\multicolumn{3}{|c|}{AE}&  \multicolumn{3}{|c|}{Bias}&
\multicolumn{3}{|c|}{Sim.std}& \multicolumn{3}{|c|}{EM.std}& \multicolumn{3}{|c|}{Sim.Cov}&
\multicolumn{3}{|c|}{EM.Cov}\\ \hline

$n$&$a$      & $b$      & $\theta$      &
$\hat{a}$& $\hat{b}$& $\hat{\theta}$&
$\hat{a}$& $\hat{b}$& $\hat{\theta}$&
$\hat{a}$& $\hat{b}$& $\hat{\theta}$&
$\hat{a}$& $\hat{b}$& $\hat{\theta}$&
$(\hat{a},\hat{b})$& $(\hat{a},\hat{\theta})$& $(\hat{b},\hat{\theta})$&
$(\hat{a},\hat{b})$& $(\hat{a},\hat{\theta})$& $(\hat{b},\hat{\theta})$
\\\hline
30&0.3 & 0.3 & 0.2 & 0.2312 & 0.3495 & 0.2764 & -0.0688 & 0.0495 & 0.0764 & 0.1831 & 0.1624 & 0.2917 & 0.2032 & 0.2694 & 0.0666 & -0.0099 & -0.0325 & -0.0111 & -0.0478 & -0.0098 & 0.0145 \\
&0.3 & 0.3 & 0.8 & 0.4260 & 0.4057 & 0.6701 & 0.1260 & 0.1057 & -0.1299 & 0.2826 & 0.3343 & 0.2450 & 0.1475 & 0.2666 & 0.0227 & 0.0209 & -0.0567 & -0.0337 & -0.0252 & -0.0006 & 0.0009 \\
&0.3 & 0.8 & 0.8 & 0.5013 & 1.0237 & 0.6327 & 0.2013 & 0.2237 & -0.1673 & 0.4204 & 0.7124 & 0.2842 & 0.2084 & 0.5097 & 0.0221 & 0.0542 & -0.1003 & -0.0868 & -0.0678 & -0.0006 & 0.0013 \\
&0.3 & 2.0 & 0.8 & 0.6184 & 2.6309 & 0.6016 & 0.3184 & 0.6309 & -0.1984 & 0.6036 & 1.6177 & 0.3054 & 0.3035 & 1.0905 & 0.0212 & 0.1505 & -0.1517 & -0.2263 & -0.2114 & -0.0005 & 0.0016 \\
&0.8 & 0.3 & 0.2 & 0.5809 & 0.4748 & 0.3636 & -0.2191 & 0.1748 & 0.1636 & 0.3249 & 0.3277 & 0.2802 & 0.2165 & 0.3109 & 0.0246 & -0.0389 & -0.0627 & -0.0093 & -0.0454 & -0.0006 & 0.0009 \\
&0.8 & 0.8 & 0.2 & 0.6237 & 1.0119 & 0.3129 & -0.1763 & 0.2119 & 0.1129 & 0.3878 & 0.5975 & 0.2858 & 0.2706 & 0.5199 & 0.0226 & -0.0834 & -0.0715 & -0.0253 & -0.0950 & -0.0005 & 0.0011 \\
&0.8 & 0.8 & 0.8 & 0.8778 & 1.4392 & 0.7318 & 0.0778 & 0.6392 & -0.0682 & 0.5691 & 1.4455 & 0.2165 & 0.2287 & 0.8612 & 0.0201 & 0.2867 & -0.1016 & -0.1417 & -0.1093 & -0.0005 & 0.0013 \\
&0.8 & 2.0 & 0.8 & 1.0569 & 2.5558 & 0.6906 & 0.2569 & 0.5558 & -0.1094 & 0.7668 & 1.9869 & 0.2518 & 0.3058 & 1.3099 & 0.0211 & 0.4190 & -0.1612 & -0.2029 & -0.2277 & -0.0006 & 0.0016 \\
&2.0 & 0.8 & 0.2 & 1.4188 & 2.0194 & 0.3680 & -0.5812 & 1.2194 & 0.1680 & 0.7892 & 1.6963 & 0.2837 & 0.4653 & 1.3923 & 0.0238 & -0.3958 & -0.1699 & -0.0083 & -0.4121 & -0.0007 & 0.0017 \\
&2.0 & 0.8 & 0.8 & 1.7003 & 3.7261 & 0.7994 & -0.2997 & 2.9261 & -0.0006 & 1.0887 & 4.8827 & 0.1556 & 0.2899 & 2.2938 & 0.0175 & 2.4324 & -0.1342 & -0.4087 & -0.3215 & -0.0004 & 0.0011 \\
&2.0 & 2.0 & 0.2 & 1.4859 & 3.2384 & 0.3477 & -0.5141 & 1.2384 & 0.1477 & 0.8581 & 2.4121 & 0.2811 & 0.5172 & 1.8461 & 0.0237 & -0.8216 & -0.1616 & -0.0316 & -0.6129 & -0.0007 & 0.0020 \\
&2.0 & 2.0 & 0.8 & 1.8109 & 4.3873 & 0.7852 & -0.1891 & 2.3873 & -0.0148 & 1.1661 & 4.5435 & 0.1668 & 0.2964 & 1.9553 & 0.0183 & 1.3466 & -0.1566 & -0.3169 & -0.2384 & -0.0004 & 0.0013 \\
\hline
70&0.3 & 0.3 & 0.2 & 0.2399 & 0.3231 & 0.2853 & -0.0601 & 0.0231 & 0.0853 & 0.1419 & 0.1022 & 0.2722 & 0.0907 & 0.0973 & 0.0160 & -0.0033 & -0.0300 & -0.0063 & -0.0064 & -0.0003 & 0.0004 \\
&0.3 & 0.3 & 0.8 & 0.3902 & 0.3146 & 0.7248 & 0.0902 & 0.0146 & -0.0752 & 0.2621 & 0.1647 & 0.1917 & 0.0803 & 0.1328 & 0.0144 & 0.0035 & -0.0464 & -0.0065 & -0.0068 & -0.0003 & 0.0004 \\
&0.3 & 0.8 & 0.8 & 0.4856 & 0.8544 & 0.6769 & 0.1856 & 0.0544 & -0.1231 & 0.3828 & 0.3836 & 0.2402 & 0.1182 & 0.2688 & 0.0146 & 0.0221 & -0.0847 & -0.0301 & -0.0201 & -0.0003 & 0.0006 \\
&0.3 & 2.0 & 0.8 & 0.5259 & 2.1742 & 0.6744 & 0.2259 & 0.1742 & -0.1256 & 0.5052 & 1.0647 & 0.2528 & 0.1539 & 0.4957 & 0.0140 & 0.0599 & -0.1157 & -0.1037 & -0.0470 & -0.0002 & 0.0007 \\
&0.8 & 0.3 & 0.2 & 0.6355 & 0.3973 & 0.3295 & -0.1645 & 0.0973 & 0.1295 & 0.2928 & 0.2236 & 0.2567 & 0.1472 & 0.1918 & 0.0162 & -0.0339 & -0.0599 & 0.0094 & -0.0196 & -0.0003 & 0.0004 \\
&0.8 & 0.8 & 0.2 & 0.6355 & 0.9007 & 0.3141 & -0.1645 & 0.1007 & 0.1141 & 0.3419 & 0.3620 & 0.2762 & 0.1767 & 0.3151 & 0.0149 & -0.0460 & -0.0770 & -0.0019 & -0.0386 & -0.0002 & 0.0005 \\
&0.8 & 0.8 & 0.8 & 0.7904 & 0.8975 & 0.7900 & -0.0096 & 0.0975 & -0.0100 & 0.4917 & 0.5740 & 0.1388 & 0.0959 & 0.3166 & 0.0124 & 0.0150 & -0.0614 & -0.0147 & -0.0133 & -0.0002 & 0.0004 \\
&0.8 & 2.0 & 0.8 & 1.0254 & 2.0842 & 0.7278 & 0.2254 & 0.0842 & -0.0722 & 0.6799 & 1.1166 & 0.1806 & 0.1493 & 0.5956 & 0.0143 & 0.1109 & -0.1119 & -0.0564 & -0.0419 & -0.0003 & 0.0007 \\
&2.0 & 0.8 & 0.2 & 1.5296 & 1.3629 & 0.3669 & -0.4704 & 0.5629 & 0.1669 & 0.7149 & 0.9592 & 0.2649 & 0.2947 & 0.7455 & 0.0163 & -0.4104 & -0.1642 & 0.0922 & -0.1422 & -0.0003 & 0.0008 \\
&2.0 & 0.8 & 0.8 & 1.6483 & 1.9878 & 0.8239 & -0.3517 & 1.1878 & 0.0239 & 0.8482 & 1.8113 & 0.0941 & 0.1078 & 0.5861 & 0.0108 & -0.2486 & -0.0714 & 0.0081 & -0.0128 & -0.0001 & 0.0004 \\
&2.0 & 2.0 & 0.2 & 1.5872 & 2.5692 & 0.3323 & -0.4128 & 0.5692 & 0.1323 & 0.7544 & 1.3501 & 0.2726 & 0.3455 & 1.0843 & 0.0154 & -0.5507 & -0.1737 & 0.0828 & -0.2496 & -0.0003 & 0.0008 \\
&2.0 & 2.0 & 0.8 & 1.7496 & 2.9251 & 0.8141 & -0.2504 & 0.9251 & 0.0141 & 0.9717 & 2.3341 & 0.1118 & 0.1275 & 0.7419 & 0.0112 & -0.1977 & -0.0984 & -0.0111 & -0.0223 & -0.0002 & 0.0004 \\ \hline
\end{tabular}
}
\end{sidewaystable}

\begin{sidewaystable}
\caption{The average MLE's, mean of the simulated
standard errors, mean of the simulated covariances, mean of the standard errors and covariances of EM estimators
obtained using observed information matrix of the LFRG distribution for $n=100,~200$.}
{\footnotesize
\begin{tabular}{|c|ccc|ccc|ccc|ccc|ccc|ccc|ccc|ccc|} \hline
&\multicolumn{3}{|c|}{Parameter}&
\multicolumn{3}{|c|}{AE}&  \multicolumn{3}{|c|}{Bias}&
\multicolumn{3}{|c|}{Sim.std}& \multicolumn{3}{|c|}{EM.std}& \multicolumn{3}{|c|}{Sim.Cov}&
\multicolumn{3}{|c|}{EM.Cov}\\ \hline

$n$&$a$      & $b$      & $\theta$      &
$\hat{a}$& $\hat{b}$& $\hat{\theta}$&
$\hat{a}$& $\hat{b}$& $\hat{\theta}$&
$\hat{a}$& $\hat{b}$& $\hat{\theta}$&
$\hat{a}$& $\hat{b}$& $\hat{\theta}$&
$(\hat{a},\hat{b})$& $(\hat{a},\hat{\theta})$& $(\hat{b},\hat{\theta})$&
$(\hat{a},\hat{b})$& $(\hat{a},\hat{\theta})$& $(\hat{b},\hat{\theta})$
\\\hline
100&0.3 & 0.3 & 0.2 & 0.2537 & 0.3054 & 0.2839 & -0.0463 & 0.0054 & 0.0839 & 0.1320 & 0.0814 & 0.2588 & 0.0759 & 0.0781 & 0.0131 & -0.0026 & -0.0277 & -0.0038 & -0.0043 & -0.0002 & 0.0002 \\
&0.3 & 0.3 & 0.8 & 0.3762 & 0.2976 & 0.7413 & 0.0762 & -0.0024 & -0.0587 & 0.2519 & 0.1225 & 0.1723 & 0.0624 & 0.1029 & 0.0119 & -0.0007 & -0.0408 & -0.0023 & -0.0041 & -0.0002 & 0.0003 \\

&0.3 & 0.8 & 0.8 & 0.4728 & 0.8312 & 0.6848 & 0.1728 & 0.0312 & -0.1152 & 0.3765 & 0.3531 & 0.2292 & 0.0942 & 0.2086 & 0.0123 & 0.0022 & -0.0810 & -0.0146 & -0.0122 & -0.0002 & 0.0004 \\

&0.3 & 2.0 & 0.8 & 0.5387 & 2.1109 & 0.6732 & 0.2387 & 0.1109 & -0.1268 & 0.5053 & 0.8081 & 0.2519 & 0.1310 & 0.4130 & 0.0116 & 0.0802 & -0.1200 & -0.0772 & -0.0339 & -0.0002 & 0.0005 \\
&0.8 & 0.3 & 0.2 & 0.6576 & 0.3688 & 0.3068 & -0.1424 & 0.0688 & 0.1068 & 0.2679 & 0.1721 & 0.2393 & 0.1257 & 0.1584 & 0.0137 & -0.0252 & -0.0547 & 0.0098 & -0.0140 & -0.0002 & 0.0002 \\

&0.8 & 0.8 & 0.2 & 0.6682 & 0.8701 & 0.3047 & -0.1318 & 0.0701 & 0.1047 & 0.3153 & 0.3063 & 0.2518 & 0.1501 & 0.2635 & 0.0131 & -0.0439 & -0.0671 & 0.0078 & -0.0276 & -0.0002 & 0.0003 \\
&0.8 & 0.8 & 0.8 & 0.8157 & 0.8487 & 0.7876 & 0.0157 & 0.0487 & -0.0124 & 0.4638 & 0.4410 & 0.1344 & 0.0826 & 0.2662 & 0.0103 & -0.0045 & -0.0579 & -0.0068 & -0.0102 & -0.0001 & 0.0003 \\

&0.8 & 2.0 & 0.8 & 1.0023 & 1.9687 & 0.7384 & 0.2023 & -0.0313 & -0.0616 & 0.6676 & 0.8627 & 0.1772 & 0.1188 & 0.4591 & 0.0117 & 0.0506 & -0.1109 & -0.0289 & -0.0251 & -0.0002 & 0.0005 \\
&2.0 & 0.8 & 0.2 & 1.6562 & 1.1870 & 0.3221 & -0.3438 & 0.3870 & 0.1221 & 0.6666 & 0.8128 & 0.2576 & 0.2631 & 0.6426 & 0.0133 & -0.3399 & -0.1526 & 0.0861 & -0.1119 & -0.0002 & 0.0005 \\
&2.0 & 0.8 & 0.8 & 1.7137 & 1.5069 & 0.8232 & -0.2863 & 0.7069 & 0.0232 & 0.8222 & 1.3411 & 0.0883 & 0.0911 & 0.4928 & 0.0090 & -0.3701 & -0.0676 & 0.0286 & -0.0101 & -0.0001 & 0.0002 \\
&2.0 & 2.0 & 0.2 & 1.6737 & 2.3679 & 0.3248 & -0.3263 & 0.3679 & 0.1248 & 0.7410 & 1.1437 & 0.2536 & 0.2941 & 0.9118 & 0.0135 & -0.4766 & -0.1654 & 0.0849 & -0.1787 & -0.0002 & 0.0007 \\
&2.0 & 2.0 & 0.8 & 1.7123 & 2.7741 & 0.8196 & -0.2877 & 0.7741 & 0.0196 & 0.8471 & 1.7805 & 0.0970 & 0.0954 & 0.5422 & 0.0091 & -0.2904 & -0.0744 & 0.0126 & -0.0106 & -0.0001 & 0.0003 \\
\hline
200&0.3 & 0.3 & 0.2 & 0.2657 & 0.3009 & 0.2583 & -0.0343 & 0.0009 & 0.0583 & 0.1057 & 0.0548 & 0.2188 & 0.0549 & 0.0555 & 0.0096 & -0.0012 & -0.0200 & -0.0016 & -0.0022 & -0.0001 & 0.0001 \\
&0.3 & 0.3 & 0.8 & 0.3663 & 0.2904 & 0.7522 & 0.0663 & -0.0096 & -0.0478 & 0.2133 & 0.0875 & 0.1443 & 0.0405 & 0.0683 & 0.0083 & -0.0022 & -0.0295 & 0.0002 & -0.0018 & -0.0001 & 0.0001 \\

&0.3 & 0.8 & 0.8 & 0.4284 & 0.7894 & 0.7244 & 0.1284 & -0.0106 & -0.0756 & 0.3168 & 0.2194 & 0.1797 & 0.0551 & 0.1260 & 0.0085 & 0.0000 & -0.0552 & -0.0044 & -0.0042 & -0.0001 & 0.0002 \\

&0.3 & 2.0 & 0.8 & 0.4399 & 2.0580 & 0.7301 & 0.1399 & 0.0580 & -0.0699 & 0.3820 & 0.5203 & 0.1805 & 0.0675 & 0.2126 & 0.0080 & 0.0219 & -0.0664 & -0.0231 & -0.0082 & -0.0001 & 0.0002 \\
&0.8 & 0.3 & 0.2 & 0.7117 & 0.3404 & 0.2706 & -0.0883 & 0.0404 & 0.0706 & 0.2339 & 0.1310 & 0.2110 & 0.0930 & 0.1135 & 0.0097 & -0.0203 & -0.0445 & 0.0124 & -0.0075 & -0.0001 & 0.0001 \\

&0.8 & 0.8 & 0.2 & 0.7010 & 0.8308 & 0.2780 & -0.0990 & 0.0308 & 0.0780 & 0.2746 & 0.2191 & 0.2240 & 0.1088 & 0.1870 & 0.0094 & -0.0297 & -0.0539 & 0.0089 & -0.0144 & -0.0001 & 0.0002 \\
&0.8 & 0.8 & 0.8 & 0.8091 & 0.8036 & 0.7900 & 0.0091 & 0.0036 & -0.0100 & 0.4127 & 0.3043 & 0.1134 & 0.0522 & 0.1649 & 0.0073 & -0.0352 & -0.0448 & 0.0067 & -0.0036 & -0.0001 & 0.0002 \\

&0.8 & 2.0 & 0.8 & 0.9619 & 1.8875 & 0.7528 & 0.1619 & -0.1125 & -0.0472 & 0.6306 & 0.6019 & 0.1577 & 0.0774 & 0.2858 & 0.0080 & -0.0733 & -0.0966 & 0.0120 & -0.0100 & -0.0001 & 0.0002 \\
&2.0 & 0.8 & 0.2 & 1.6657 & 1.1160 & 0.3194 & -0.3343 & 0.3160 & 0.1194 & 0.6310 & 0.6689 & 0.2503 & 0.1846 & 0.4392 & 0.0095 & -0.3121 & -0.1473 & 0.1010 & -0.0538 & -0.0001 & 0.0003 \\
&2.0 & 0.8 & 0.8 & 1.7492 & 1.1671 & 0.8217 & -0.2508 & 0.3671 & 0.0217 & 0.7374 & 0.9482 & 0.0788 & 0.0614 & 0.3201 & 0.0063 & -0.4200 & -0.0551 & 0.0385 & -0.0041 & 0.0000 & 0.0001 \\
&2.0 & 2.0 & 0.2 & 1.7267 & 2.2780 & 0.2975 & -0.2733 & 0.2780 & 0.0975 & 0.6354 & 0.8433 & 0.2322 & 0.2115 & 0.6426 & 0.0097 & -0.3409 & -0.1360 & 0.0902 & -0.0912 & -0.0001 & 0.0003 \\
&2.0 & 2.0 & 0.8 & 1.8887 & 2.1718 & 0.8073 & -0.1113 & 0.1718 & 0.0073 & 0.8520 & 1.3790 & 0.0860 & 0.0688 & 0.3897 & 0.0068 & -0.5958 & -0.0688 & 0.0449 & -0.0051 & -0.0001 & 0.0001 \\ \hline

\end{tabular}
}
\end{sidewaystable}

\begin{sidewaystable}
\caption{The average MLE's, mean of the simulated
standard errors, mean of the simulated covariances, mean of the standard errors and covariances of EM estimators
obtained using observed information matrix of the LFRP distribution for $n=30,~70$.}
{\footnotesize
\begin{tabular}{|c|ccc|ccc|ccc|ccc|ccc|ccc|ccc|ccc|} \hline
&\multicolumn{3}{|c|}{Parameter}&
\multicolumn{3}{|c|}{AE}&  \multicolumn{3}{|c|}{Bias}&
\multicolumn{3}{|c|}{Sim.std}& \multicolumn{3}{|c|}{EM.std}& \multicolumn{3}{|c|}{Sim.Cov}&
\multicolumn{3}{|c|}{EM.Cov}\\ \hline

$n$&$a$      & $b$      & $\theta$      &
$\hat{a}$& $\hat{b}$& $\hat{\theta}$&
$\hat{a}$& $\hat{b}$& $\hat{\theta}$&
$\hat{a}$& $\hat{b}$& $\hat{\theta}$&
$\hat{a}$& $\hat{b}$& $\hat{\theta}$&
$(\hat{a},\hat{b})$& $(\hat{a},\hat{\theta})$& $(\hat{b},\hat{\theta})$&
$(\hat{a},\hat{b})$& $(\hat{a},\hat{\theta})$& $(\hat{b},\hat{\theta})$
\\\hline
30&0.3 & 0.3 & 0.2 & 0.2268 & 0.3270 & 0.7121 & -0.0732 & 0.0270 & 0.5121 & 0.1556 & 0.1438 & 1.0566 & 0.1343 & 0.1370 & 0.0513 & -0.0078 & -0.0801 & -0.0474 & -0.0126 & 0.0001 & -0.0001 \\
&0.3 & 0.3 & 0.8 & 0.2719 & 0.3391 & 1.1002 & -0.0281 & 0.0391 & 0.3002 & 0.1710 & 0.1811 & 1.2280 & 0.1414 & 0.1575 & 0.0870 & -0.0077 & -0.1186 & -0.0784 & -0.0149 & 0.0001 & 0.0000 \\
&0.3 & 0.8 & 0.8 & 0.3030 & 0.8565 & 0.9042 & 0.0030 & 0.0565 & 0.1042 & 0.2322 & 0.3469 & 1.1650 & 0.2027 & 0.3409 & 0.0806 & -0.0225 & -0.1346 & -0.1637 & -0.0474 & 0.0004 & -0.0006 \\
&0.3 & 2.0 & 2.0 & 0.4780 & 2.4784 & 1.4311 & 0.1780 & 0.4784 & -0.5689 & 0.4067 & 1.1848 & 1.5017 & 0.3400 & 1.0282 & 0.1036 & -0.0165 & -0.3661 & -0.9000 & -0.2416 & -0.0004 & 0.0014 \\
&0.8 & 0.3 & 0.2 & 0.6104 & 0.4477 & 0.7864 & -0.1896 & 0.1477 & 0.5864 & 0.2778 & 0.3481 & 0.9907 & 0.2310 & 0.3036 & 0.0579 & -0.0424 & -0.1655 & -0.0182 & -0.0463 & 0.0005 & -0.0006 \\
&0.8 & 0.8 & 0.2 & 0.6352 & 0.9567 & 0.7641 & -0.1648 & 0.1567 & 0.5641 & 0.3515 & 0.5701 & 0.9967 & 0.2915 & 0.5355 & 0.0548 & -0.0698 & -0.2077 & -0.0955 & -0.1071 & 0.0005 & -0.0007 \\
&0.8 & 0.8 & 0.8 & 0.7398 & 0.9608 & 1.1484 & -0.0602 & 0.1608 & 0.3484 & 0.3700 & 0.6192 & 1.2425 & 0.3099 & 0.6151 & 0.0891 & -0.0583 & -0.2793 & -0.2057 & -0.1291 & 0.0008 & -0.0013 \\
&0.8 & 2.0 & 0.8 & 0.7810 & 2.2961 & 0.9942 & -0.0190 & 0.2961 & 0.1942 & 0.4661 & 1.2299 & 1.1943 & 0.4132 & 1.1863 & 0.0830 & -0.1640 & -0.2999 & -0.5348 & -0.3451 & 0.0011 & -0.0031 \\
&2.0 & 0.8 & 0.2 & 1.5673 & 1.8453 & 0.7406 & -0.4327 & 1.0453 & 0.5406 & 0.6635 & 1.5833 & 0.9178 & 0.6530 & 2.0773 & 0.0579 & -0.4782 & -0.4028 & 0.1434 & -1.0103 & 0.0027 & -0.0081 \\
&2.0 & 0.8 & 0.8 & 1.6017 & 1.9315 & 1.2779 & -0.3983 & 1.1315 & 0.4779 & 0.6753 & 1.8262 & 1.0072 & 0.6785 & 2.3007 & 0.0995 & -0.4991 & -0.4640 & 0.0231 & -1.1994 & 0.0051 & -0.0155 \\
&2.0 & 2.0 & 0.2 & 1.5943 & 2.8921 & 0.7125 & -0.4057 & 0.8921 & 0.5125 & 0.7285 & 2.0222 & 0.9363 & 0.6927 & 2.3154 & 0.0550 & -0.6454 & -0.4251 & -0.0894 & -1.1805 & 0.0024 & -0.0075 \\
&2.0 & 2.0 & 0.8 & 1.7306 & 3.1118 & 1.1214 & -0.2694 & 1.1118 & 0.3214 & 0.7609 & 2.4931 & 1.0674 & 0.6529 & 0.9049 & 0.0925 & -0.6340 & -0.4909 & -0.3974 & -0.5327 & 0.0033 & -0.0009 \\

\hline
70&0.3 & 0.3 & 0.2 & 0.2450 & 0.3033 & 0.6867 & -0.0550 & 0.0033 & 0.4867 & 0.1212 & 0.0977 & 1.0236 & 0.0881 & 0.0855 & 0.0323 & -0.0028 & -0.0795 & -0.0354 & -0.0052 & 0.0001 & -0.0001 \\
&0.3 & 0.3 & 0.8 & 0.2929 & 0.3027 & 1.0959 & -0.0071 & 0.0027 & 0.2959 & 0.1508 & 0.1123 & 1.3605 & 0.0944 & 0.0970 & 0.0554 & -0.0013 & -0.1458 & -0.0660 & -0.0063 & 0.0001 & 0.0000 \\
&0.3 & 0.8 & 0.8 & 0.3051 & 0.7873 & 1.0921 & 0.0051 & -0.0127 & 0.2921 & 0.1883 & 0.2909 & 1.4038 & 0.1266 & 0.2070 & 0.0548 & -0.0047 & -0.1659 & -0.2215 & -0.0179 & 0.0001 & -0.0001 \\
&0.3 & 2.0 & 0.2 & 0.2391 & 1.9777 & 0.5873 & -0.0609 & -0.0223 & 0.3873 & 0.1793 & 0.4871 & 1.0635 & 0.1662 & 0.4179 & 0.0309 & -0.0178 & -0.0748 & -0.3081 & -0.0477 & 0.0002 & -0.0005 \\
&0.8 & 0.3 & 0.2 & 0.6729 & 0.3749 & 0.6282 & -0.1271 & 0.0749 & 0.4282 & 0.2281 & 0.2025 & 0.7337 & 0.1594 & 0.1919 & 0.0346 & -0.0215 & -0.1220 & 0.0065 & -0.0211 & 0.0003 & -0.0003 \\
&0.8 & 0.8 & 0.2 & 0.6444 & 0.8694 & 0.7192 & -0.1556 & 0.0694 & 0.5192 & 0.2835 & 0.3364 & 1.0158 & 0.1899 & 0.3216 & 0.0329 & -0.0294 & -0.2093 & -0.0598 & -0.0428 & 0.0002 & -0.0003 \\
&0.8 & 0.8 & 0.8 & 0.7241 & 0.8404 & 1.2165 & -0.0759 & 0.0404 & 0.4165 & 0.3056 & 0.3728 & 1.2053 & 0.1950 & 0.3582 & 0.0599 & -0.0138 & -0.2838 & -0.1392 & -0.0481 & 0.0004 & -0.0006 \\
&0.8 & 2.0 & 0.8 & 0.7739 & 2.0331 & 1.1202 & -0.0261 & 0.0331 & 0.3202 & 0.3986 & 0.7978 & 1.3623 & 0.2632 & 0.7119 & 0.0557 & -0.0168 & -0.3880 & -0.4949 & -0.1339 & 0.0005 & -0.0013 \\
&2.0 & 0.8 & 0.2 & 1.6791 & 1.2448 & 0.6285 & -0.3209 & 0.4448 & 0.4285 & 0.5079 & 0.8163 & 0.6321 & 0.3927 & 0.9781 & 0.0359 & -0.2179 & -0.2416 & 0.1172 & -0.2783 & 0.0009 & -0.0022 \\
&2.0 & 0.8 & 0.8 & 1.7447 & 1.2173 & 1.1918 & -0.2553 & 0.4173 & 0.3918 & 0.6074 & 0.9624 & 0.9186 & 0.4071 & 1.1315 & 0.0635 & -0.2875 & -0.4616 & 0.2037 & -0.3332 & 0.0018 & -0.0045 \\
&2.0 & 2.0 & 0.2 & 1.6687 & 2.4317 & 0.6375 & -0.3313 & 0.4317 & 0.4375 & 0.5900 & 1.2990 & 0.7771 & 0.4591 & 1.4380 & 0.0346 & -0.3994 & -0.3235 & 0.0663 & -0.4980 & 0.0011 & -0.0034 \\
&2.0 & 2.0 & 0.8 & 1.6819 & 2.3178 & 1.3846 & -0.3181 & 0.3178 & 0.5846 & 0.7056 & 1.3202 & 1.2106 & 0.4633 & 1.5068 & 0.0644 & -0.2284 & -0.6967 & -0.1234 & -0.5240 & 0.0021 & -0.0063 \\\hline
\end{tabular}
}
\end{sidewaystable}

\begin{sidewaystable}
\caption{The average MLE's, mean of the simulated
standard errors, mean of the simulated covariances, mean of the standard errors and covariances of EM estimators
obtained using observed information matrix of the LFRP distribution for $n=100,~200$.}
{\footnotesize
\begin{tabular}{|c|ccc|ccc|ccc|ccc|ccc|ccc|ccc|ccc|} \hline
&\multicolumn{3}{|c|}{Parameter}&
\multicolumn{3}{|c|}{AE}&  \multicolumn{3}{|c|}{Bias}&
\multicolumn{3}{|c|}{Sim.std}& \multicolumn{3}{|c|}{EM.std}& \multicolumn{3}{|c|}{Sim.Cov}&
\multicolumn{3}{|c|}{EM.Cov}\\ \hline

$n$&$a$      & $b$      & $\theta$      &
$\hat{a}$& $\hat{b}$& $\hat{\theta}$&
$\hat{a}$& $\hat{b}$& $\hat{\theta}$&
$\hat{a}$& $\hat{b}$& $\hat{\theta}$&
$\hat{a}$& $\hat{b}$& $\hat{\theta}$&
$(\hat{a},\hat{b})$& $(\hat{a},\hat{\theta})$& $(\hat{b},\hat{\theta})$&
$(\hat{a},\hat{b})$& $(\hat{a},\hat{\theta})$& $(\hat{b},\hat{\theta})$
\\\hline
100&0.3 & 0.3 & 0.2 & 0.2510 & 0.2986 & 0.6784 & -0.0490 & -0.0014 & 0.4784 & 0.1178 & 0.0865 & 0.9771 & 0.0746 & 0.0716 & 0.0268 & -0.0020 & -0.0795 & -0.0318 & -0.0037 & 0.0000 & 0.0000 \\
&0.3 & 0.3 & 0.8 & 0.2986 & 0.2921 & 1.0949 & -0.0014 & -0.0079 & 0.2949 & 0.1440 & 0.0937 & 1.3350 & 0.0791 & 0.0799 & 0.0467 & 0.0003 & -0.1412 & -0.0638 & -0.0043 & 0.0001 & 0.0000 \\
&0.3 & 0.8 & 0.8 & 0.3192 & 0.7828 & 0.9748 & 0.0192 & -0.0172 & 0.1748 & 0.1760 & 0.2206 & 1.3041 & 0.1088 & 0.1745 & 0.0447 & 0.0015 & -0.1525 & -0.1673 & -0.0131 & 0.0001 & -0.0001 \\
&0.3 & 2.0 & 2.0 & 0.4094 & 2.1630 & 1.8918 & 0.1094 & 0.1630 & -0.1082 & 0.2906 & 0.7813 & 1.5306 & 0.1582 & 0.4815 & 0.0644 & 0.0505 & -0.3126 & -0.8497 & -0.0530 & -0.0001 & 0.0003 \\
&0.8 & 0.3 & 0.2 & 0.6932 & 0.3500 & 0.5920 & -0.1068 & 0.0500 & 0.3920 & 0.2146 & 0.1621 & 0.6884 & 0.1348 & 0.1575 & 0.0283 & -0.0181 & -0.1137 & 0.0155 & -0.0147 & 0.0002 & -0.0002 \\
&0.8 & 0.8 & 0.2 & 0.6675 & 0.8306 & 0.7101 & -0.1325 & 0.0306 & 0.5101 & 0.2709 & 0.2695 & 0.9097 & 0.1603 & 0.2665 & 0.0279 & -0.0232 & -0.1916 & -0.0330 & -0.0301 & 0.0002 & -0.0003 \\
&0.8 & 0.8 & 0.8 & 0.7297 & 0.8054 & 1.2355 & -0.0703 & 0.0054 & 0.4355 & 0.3084 & 0.2986 & 1.3018 & 0.1645 & 0.2952 & 0.0495 & -0.0005 & -0.3281 & -0.1313 & -0.0337 & 0.0003 & -0.0004 \\
&0.8 & 2.0 & 0.8 & 0.8036 & 1.9297 & 1.1362 & 0.0036 & -0.0703 & 0.3362 & 0.3907 & 0.6520 & 1.3909 & 0.2217 & 0.5834 & 0.0470 & -0.0098 & -0.4034 & -0.3985 & -0.0929 & 0.0003 & -0.0009 \\
&2.0 & 0.8 & 0.2 & 1.6990 & 1.1625 & 0.5828 & -0.3010 & 0.3625 & 0.3828 & 0.4652 & 0.7224 & 0.6065 & 0.3287 & 0.7825 & 0.0291 & -0.1887 & -0.2165 & 0.0863 & -0.1879 & 0.0006 & -0.0015 \\
&2.0 & 0.8 & 0.8 & 1.7989 & 1.0929 & 1.1276 & -0.2011 & 0.2929 & 0.3276 & 0.5479 & 0.8090 & 0.8144 & 0.3446 & 0.9113 & 0.0531 & -0.2286 & -0.3624 & 0.1468 & -0.2279 & 0.0014 & -0.0035 \\
&2.0 & 2.0 & 0.2 & 1.7125 & 2.3896 & 0.5979 & -0.2875 & 0.3896 & 0.3979 & 0.5715 & 1.0448 & 0.8080 & 0.3945 & 1.2397 & 0.0273 & -0.3206 & -0.3604 & 0.1038 & -0.3719 & 0.0008 & -0.0024 \\
&2.0 & 2.0 & 0.8 & 1.7986 & 2.2138 & 1.2414 & -0.2014 & 0.2138 & 0.4414 & 0.6837 & 1.1497 & 1.2379 & 0.4381 & 1.4671 & 0.0513 & -0.2172 & -0.6871 & -0.1096 & -0.5110 & 0.0021 & -0.0068 \\

\hline
200&0.3 & 0.3 & 0.2 & 0.2534 & 0.2956 & 0.5923 & -0.0466 & -0.0044 & 0.3923 & 0.0977 & 0.0593 & 0.8726 & 0.0530 & 0.0501 & 0.0187 & -0.0004 & -0.0608 & -0.0223 & -0.0018 & 0.0000 & 0.0000 \\
&0.3 & 0.3 & 0.8 & 0.3135 & 0.2898 & 0.9469 & 0.0135 & -0.0102 & 0.1469 & 0.1281 & 0.0691 & 1.1740 & 0.0573 & 0.0568 & 0.0318 & 0.0009 & -0.1160 & -0.0418 & -0.0022 & 0.0000 & 0.0000 \\
&0.3 & 0.8 & 0.8 & 0.3107 & 0.7721 & 0.9987 & 0.0107 & -0.0279 & 0.1987 & 0.1452 & 0.1808 & 1.2164 & 0.0749 & 0.1204 & 0.0330 & 0.0051 & -0.1258 & -0.1452 & -0.0062 & 0.0001 & -0.0001 \\
&0.3 & 2.0 & 2.0 & 0.3761 & 2.0328 & 2.0201 & 0.0761 & 0.0328 & 0.0201 & 0.2262 & 0.6176 & 1.4304 & 0.1014 & 0.3119 & 0.0476 & 0.0588 & -0.2343 & -0.7048 & -0.0219 & -0.0001 & 0.0002 \\
&0.8 & 0.3 & 0.2 & 0.7175 & 0.3335 & 0.4667 & -0.0825 & 0.0335 & 0.2667 & 0.1747 & 0.1137 & 0.5576 & 0.0980 & 0.1110 & 0.0182 & -0.0119 & -0.0801 & 0.0152 & -0.0077 & 0.0001 & -0.0001 \\
&0.8 & 0.8 & 0.2 & 0.6773 & 0.8208 & 0.6577 & -0.1227 & 0.0208 & 0.4577 & 0.2358 & 0.1984 & 0.8993 & 0.1147 & 0.1885 & 0.0192 & -0.0107 & -0.1717 & -0.0352 & -0.0154 & 0.0001 & -0.0002 \\
&0.8 & 0.8 & 0.8 & 0.7652 & 0.7624 & 1.1273 & -0.0348 & -0.0376 & 0.3273 & 0.2889 & 0.2184 & 1.1571 & 0.1183 & 0.2054 & 0.0345 & -0.0061 & -0.2780 & -0.0721 & -0.0170 & 0.0002 & -0.0003 \\
&0.8 & 2.0 & 0.8 & 0.8263 & 1.9088 & 0.9758 & 0.0263 & -0.0912 & 0.1758 & 0.3345 & 0.4817 & 1.1864 & 0.1594 & 0.4126 & 0.0325 & 0.0003 & -0.2993 & -0.2746 & -0.0476 & 0.0002 & -0.0007 \\
&2.0 & 0.8 & 0.2 & 1.8066 & 1.0105 & 0.4578 & -0.1934 & 0.2105 & 0.2578 & 0.3859 & 0.5031 & 0.4822 & 0.2401 & 0.5462 & 0.0189 & -0.1146 & -0.1531 & 0.0693 & -0.0967 & 0.0003 & -0.0007 \\
&2.0 & 0.8 & 0.8 & 1.8740 & 0.9268 & 1.0597 & -0.1260 & 0.1268 & 0.2597 & 0.5460 & 0.6460 & 0.7505 & 0.2430 & 0.6068 & 0.0367 & -0.2250 & -0.3688 & 0.2097 & -0.1064 & 0.0006 & -0.0016 \\
&2.0 & 2.0 & 0.2 & 1.7485 & 2.2580 & 0.5199 & -0.2515 & 0.2580 & 0.3199 & 0.4688 & 0.7522 & 0.5966 & 0.2802 & 0.8341 & 0.0191 & -0.2029 & -0.2344 & 0.1093 & -0.1789 & 0.0004 & -0.0012 \\
&2.0 & 2.0 & 0.8 & 1.7977 & 2.1057 & 1.2005 & -0.2023 & 0.1057 & 0.4005 & 0.6364 & 0.8473 & 0.9882 & 0.3027 & 1.0051 & 0.0372 & -0.2321 & -0.5396 & 0.0745 & -0.2411 & 0.0011 & -0.0035 \\ \hline

\end{tabular}
}
\end{sidewaystable}

\section{ Data analysis}
\label{se.app}
In this section, we analyze a real data set to demonstrate the
performance of the LFRPS distributions in practice. The data set studied by
\cite{ab-ab-qa-94},
 which represent the lifetime in days of 40 patients suffering from leukemia from one of the Ministry of Health Hospitals in Saudi Arabia. The TTT plot in Fig. \ref{fig.TTT} states that this data posses an increasing hazard rate function and also indicates that appropriateness of LFRPS to fit this data. For this data set, we compare the results of the fits of the LFRPS, GLFR, LFR, EG, RG, Rayligh and exponential distributions, where the LFR, EG, RG, Rayligh and exponential are submodels of the LFRPS distributions and the pdf of the GLFR distribution is given by
$$
f(x;a,b,\theta)=\theta(a+bx)e^{-ax-\frac{b}{2}x^2}(1-e^{-ax-\frac{b}{2}x^2})^{\theta-1},~~ x\geq0.
$$
The main reason for the use of the LFRG distribution in the LFRPS family of distributions, is that it attains the best fit to data in the class of LFRPS distributions among LFRG, LFRP, LFRB and LFRL distributions. Table 5 gives the MLEs of the parameters, the standard errors of the MLEs, a $100(1-\alpha)$ confidence intervals for the $a$, $b$ and $\theta$, -2log(likelihood), AIC, AICC, BIC, K–S statistic with its respective p-value, AD and CM statistics.

The AIC (Akaike Information Criterion) is given by -2log(likelihood)+2$p$, where $p$ is the number of parameters index to the model. The p-value according to the K–S statistic is computed approximately using the chi-square statistic. The AD (Anderson-Darling) and CM (Cramer von Mises) are given respectively by $$AD=-n-\sum_{i=1}^{n}\frac{2i-1}{n}\left [\ln F(y_{i})+\ln (1-F(y_{n-i+1}))\right ],$$ and $$CM=\frac{1}{12n}+\sum_{i=1}^{n}\left [\frac{2i-1}{2n}-F(y_{i})\right ]^2.$$ The values of these criteria in Table 5 emphasize that the LFRG gives the best fit to the leukemia data in comparing with GLFR, LFR, EG, RG, Rayligh and exponential distributions.

Plots of the estimated pdf, cdf and survival function with the qq-plot of fitted distributions are given in Fig. \ref{fig.hex} and Fig. \ref{fig.qex}, respectively. These plots suggest that the LFRG distribution is superior to the other distributions in terms of model fitting.

For parametric comparisons, we have used the likelihood ratio (LR) test statistics, $\Lambda_{H_{0}}=2({{\mathcal L}}_{H_1}-{{\mathcal L}}_{H_0})$, to test the null
hypotheses against the alternative one mentioned above. Table 6 gives the null hypothesis
$H_0$, the value of log-likelihood function under $H_0$, ${{\mathcal L}}_{H_0}$, the value of the likelihood ratio test statistics, $\Lambda_{H_{0}}$, the degree of freedom of $\Lambda_{H_{0}}$, df, and the corresponding p-value. From the p-values it is clear that we reject $H_{0}: LFR (\theta=0)$ at level of significance $\alpha \ge 0.0018$, $H_{0}: EG (b=0)$ at level of significance $\alpha \ge 0.1987$,
$H_{0}: RG (a=0)$ at level of significance $\alpha \ge 0.1844$, $H_{0}: Rayleigh (a=\theta=0)$ at level of significance $\alpha \ge 0.019$ and
$H_{0}: Exp (b=\theta=0)$ at any level of significance.
This concludes that the LFRG distribution
is the best among all distributions used here to fit the current data set.

\begin{figure}[h]
\centering
\includegraphics[scale=0.40]{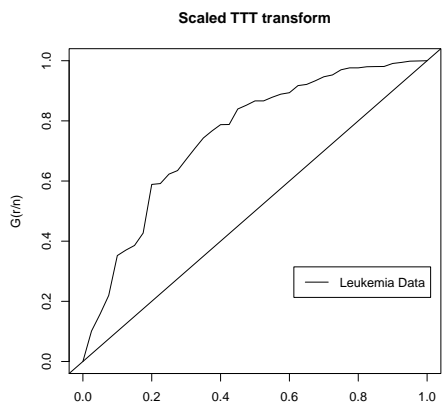}
\vspace{-0.8cm}
\caption[]{\label{fig.TTT}TTT plot for the Leukemia data.}
\end{figure}

\begin{sidewaystable}
\caption{The MLEs, standard errors of the MLEs and 95\% confidence intervals for parameters of the LFRG, GLFR, LFR, EG, RG, Rayligh and exponential models, and -2log L, AIC, AICC, BIC, K-S, P-values of the K-S statistic, AD and CM statistics for the Leukemia data.}
{\normalsize
\begin{tabular}{|c|ccc|cccccccc|} \hline
Model & MLEs & S.E.  & C.I. &  -2log L & AIC & AICC & BIC & K-S & P-value & AD & CM  \\ \hline
 &$\hat{a}$=7.73E-4&5.3E-4&(-0.0003,0.0018) &&&&&&&&\\
LFRG &$\hat{b}$=2.05E-6&4.3E-7&(1.3E-6,2.9E-6)& 604.0 & 610.0 & 610.6 & 615.0 & 0.0831 & 0.9450 & 0.4381 & 0.1304 \\
&$\hat{\theta}$=-8.448&2.1E-8&(-8.448,-8.448)&&&&&&&&\\  \hline
&$\hat{a}$=2.1E-4&3.6E-4&(-0.0005,0.0009)&&&&&&&& \\
GLFR &$\hat{b}$=1.4E-6&6.2E-7&(1.7E-7,2.6E-6)&610.7 & 616.7 & 617.3 & 621.7 & 0.1481 & 0.3440 & 1.1053 & 0.2687\\
&$\hat{\theta}$=1.5528&0.6108&(0.3183,2.7873)&&&&&&&&\\ \hline
LFR&$\hat{a}$=1E-8&8.1E-9&(-5.8E-9,2.6E-8)& 612.0 & 616.0 & 612.3 & 619.4 & 0.1659 & 0.2210 & 1.5019 & 0.3719\\
&$\hat{b}$=1.3E-6&1.1E-7&(1.1E-6,1.5E-6)&&&&&&&& \\ \hline
EG&$\hat{a}$=0.00338&0.0005&(0.0024,0.0044)& 607.5 & 611.5 & 611.8 & 614.8 & 0.0882 & 0.9150 & 0.6151 & 0.1484\\
&$\hat{\theta}$=-48.64&34.028&(-117.41,20.130)&&&&&&&&\\ \hline
RG&$\hat{b}$=2.48E-6&2.5E-7&(2E-6,3E-6)& 607.1 & 611.1 & 611.5 & 614.1 & 0.0959 & 0.8550 & 0.6543 & 0.1446\\
&$\hat{\theta}$=-3.57&0.09136&(-3.749,-3.391)&&&&&&&&\\ \hline
Rayleigh&$\hat{b}$=1.3E-6&1.5E-8&(1.3E-6,1.35E-6)& 612.0 & 614.0 & 614.1 & 615.7 & 0.1659 & 0.2210 & 1.5018 & 0.3719\\ \hline
Exp&$\hat{a}$=8.8E-4&1.4E-4&(0.0006,0.0012)& 642.9 & 644.9 & 645.0 & 646.6 & 0.3031 & 0.0013 & 1.2102 & 5.7083\\ \hline
\end{tabular}
}
\end{sidewaystable}


\begin{table}[]
\caption{Null hypothesis $H_0$, the value of log-likelihood function under $H_0$; ${{\mathcal L}}_{H_0}$ , the value of the likelihood ratio test statistics, ${\Lambda }_{H_0}$ , the degree of freedom of ${\Lambda }_{H_0}$, df, the corresponding p-value for Leukemia data..}
\begin{tabular}{|l|lllll|} \hline
Model & $H_0$ & ${{\mathcal L}}_{H_0}$ & ${\Lambda }_{H_0}$ & df & P-value \\ \hline
LFR & $\theta =0$ & 306 & 8 & 1 & 0.0018 \\ \hline
EG & $b=0$ & 303.75 & 3 & 1 & 0.1987 \\ \hline
RG & $a=0$ & 303.55 & 3.1 & 1 & 0.1844 \\ \hline
Rayleigh & $\theta =0$,  $a=0$ & 306 & 8 & 2 & 0.0190 \\ \hline
Exp & $\theta =0$,  $b=0$ & 321.45 & 38.9 & 2 & 8.3E-19 \\ \hline
\textbf{LFRG} & ${{\mathcal L}}_{H_1}$=& \textbf{302} & -- & -- & -- \\ \hline
\end{tabular}
\end{table}

\begin{figure}[]
\centering
\includegraphics[scale=0.28]{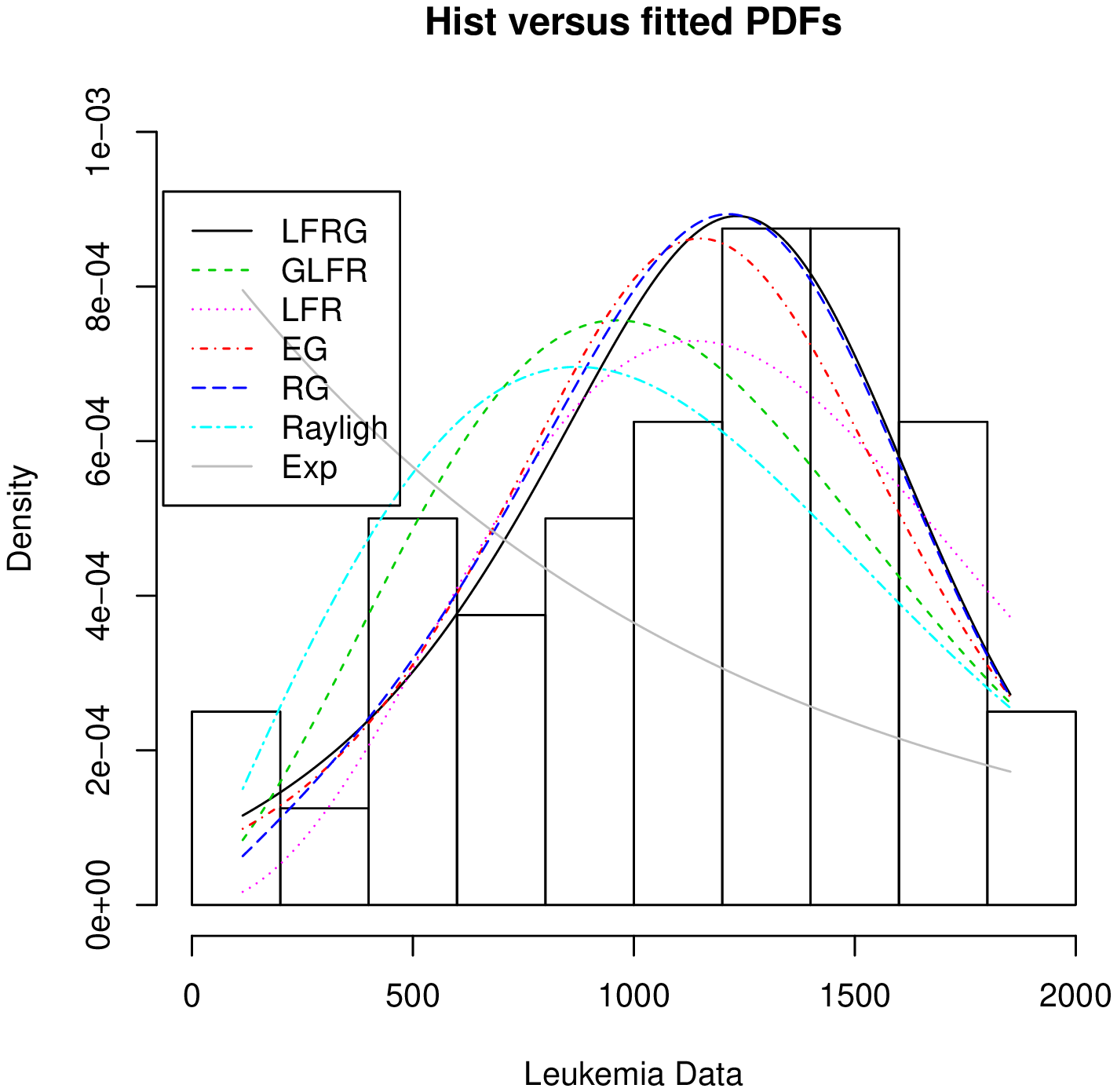}
\includegraphics[scale=0.28]{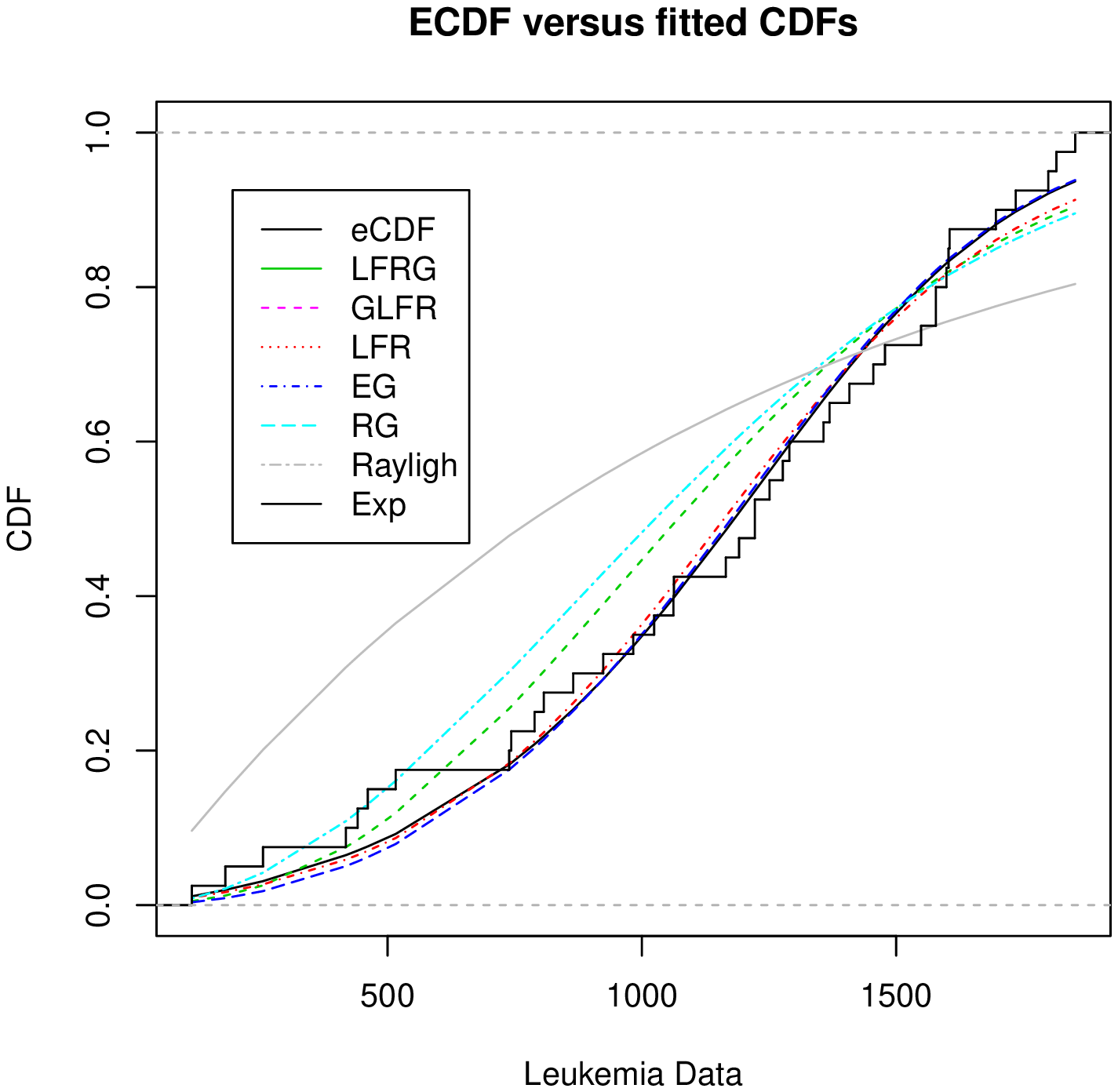}
\vspace{-0.4cm}
\caption[]{\label{fig.hex} The empirical and fitted densities, distribution functions and survival functions of the LFRG, GLFR, LFR, EG, RG, Rayligh and exponential models for the Leukemia data.}
\end{figure}

\begin{figure}[]
\centering
\includegraphics[scale=0.28]{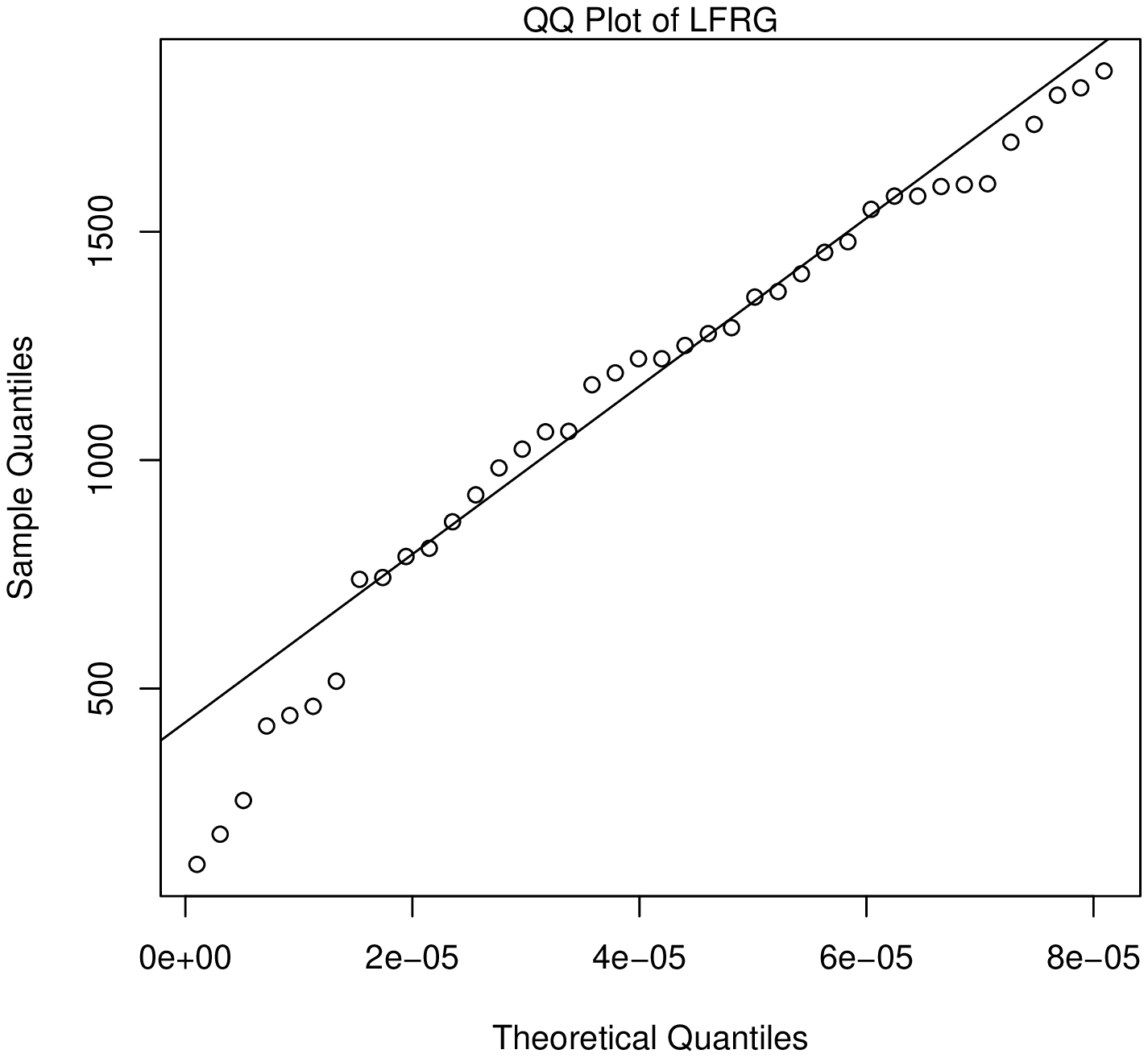}
\includegraphics[scale=0.28]{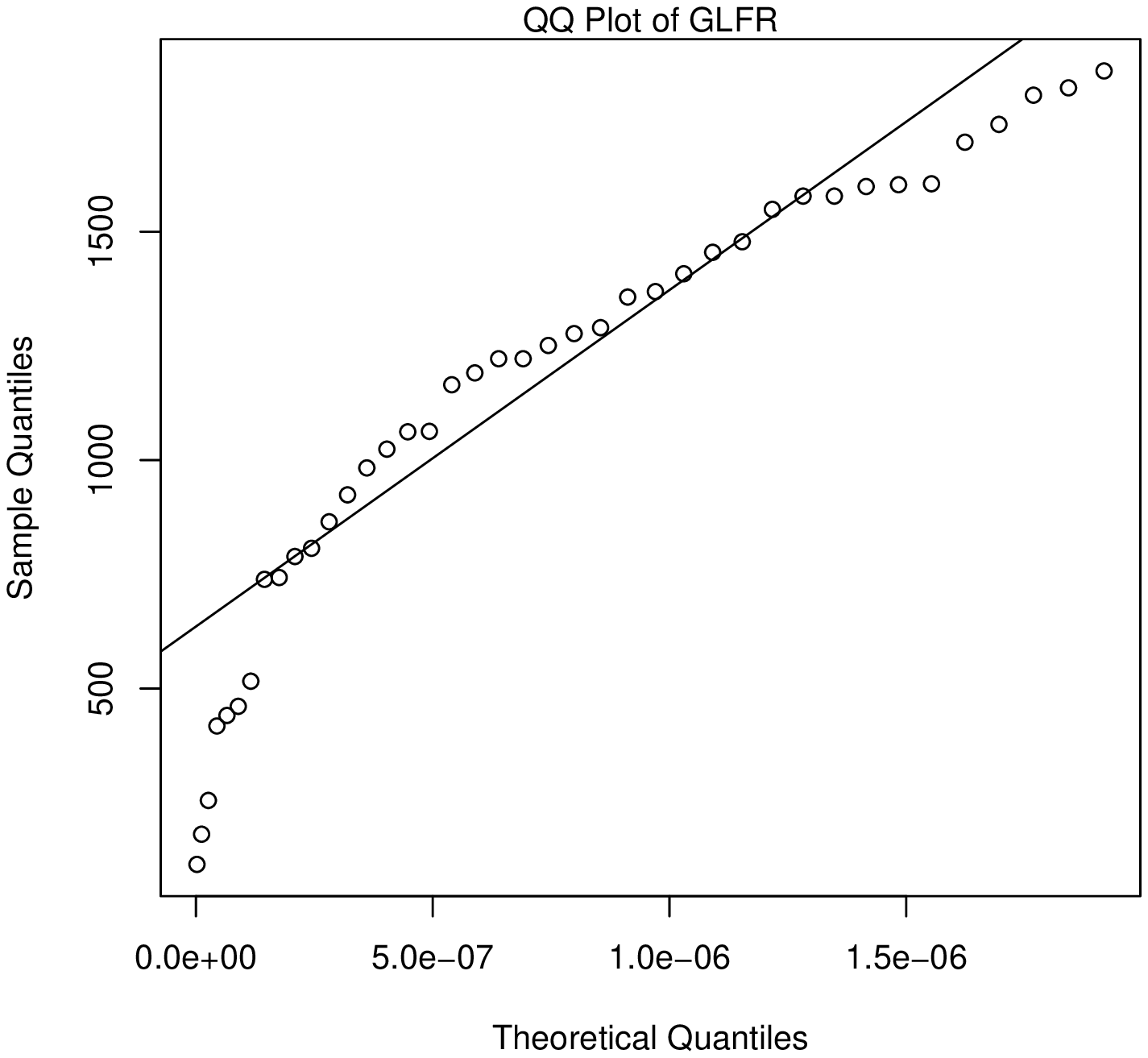}
\includegraphics[scale=0.28]{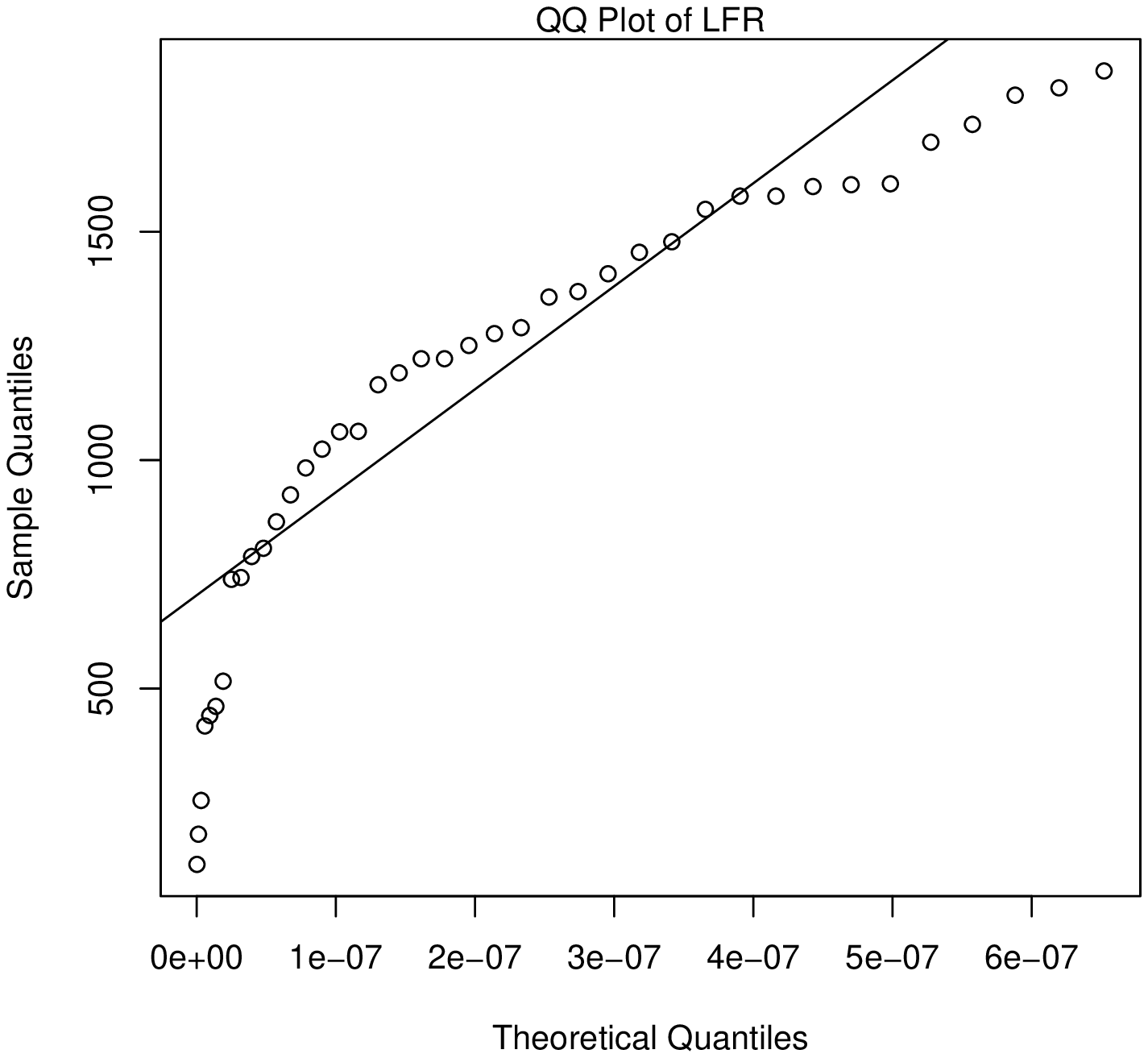}
\includegraphics[scale=0.28]{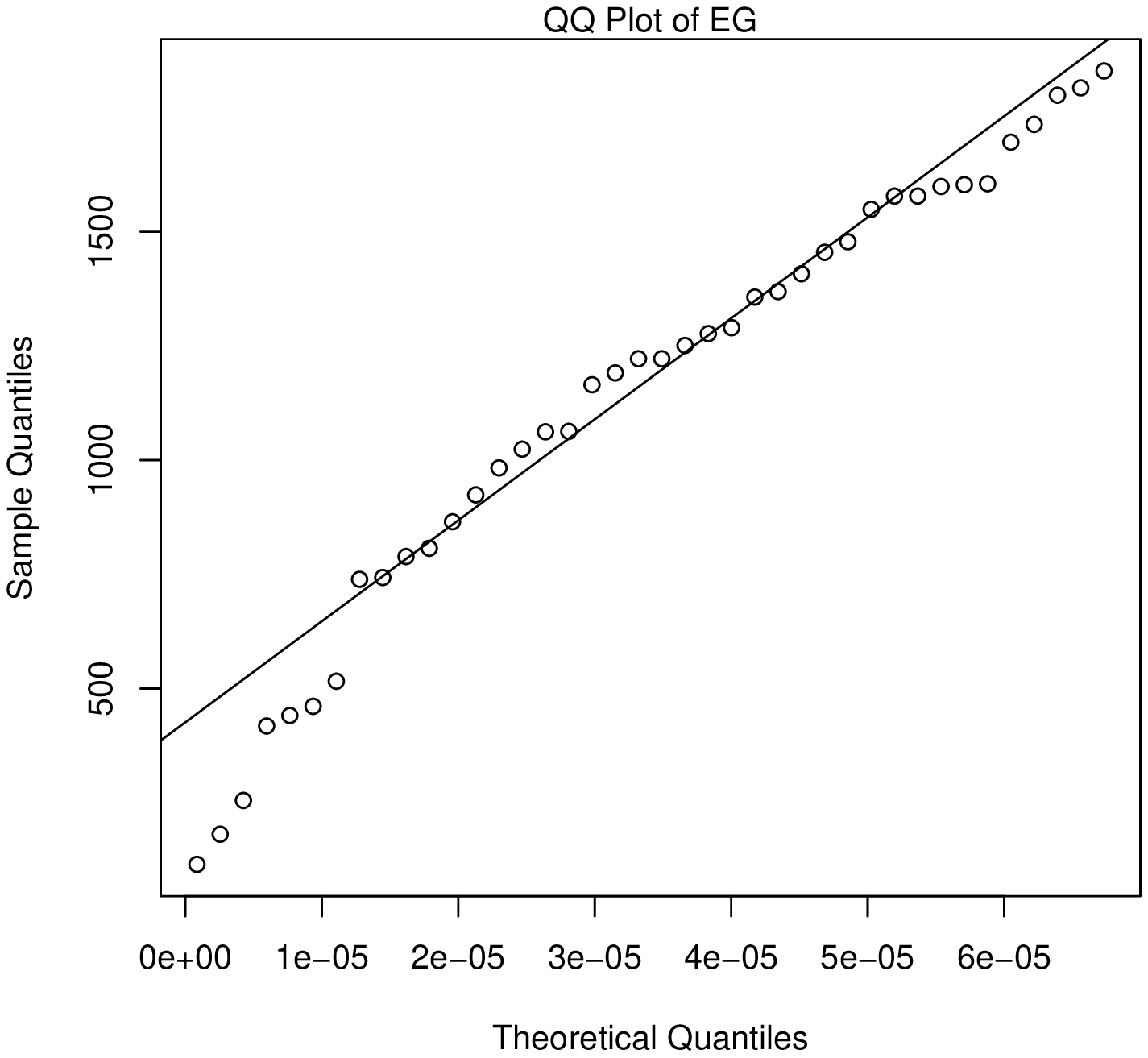}
\includegraphics[scale=0.28]{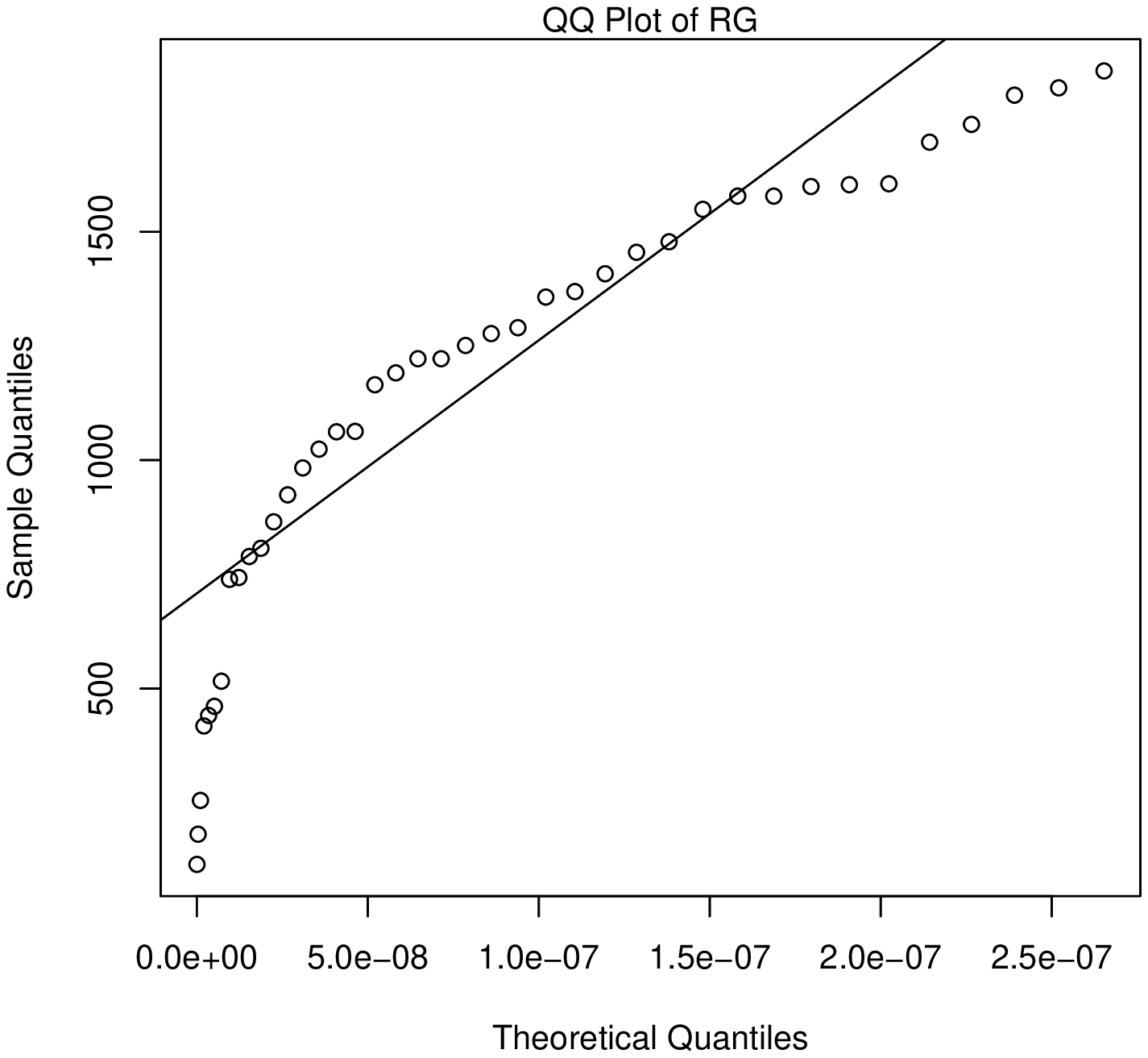}
\includegraphics[scale=0.28]{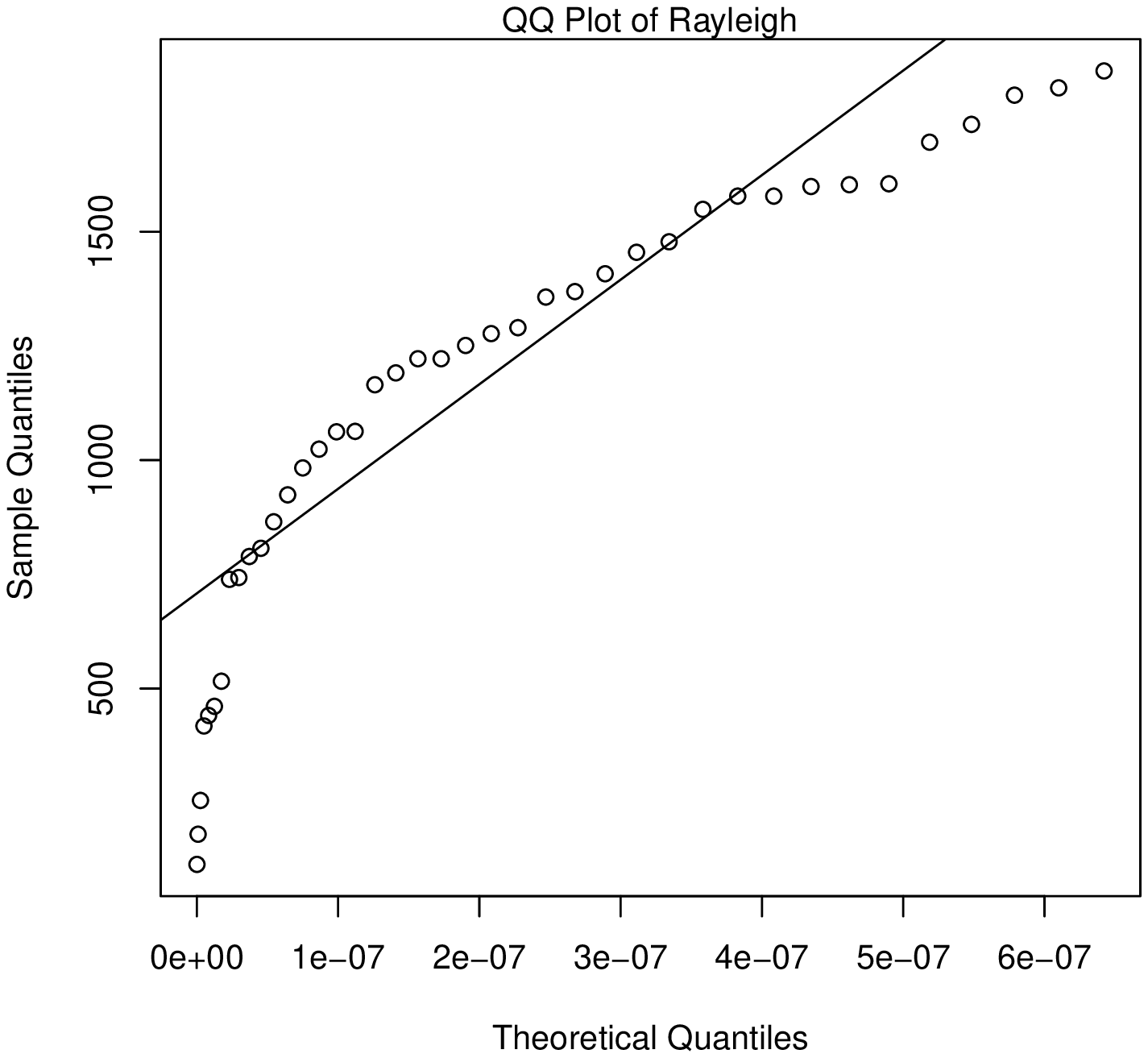}
\vspace{-0.4cm}
\caption[]{\label{fig.qex} QQ plots of the LFRG, GLFR, LFR, EG, RG and Rayligh models.}
\end{figure}

\newpage
‎\section{Conclusion}
\label{se.con}
We introduce a new class of lifetime distributions called the linear failure rate-power series (LFRPS) class of distributions, which generalizes the linear failure rate (LFR) distribution and is obtained by compounding the LFR distribution and power series (PS) class of distributions. This new class of distributions contains some new distributions such as linear failure rate geometric (LFRG) distribution, linear failure rate Poisson (LFRP) distribution, linear failure rate logarithmic (LFRL) distribution, linear failure rate binomial (LFRB) distribution and Raylight-power series (RPS) class of distributions. Some former works such as exponential-power series (EPS) class of distributions
\citep{ch-ga-09},
exponential geometric (EG) distribution
\citep{ad-lo-98},
 exponential Poisson (EP) distribution
  \citep{kus-07},
 and exponential logarithmic (EL) distribution
 \citep{ta-re-08}
 are special cases of the new proposed model.

The ability of the new proposed model is in covering five possible hazard rate function i.e., increasing, decreasing, upside-down bathtub (unimodal), bathtub and increasing-decreasing-increasing shaped. Several properties of the LFRPS distributions such as
moments, maximum likelihood estimation procedure via an EM-algorithm and inference for a large sample, are discussed in this paper. In order to show the flexibility and potentiality of the new class of distributions, the fitted results of the new class of distributions and some its submodels are compared using a real data set.
‎
\section*{Appendix A}
Let
$A_{2i}=\frac{C''\left(\theta p_i\right)}{C'\left(\theta p_i\right)}$ and $A_{3i}=\frac{C'''\left(\theta p_i\right)}{C'\left(\theta p_i\right)}$. Then
\begin{eqnarray*}
&&I_{aa}=-\frac{\partial^2{\ell }_n}{\partial a^2}
=\sum^n_{i=1}{\frac{1}{{\left(a+bx_i\right)}^2}}
-\theta\sum^n_{i=1} x^2_ip_iA_{2i}
-\theta^2\sum^n_{i=1}x^2_ip^2_iA_{3i}
+\theta^2\sum^n_{i=1}x^2_ip_iA_{2i}^2\\
&& I_{ab}=I_{ba}=-\frac{\partial^2{\ell }_n}{\partial b\partial a}
=\sum^n_{i=1}{\frac{x_i}{{\left(a+bx_i\right)}^2}}
+\frac{\theta}{2}\sum^n_{i=1} x^3_ip_iA_{2i}
+\frac{\theta^2}{2}\sum^n_{i=1}x^3_ip^2_iA_{3i}
-\frac{\theta^2}{2}\sum^n_{i=1}x^3_ip_iA_{2i}^2\\
&& I_{a\theta }=I_{\theta a}=-\frac{\partial^2{\ell }_n}{\partial \theta\partial a}=
\sum^n_{i=1}x_ip_iA_{2i}
+\theta\sum^n_{i=1} x_ip^2_iA_{3i}
-\theta\sum^n_{i=1} x_ip^2_iA_{2i}^2
\end{eqnarray*}
\begin{eqnarray*}
&&I_{bb}=-\frac{\partial^2{\ell }_n}{\partial b^2}
=\sum^n_{i=1}{\frac{x^2_i}{{\left(a+bx_i\right)}^2}}
-\frac{\theta}{4}\sum^n_{i=1} x^4_ip_iA_{2i}
-\frac{\theta^2}{4}\sum^n_{i=1}x^4_ip^2_iA_{3i}
+\frac{\theta^2}{4}\sum^n_{i=1}x^4_ip_iA_{2i}^2\\
&&I_{b\theta }=I_{\theta b}=\frac{\partial^2{\ell }_n}{\partial \theta \partial b}
=-\sum^n_{i=1}x_ip_iA_{2i}
-\theta\sum^n_{i=1} x_ip^2_iA_{3i}
-\theta \sum^n_{i=1} x_ip^2_iA_{2i}^2
\\
&&
I_{\theta \theta }=\frac{\partial^2{\ell }_n}{\partial \theta^2}
=\frac{n}{{\theta }^2}
+\sum^n_{i=1}p_iA_{3i}
-\sum^n_{i=1}p_iA_{2i}^2
-\frac{nC''(\theta )}{C(\theta )}+\frac{n{\left(C'\left(\theta \right)\right)}^2}{{\left(C\left(\theta \right)\right)}^2}\hspace{3.5cm}
\end{eqnarray*}

\section*{Appendix B}

\subsection*{B.1}
Let
\[{\rm g}_1\left(a;b,\theta,{\boldsymbol x}\right)=\frac{\partial l_n}{\partial a}=\sum^n_{i=1}{\frac{1}{a+bx_i}}-n\bar{x}-w_1\left(a;b,\theta,{\boldsymbol x}\right)\]
where
$w_1\left(a;b,\theta,{\boldsymbol x}\right)=\sum^n_{i=1}\frac{\theta x_iu^a_iv^b_iC''\left(\theta u^a_iv^b_i\right)}{C'\left(\theta u^a_iv^b_i\right)}=-[\frac{\partial}{\partial a}\sum^n_{i=1}\log  \left(C'\left(\theta p_i\right)\right)]$, $u_i=e^{-x_i}$,  $v_i=e^{-\frac{1}{2}x^2_i}$, and $p_i=\exp(-ax_i-\frac{b}{2}x^2_i)$.

\noindent i. If $\theta>0$, then, $w_1\left(a;b,\theta,{\boldsymbol x}\right)$ is strictly decreasing in $a$ and
\[{\mathop{\lim }_{a\rightarrow 0} w_1\left(a;b,\theta,{\boldsymbol x}\right)}=k_1>0,\ \ \ \ \ \ \ \ \ \ \ {\mathop{\lim}_{a\rightarrow\infty} w_1\left(a;b,\theta,{\boldsymbol x}\right)}=0.\]
Therefore,
\[{\mathop{\lim }_{a\rightarrow 0}{\rm g}_1\left(a;b,\theta,{\boldsymbol x}\right)}=\infty,\ \ \ \ \ \ \ \ \ \ \ \ \ \ \ {\mathop{\lim }_{a\rightarrow\infty} {{\rm g}}_1\left(a;b,\theta,{\boldsymbol x}\right)}=-n\bar{x}<0,\]
and
\[{\rm g}_1\left(a;b,\theta,{\boldsymbol x}\right)<\sum^n_{i=1}\frac{1}{a+bx_i}-n\bar{x}<\frac{n}{a+bx_{(1)}}
-n\bar{x},\]
\[{{\rm g}}_1\left(a;b,\theta,{\boldsymbol x}\right)>\sum^n_{i=1}{\frac{1}{a+bx_i}}-n\bar{x}-k_1>\frac{n}{a+bx_{(n)}}
-n\bar{x}-k_1.\]
Therefore, ${{\rm g}}_1\left(a;b,\theta,{\boldsymbol x}\right)<0$ when
$\frac{n}{a+bx_{(1)}}-n\bar{x}<0$, and ${{\rm g}}_1\left(a;b,\theta,{\boldsymbol x}\right)>0$ when
$\frac{n}{a+bx_{(n)}}-n\bar{x}-k_1>0$. Hence, the proof is completed.

\noindent ii.  If $\theta<0$, then, $w_1\left(a;b,\theta,{\boldsymbol x}\right)$ is strictly increasing in $a$ and
\[{\mathop{\lim }_{a\rightarrow 0} w_1\left(a;b,\theta,{\boldsymbol x}\right)}=k_1<0,\ \ \ \ \ \ \ \ \ \ \ {\mathop{\lim}_{a\rightarrow\infty} w_1\left(a;b,\theta,{\boldsymbol x}\right)}=0.\]
Therefore,
\[{\mathop{\lim }_{a\rightarrow 0}{\rm g}_1\left(a;b,\theta,{\boldsymbol x}\right)}=\infty,\ \ \ \ \ \ \ \ \ \ \ \ \ \ \ {\mathop{\lim }_{a\rightarrow\infty} {{\rm g}}_1\left(a;b,\theta,{\boldsymbol x}\right)}=-n\bar{x}<0,\]
and
\[{\rm g}_1\left(a;b,\theta,{\boldsymbol x}\right)>\sum^n_{i=1}\frac{1}{a+bx_i}-n\bar{x}>\frac{n}{a+bx_{(n)}}
-n\bar{x},\]
\[{{\rm g}}_1\left(a;b,\theta,{\boldsymbol x}\right)<\sum^n_{i=1}{\frac{1}{a+bx_i}}-n\bar{x}-k_1<\frac{n}{a+bx_{(1)}}
-n\bar{x}-k_1.\]
Therefore, ${{\rm g}}_1\left(a;b,\theta,{\boldsymbol x}\right)>0$ when
$\frac{n}{a+bx_{(n)}}-n\bar{x}>0$, and ${{\rm g}}_1\left(a;b,\theta,{\boldsymbol x}\right)<0$ when
$\frac{n}{a+bx_{(1)}}-n\bar{x}-k_1<0$.  Then, for a given $b>0$, and $\theta<0$, the root of ${\rm g}_1\left(a;b,\theta,{\boldsymbol x}\right)=0$ lies in the following interval:
\[\left(\frac{1}{\bar{x}}-bx_{(n)},(\bar{x}+\frac{k_1}{n})^{-1}-bx_{(1)}\right).\]

\subsection*{B.2}

 Let
\[{\rm g}_3\left(\theta;a,b,{\boldsymbol x}\right)=\frac{n}{\theta}+\sum^n_{i=1}{\frac{p_iC''\left(\theta p_i\right)}{C'\left(\theta p_i\right)}}-\frac{nC'(\theta )}{C(\theta )}.\]

\noindent i. For LFRP, it is clear that
\[\mathop{\lim}_{\theta \rightarrow \infty} {\rm g}_3\left(\theta ;a,b,{\boldsymbol x}\right)=\sum^n_{i=1}{p_i-n}<0,\ \ \ \ \ \ \ \ \ \ \ \ \ \ \ \mathop{\lim }_{\theta \rightarrow 0^+} {\rm g}_3\left(\theta ;a,b,{\boldsymbol x}\right)= \sum^n_{i=1}{p_i-\frac{n}{2}}.\]
Therefore, the equation ${\rm g}_3\left(\theta;a,b,{\boldsymbol x}\right) =0$ has at least one root for $\theta>0$, if $\sum^n_{i=1}p_i-\frac{n}{2}>0$ or $\sum^n_{i=1}p_i >n/2$.

\noindent ii. For LFRG, it is clear that
\[\mathop{\lim}_{\theta \rightarrow1^{-}}{\rm g}_ 3\left(\theta ;a,b,{\boldsymbol x}\right)=-\infty,\ \ \ \ \ \ \ \ \ \ \ \ \ \ \ \mathop{\lim }_{\theta \rightarrow0^+} {\rm g}_3\left(\theta ;a,b,{\boldsymbol x}\right)=-n+2\sum^n_{i=1}p_i.\]
Therefore, the equation ${{\rm g}}_3\left(\theta;a,b,{\boldsymbol x}\right){\rm =0}$ has at least one root for $0<\theta<1$, if $-n+2\sum^n_{i=1}p_i>0$ or $\sum^n_{i=1}p_i>\frac{n}{2}$.

\noindent iii. For LFRL, it is clear that
\[\mathop{\lim}_{\theta \rightarrow0^+}{\rm g}_3\left(\theta ;a,b,{\boldsymbol x}\right)=\sum^n_{i=1}p_i-\frac{n}{2},\ \ \ \ \ \ \ \ \ \ \ \ \ \ \ \mathop{\lim}_{\theta \rightarrow 1^-}{\rm g}_3\left(\theta;a,b,{\boldsymbol x}\right)=-\infty.\]
Therefore, the equation ${{\rm g}}_3\left(\theta;a,b,{\boldsymbol x}\right)=0$ has at least one root for $0<\theta<1$, if $\sum^n_{i=1}p_i-\frac{n}{2}>0$ or $\sum^n_{i=1}p_i>\frac{n}{2}$.

\noindent iv.It is clear that
\[\mathop{\lim}_{p\rightarrow0^+}{\rm g}_3\left(\theta;a,b,{\boldsymbol x}\right)=\sum^n_{i=1}p_i\left(m-1\right)-\frac{n\left(m-1\right)}{2},\ \ \ \ \ \mathop{\lim}_{p\rightarrow 1^{-}}{\rm g}_3\left(\theta ;a,b,{\boldsymbol x}\right)=\sum^n_{i=1}\frac{-m+1+mp_i}{p_i}.\]
Therefore, the equation ${{\rm \ g}}_{{\rm 3}}\left(\theta ;a,b,{\boldsymbol x}\right)=0$ has at least one root for $0<p<1$, if $\sum^n_{i=1}p_i\left(m-1\right)-\frac{n\left(m-1\right)}{2}>0$ and $\sum^n_{i=1}\frac{-m+1+mp_i}{p_i}<0$ or $\sum^n_{i=1}p_i>\frac{n}{ 2}$ and $\sum^n_{i=1}{p_i}^{-1}>\frac{nm}{1-m}$.

\section*{Appendix C}
\subsection*{C.1}
 We can easily show that $h_1\left(a\right)$ is decreasing function with respect to $a$. Also
\[{\mathop{\lim }_{a\rightarrow0^+} h_1\left(a\right)\ }=+\infty,\ \ \ \ \ \ \ \ \ \ \mathop{\lim }_{a\rightarrow+\infty} h_1\left(a\right)\ =-c_1<0.\]
Therefore, the root of $h_1\left(a\right)=0$ is unique. It can easily show that
\[h_1\left(a\right)<\frac{n}{a+{\hat{b}}^{(t)}x_{(1)}}
-c_1 \ \ \ \ \ \ h_1\left(a\right)>\frac{n}{a+{\hat{b}}^{(t)}x_{(n)}}-c_1.\]
Therefore, $h_1\left(a\right)<0$ when
$\frac{n}{a+{\hat{b}}^{(t)}x_{(1)}}-c_1<0$, and $h_1\left(a\right)>0$ when $\frac{n}{a+{\hat{b}}^{(t)}x_{(n)}}-c_1>0$.
Hence, the proof is completed.

\subsection*{C.2}
 We can easily show that $h_2\left(b\right)$ is decreasing function with respect to $b$. Also
\[{\mathop{\lim }_{b\rightarrow0^+} h_2\left(b\right)\ }=+\infty,\ \ \ \ \ \ \ \ \ \ \mathop{\lim }_{b\rightarrow+\infty} h_2\left(b\right)=-\frac{1}{2}\sum^n_{i=1}{{\hat{z}}^{\left(t\right)}_ix^2_i}<0.\]
Therefore, the root of $h_2\left(b\right)=0$ is unique. It can easily show that
\[h_2\left(b\right)>\frac{nx_{(1)}}{{\hat{a}}^{(t)}+bx_{(1)}}-
\frac{c_2}{2}, \ \ \ \ \ \ \ \ \
h_2\left(b\right)<\frac{nx_{(n)}}{{\hat{a}}^{(t)}+bx_{(n)}}-
\frac{c_2}{2}.\]
Therefore, $h_2\left(b\right)<0$ when
$\frac{nx_{(n)}}{{\hat{a}}^{(t)}+bx_{(n)}}-
\frac{c_2}{2}<0$, and $h_2\left(b\right) >0$ when
$\frac{nx_{(1)}}{{\hat{a}}^{(t)}+bx_{(1)}}-
\frac{c_2}{2}>0$.
Hence, the proof is completed.

\subsection*{C.3}

\noindent i. For LFRG distribution, $C\left(\theta \right)=\frac{\theta }{1-\theta }$. Therefore, the root of
\[h_3\left(\theta \right)=\theta -\theta (1-\theta )c_0=\theta \left(1-(1-\theta )c_0\right)=0,\]
is unique and is equal to $1-\frac{n}{c_0}$, and $0<1-\frac{n}{c_0}<1$

\[{\hat{z}}^{\left(t\right)}_i=1+\frac{{\widehat{\theta }}^{\left(t\right)}u^{\left(t\right)}_iC''\left({\widehat{\theta }}^{\left(t\right)}u^{\left(t\right)}_i\right)}{C'\left({\hat{\theta }}^{\left(t\right)}u^{\left(t\right)}_i\right)}=1+\frac{2{\widehat{\theta }}^{\left(t\right)}u^{\left(t\right)}_i}{\left(1-{\widehat{\theta }}^{\left(t\right)}u^{\left(t\right)}_i\right)}.\]
Let
\[f\left(x\right)=1+\frac{2x}{(1-x)}.\]
Therefore, $f(x)$ is increasing function and
$$
f(x)>1 \ \ \ \ 0<x<1,\ \ \ \  {\rm and} \ \ \ \ -1<f(x)<1 \ \ \ \ x<0.
$$
If $0<{\hat{\theta }}^{\left(t\right)}<1$, then $0<\hat{\theta}^{(t)}u^{(t)}_i<1$. Therefore,
$1<{\hat{z}}^{(t)}_i$, and $0<1-\frac{n}{c_0}<1$.

\bigskip

\noindent ii. For LFRP distribution, $C\left(\theta \right)=e^{\theta }-1$. Therefore,
\[h_3\left(\theta \right)=\theta -\frac{c_0(e^{\theta }-1)}{e^{\theta }}.\]
 Since
\[\frac{\partial }{\partial \theta }h_3\left(\theta \right)=1-c_0e^{-\theta},  \ \ \ \ \frac{\partial^2 }{\partial \theta^2 }h_3\left(\theta \right)=c_0e^{-\theta}>0,\]
where
$
c_0=1+ \frac{1}{n}\sum_{i=1}^n\hat{\theta}^{(t)}\hat{u}_i^{(t)}$ and $1<c_0<1+\hat{\theta}^{(t)},
$
and also,
\[{\mathop{\lim }_{\theta \to 0^+} h_3\left(\theta \right)\ }=0 ,\ \ \ \ \ \ \ \ \ \ \ \ {\mathop{\lim }_{\theta \to \infty } h_3\left(\theta \right)\ }=\infty ,\]
therefore, $h_3\left(\theta \right)$ has a minimum value at $\theta_0=\log(c_0)$, and $h_3\left(\theta_0 \right)<0$. Thus, the root of $h_3\left(\theta \right)=0$ is unique.

\bigskip

\noindent For LFRL distribution, $C\left(\theta \right)=-\log(1-\theta)$. Therefore,
\[h_3\left(\theta \right)=\theta -c_0(1-\theta)\log(1-\theta).\]
 Since
\[\frac{\partial }{\partial \theta }h_3\left(\theta \right)=1-c_0-c_0\log(1-\theta),  \ \ \ \ \frac{\partial^2 }{\partial \theta^2 }h_3\left(\theta \right)=\frac{c_0}{1-x}>0,\]
where
$
c_0= \sum_{i=1}^n\frac{1}{1-\hat{\theta}^{(t)}\hat{u}_i^{(t)}}$  and $1<c_0,$ and also,
\[{\mathop{\lim }_{\theta \to 0^+} h_3\left(\theta \right)\ }=0 ,\ \ \ \ \ \ \ \ \ \ \ \ {\mathop{\lim }_{\theta \to \infty } h_3\left(\theta \right)\ }=\infty ,\]
therefore, $h_3\left(\theta \right)$ has a minimum value at $\theta_0=1-\exp(\frac{1}{c_0}-1)$, and $h_3\left(\theta_0 \right)<0$. Thus, the root of $h_3\left(\theta \right)=0$ is unique.

\bigskip

\noindent For LFRB distribution, $C\left(\theta \right)=(\theta+1)^m-1$. Therefore,
\[h_3\left(\theta \right)=\theta -c_0\frac{(\theta+1)^m-1}{m(\theta+1)^{m-1}}.\]
 Since
\[\frac{\partial }{\partial \theta }h_3\left(\theta \right)=1-\frac{c_0}{m}-\frac{c_0(\theta+1)^{-m}}{m}(1-m),  \ \ \ \ \frac{\partial^2 }{\partial \theta^2 }h_3\left(\theta \right)=c_0(\theta+1)^{-m-1}(m-1)>0,\]
where
$
c_0= 1+\frac{m-1}{n}\sum_{i=1}^n\frac{\hat{\theta}^{(t)}\hat{u}^{(t)}}{\hat{\theta}^{(t)}\hat{u}^{(t)}+1}$  and $1<c_0<m,$ and also,
\[{\mathop{\lim }_{\theta \to 0^+} h_3\left(\theta \right)\ }=0 ,\ \ \ \ \ \ \ \ \ \ \ \ {\mathop{\lim }_{\theta \to \infty } h_3\left(\theta \right)\ }=\infty ,\]
therefore, $h_3\left(\theta \right)$ has a minimum value at $\theta_0=(\frac{c_0(1-m)}{m-c_0})^{\frac{1}{m}}-1$, and $h_3\left(\theta_0 \right)<0$. Thus, the root of $h_3\left(\theta \right)=0$ is unique.

\end{document}